\documentclass[11pt]{article}
\usepackage[a4paper,margin=24mm]{geometry}
\usepackage[utf8]{inputenc}
\usepackage[T1]{fontenc}
\usepackage[english]{babel}
\usepackage{
    amsmath,
    amssymb,
    mathtools,
    amsthm,
    bbm,
    amsbsy,
    stackrel,
    mathrsfs
}
\usepackage{dsfont}
\usepackage{subcaption}
\usepackage{enumerate}
\usepackage{hyperref}
\usepackage{color,xcolor}
\usepackage{csquotes}
\usepackage[labelfont=bf]{caption}
\usepackage{algpseudocode,algorithmicx,algorithm}
\usepackage{multirow,tabularx,booktabs}

\newtheorem{theorem}{Theorem}

\newtheorem{corollary}[theorem]{Corollary}

\newtheorem{lemma}[theorem]{Lemma}

\theoremstyle{definition}

\newtheorem{assumption}[theorem]{Assumption}

\theoremstyle{remark}
\newtheorem{remark}[theorem]{Remark}

\numberwithin{equation}{section}
\numberwithin{theorem}{section}

\DeclareMathOperator{\Pas}{\mathds{P}\text{-as}}
\DeclareMathOperator*{\argmin}{argmin}
\DeclareMathOperator*{\Argmin}{Argmin}
\DeclareMathOperator*{\supp}{supp}

\DeclareMathOperator{\VaR}{VaR}
\DeclareMathOperator{\ES}{ES}

\DeclareMathOperator{\oo}{o}
\DeclareMathOperator{\OO}{O}
\DeclareMathOperator{\Cost}{Cost}
\DeclareMathOperator{\minimize}{minimise}
\DeclareMathOperator{\subjectTo}{subject~to}

\DeclareMathOperator{\Var}{\mathbb{V}ar}

\hypersetup{
    colorlinks=true,
    urlcolor=black,
    linkcolor=black,
    citecolor=black
}

\title{Asymptotic Error Analysis of Multilevel Stochastic Approximations for the Value-at-Risk and Expected Shortfall\
}
\author{%
St{\'e}phane~Cr{\'e}pey\footnote{
Universit\'e Paris Cit\'e, Laboratoire de Probabilit\'es, Statistique et Mod\'elisation (LPSM), CNRS UMR 8001.
The research of S. Cr\'epey has benefited from the support of the Chair \textit{Capital Markets Tomorrow: Modeling and Computational Issues} under the aegis of the Institut Europlace de Finance, a joint initiative of Laboratoire de Probabilit\'es, Statistique et Mod\'elisation (LPSM) / Universit\'e Paris Cit\'e and Cr\'edit Agricole CIB.
\texttt{stephane.crepey@lpsm.paris}
}
\and
Noufel~Frikha\footnote{
Universit\'e Paris 1 Panth\'eon-Sorbonne, Centre d'Economie de la Sorbonne (CES), 106 Boulevard de l’H\^opital, 75642 Paris Cedex 13.
The research of N. Frikha has benefited from the support of the Institut Europlace de Finance.
\texttt{noufel.frikha@univ-paris1.fr}
}
\and
Azar~Louzi\footnote{Universit\'e Paris Cit\'e, CNRS, Laboratoire de Probabilit\'es, Statistique et Mod\'elisation. \texttt{azar.louzi@lpsm.paris}}
\and
Gilles~Pag{\`e}s\footnote{
Sorbonne Universit\'e, Laboratoire de Probabilit\'es, Statistique et Mod\'elisation (LPSM), CNRS UMR 8001, case 158, 4 place Jussieu, F-75252 Paris Cedex 5. \texttt{gilles.pages@sorbonne-universite.fr}
}
}
\date{December 23, 2024}

\begin{document}

\maketitle

\begin{center}
\begin{minipage}{.8\linewidth}
\small
{\bf Abstract.}
Cr\'epey, Frikha, and Louzi (2025) introduced a nested stochastic approximation algorithm and its multilevel acceleration to compute the value-at-risk and expected shortfall of a random financial loss. We hereby establish central limit theorems for the renormalized estimation errors associated with both algorithms as well as their averaged versions. Our findings are substantiated through a numerical example.
\medskip



\noindent
{\bf DOI.}
10.1214/24-EJP1246
\quad
\noindent
{\bf arXiv.}
2311.15333
\quad
\noindent
{\bf HAL.}
04304985

\end{minipage}
\end{center}

\section*{Introduction}

Stochastic approximation (SA) methods aim to find a root of a given function of which only noisy observations are available~\cite{RM51}. When the noisy observations are biased, with a bias that is reducible at an additional computational cost, nested SA (NSA) algorithms naively inject them into the considered SA scheme but score a complexity that increases significantly with lower biases~\cite{10.1214/15-AAP1109}. Multilevel stochastic approximation (MLSA) schemes accelerate such approaches by first producing a highly biased estimator at low cost, then correcting it incrementally with paired estimators of lower and lower biases~\cite{10.1214/15-AAP1109}.
They often score a significantly lower complexity than their nested counterparts when employed to attain some prescribed accuracy~\cite{10.1214/15-AAP1109,doi:10.1137/18M1173186,Dereich2019}.

The value-at-risk (VaR) and expected shortfall (ES) currently stand as the most widely used risk metrics in finance~\cite{BISFRTB} and in insurance~\cite{Solvency2,SolvencyII}.
An MLSA scheme was introduced in~\cite{mlsa} to estimate the VaR and ES of a random loss that is written as a conditional expectation, but not necessarily in a closed form.
The multilevel feature therein deals with the loss's needed amounts of inner Monte Carlo samplings, whereby biased noises emanate. Theoretical $L^2(\mathbb{P})$-error and complexity analyses, along with numerical case studies, have demonstrated the overperformance of the MLSA algorithm in comparison with its nested counterpart.

However, \cite{mlsa} does not delve into analyzing the asymptotic estimation errors, which are paramount to derive trust regions and confidence intervals for the VaR and ES.
Although a generic central limit theorem (CLT) for MLSA is already available in~\cite[Theorem~2.11]{10.1214/15-AAP1109}, our case is more delicate as it involves an objective function that is not strongly convex. Besides, unlike~\cite[Theorem~3.2 \&~3.3]{barrera:hal-01710394} (that rely on~\cite[Theorems~2.1 \&~3.2]{fort2015central}), our MLSA method employs a two-time-scale scheme for the VaR and ES.
Indeed, in line with~\cite{ebc8dfc8-2324-3c9a-9ce8-b1bdaf4c993d,e3e3cb93-9ed3-368b-85cf-e5a71e8cb33a}, it uses a slower learning rate $\{\gamma_n,n\geq1\}$ for the VaR, more accurate and independent of the ES, and a faster learning rate for the latter. But due to the nested nature of our framework, the results of the two aforementioned papers do not directly apply to our setting, rendering a thorough analysis necessary to derive CLTs for our NSA and MLSA estimators.
Furthermore, the optimal convergence rates of our algorithms are attained by selecting $\{\gamma_n=\gamma_1 n^{-1},n\geq1\}$, $\gamma_1>0$, as the VaR learning rate, with a non trivial constraint on $\gamma_1$~\cite{mlsa}.
This classical constraint is classicallly circumvented by the Polyak-Ruppert averaging principle~\cite{ruppert1991handbook,doi:10.1137/0330046}, hence motivating further asymptotic analyses of averaged versions of our algorithms.
\\

Throughout, we consider a probability space $(\Omega,\mathcal{F},\mathbb{P})$ that is rich enough to support all the random variables defined hereafter.
The VaR $\xi^0_\star$ and ES $\chi^0_\star$ of a random loss represented by an $L^1(\mathbb{P})$ real-valued random variable $X_0$, at some confidence level $\alpha\in(0,1)$, are defined by (\cite{doi:10.1002/9780470061602.eqf15003,ACERBI20021487})
\begin{align}
\xi^0_\star
&=\VaR_\alpha(X_0)
\coloneqq\inf{\big\{\xi\in\mathbb{R}\colon\mathbb{P}(X_0\leq\xi)\geq\alpha\big\}},\nonumber\\
\chi^0_\star
&=\ES_\alpha(X_0)
\coloneqq\frac1{1-\alpha}\int_\alpha^1\VaR_a(X_0)\mathrm{d}a
=\mathbb{E}[X_0\,|\,X_0\geq\xi^0_\star]
=\frac1{1-\alpha}\mathbb{E}[X_0\mathds{1}_{X_0\geq\xi^0_\star}].\label{eq:ES}
\end{align}
In some literature, these two quantities are referred to as the quantile and super-quantile
~\cite{10.1214/21-EJP648,10.1214/21-EJS1908}.
They can famously be retrieved as solutions to the problem (\cite[Theorem~1]{10.21314/JOR.2000.038}, \cite[Section~4.3]{10.1111/j.1467-9965.2007.00311.x} and~\cite[Proposition~2.1]{10.1007/978-3-642-04107-5_11})
\begin{equation}
\label{eq:sa:opt}
\min_{\xi\in\mathbb{R}}{V_0(\xi)},
\quad\mbox{ where }\quad
V_0(\xi)\coloneqq\xi+\frac1{1-\alpha}\mathbb{E}[(X_0-\xi)^+].
\end{equation}
More precisely, if the cdf $F_{X_0}$ of $X_0$ is continuous, then $V_0$ is convex and continuously differentiable on $\mathbb{R}$, with $V_0'(\xi)=(1-\alpha)^{-1}(F_{X_0}(\xi)-\alpha)$, $\xi\in\mathbb{R}$.
If $F_{X_0}$ is additionally increasing, then
\begin{equation*}
\xi^0_\star=\argmin{V_0}
\quad\mbox{ and }\quad
\chi^0_\star=\min{V_0}.
\end{equation*}
If $X_0$ admits a continuous pdf $f_{X_0}$, then $V_0$ is twice continuously differentiable on $\mathbb{R}$ with
$V_0''(\xi)=(1-\alpha)^{-1}f_{X_0}(\xi)$, $\xi\in\mathbb{R}$.

Henceforth, we assume
\begin{equation}
\label{def:X0:cond:expect}
X_0=\mathbb{E}[\varphi(Y,Z)|Y],
\end{equation}
where $Y\in\mathbb{R}^d$ and $Z\in\mathbb{R}^q$ are independent random variables and $\varphi\colon\mathbb{R}^d\times\mathbb{R}^q\to\mathbb{R}$ is a measurable function such that $\varphi(Y,Z)\in L^1(\mathbb{P})$. In financial applications, $\varphi(Y,Z)$ models the future cash flows of a portfolio, $Y$ the underlying risk factors up to a certain time horizon, $Z$ the underlying risk factors beyond that time horizon, and $X_0$ the loss of the portfolio at that time horizon.

Under the approach initiated by~\cite{BardouFrikhaPages+2009+173+210} (see also~\cite{10.1007/978-3-642-04107-5_11,bardou2016cvar}, as well as~\cite{doi:10.1137/120903142} for shortfall risk measures), assuming that exact i.i.d.~samples of $X_0$ are available, one may compute the pair $(\xi^0_\star,\chi^0_\star)$ using a two-time-scale SA algorithm.
However, we do not assume neither the distribution of $\varphi(Y,Z)$ conditionally on $Y$ to be known, rendering $X_0$ incomputable in closed form, nor the availability of i.i.d.~simulations of $X_0$.
We rather consider an ordinary setting where the portfolio's risk factors $Y$ and $Z$ emanate from a model, hence are simulatable, and $\varphi$ describing the corresponding future cash flows is computable.
See Section~\ref{sec:swap} for a concrete example.
\\

The paper is structured as follows.
Sections~\ref{sec:nested}--\ref{sec:mlsa:avg} introduce the NSA, averaged NSA (ANSA), MLSA and averaged MLSA (AMLSA) algorithms, establish CLTs for their renormalized VaR and ES estimation errors (in Corollaries~\ref{crl:nested:clt} and~\ref{crl:avg:nested:clt} and Theorems~\ref{thm:ml:clt} and~\ref{thm:avg:ml:clt}) and discuss their complexities and related parametrizations (in Thoerems~\ref{cost:nsa}, \ref{cost:ansa}, \ref{thm:cost:ml} and~\ref{thm:cost:aml}).
Section~\ref{sec:swap} presents a numerical case study supporting the paper's theoretical findings.

\section{Nested Stochastic Approximation Algorithm}
\label{sec:nested}

The brute force, naively nested SA algorithm introduced in~\cite{mlsa} for solving \eqref{eq:sa:opt}--\eqref{def:X0:cond:expect} is based on defining the approximate problem
\begin{equation}
\label{eq:sa:nested:opt}
\min_{\xi\in\mathbb{R}}{V_h(\xi)},
\quad\mbox{ where }\quad
V_h(\xi)\coloneqq\xi+\frac1{1-\alpha}\mathbb{E}[(X_h-\xi)^+]
\end{equation}
and $X_h$ is the empirical mean approximation of $X_0$ given by
\begin{equation*}
X_h\coloneqq\frac1K\sum_{k=1}^K\varphi(Y,Z^{(k)}),
\quad
h=\frac1K\in
\mathcal{H}\coloneqq\Big\{\frac1{K'},K'\in\mathbb{N}_0\Big\},
\end{equation*}
and $\{Z^{(k)},1\leq k\leq K\}\stackrel{\text{\tiny\rm i.i.d.}}{\sim}Z$ are independent of $Y$.
$X_h$, $h\in\mathcal{H}$, represents a biased (Monte Carlo) estimator of $X_0$ that, in principle, converges $\Pas$ to $X_0$ under suitable assumptions via the conditional law of large numbers. Note that the notion of bias here is linked to the ensuing solution $(\xi^h_\star,\chi^h_\star)$ to \eqref{eq:sa:nested:opt}. Indeed, for $h\in\mathcal{H}$, if the cdf $F_{X_h}$ of $X_h$ is continuous and increasing, then
\begin{equation}
\label{eq:sa:nested:sol}
\xi^h_\star =\argmin{V_h}
\quad\mbox{ and }\quad
\chi^h_\star=\min{V_h}
\end{equation}
provide biased estimators of $\xi^0_\star$ and $\chi^0_\star$.
The bias of $(\xi^h_\star,\chi^h_\star)$ with respect to $(\xi^0_\star,\chi^0_\star)$ is controlled by the \emph{bias parameter} $h$: ideally, the smaller the $h$, the smaller the approximation bias, and vice versa.

Let $h\in\mathcal{H}$.
If the cdf $F_{X_h}$ of $X_h$ is continuous, then $V_h$ is continuously differentiable on $\mathbb{R}$, with
\begin{equation}
V_h'(\xi)=\frac1{1-\alpha}(F_{X_h}(\xi)-\alpha)=\mathbb{E}[H_1(\xi,X_h)],
\quad
H_1(\xi,x)=1-\frac1{1-\alpha}\mathds{1}_{x \geq\xi},
\quad\xi,x\in\mathbb{R}.
\label{eq:H1}
\end{equation}
If $X_h$ admits a continuous pdf $f_{X_h}$, then $V_h$ is twice continuously differentiable on $\mathbb{R}$ with
$V_h''(\xi)=(1-\alpha)^{-1}f_{X_h}(\xi)$, $\xi\in\mathbb{R}$.

In view of \eqref{eq:H1}, if $F_{X_h}$ is continuous and increasing, then $V_h$ is strictly convex so that $\xi^h_\star$ is unique and $\left\{\xi^{h}_\star\right\} = \left\{ V'_{h}=0\right\}$. This hints at a stochastic gradient descent type algorithm based on \eqref{eq:H1} to approximate $\xi^h_\star$. Assuming that such an algorithm is driven by innovations $\{X_h^{(n)},n\geq1\}\stackrel{\text{\tiny\rm i.i.d.}}{\sim}X_h$ and produces iterates $\{\xi^h_n,n\geq0\}$ converging to $\xi^h_\star$ $\Pas$, then \eqref{eq:sa:nested:opt}, \eqref{eq:sa:nested:sol}, the law of large numbers together with Ces\`aro's lemma yield under some suitable assumptions:
\begin{equation*}
\chi^h_\star
=\mathbb{E}\Big[\xi^h_\star+\frac1{1-\alpha}(X_h-\xi^h_\star)^+\Big]
=\lim_{n\to\infty}
\frac1n\sum_{k=1}^n\Big(\xi^h_{k-1}+\frac1{1-\alpha}(X_h^{(k)}-\xi^h_{k-1})^+\Big).
\end{equation*}
Thus $\chi^h_\star$ can be approximated by the sample mean on the right hand side above.
Define
\begin{equation*}
H_2(\xi,\chi,x)=\chi-\Big(\xi+\frac1{1-\alpha}(x-\xi)^+\Big),
\quad\xi,\chi,x\in\mathbb{R}.
\end{equation*}
Following the previous discussion, we derive the two-time-scale NSA dynamics
\begin{align}
\xi^h_{n+1}&=\xi^h_n-\gamma_{n+1}H_1(\xi^h_n,X_h^{(n+1)}),
\label{eq:sa:nested:alg:xi}\\
\chi^h_{n+1}&=\chi^h_n-\frac1{n+1}H_2(\xi^h_n,\chi^h_n,X_h^{(n+1)}),
\label{eq:sa:nested:alg:C}
\end{align}
$n\in\mathbb{N}$, driven by the innovations $\{X_h^{(n)},n\geq1\}\stackrel{\text{\tiny\rm i.i.d.}}{\sim}X_h$, starting from real-valued random variables $\xi^h_0$ and $\chi^h_0=0$, $\xi^h_0$ being independent from $\{X_h^{(n)},n\geq1\}$, and advancing at the positive learning rate $\{\gamma_n,n\geq1\}$ such that $\sum_{n=1}^\infty\gamma_n=\infty$ and $\sum_{n=1}^\infty\gamma_n^2<\infty$.

\subsection{Convergence Rate Analysis}

We postulate below some essential assumptions on $\{X_h,h\in\mathcal{H}\cup\{0\}\}$.
\begin{assumption}[{\cite[Assumptions~2.2 \&~2.5]{mlsa}}]\label{asp:misc}\
\begin{enumerate}[(i)]
    \item\label{asp:misc:i}
For any $h\in \mathcal{H}$, $F_{X_h}$ admits the first order Taylor expansion
\begin{equation*}
F_{X_h}(\xi)-F_{X_0}(\xi)=v(\xi)h+w(\xi,h)h,
\quad\xi\in\mathbb{R},
\end{equation*}
for some functions $v,w(\,\cdot\,,h)\colon\mathbb{R}\to\mathbb{R}$ satisfying, for any $\xi^0_\star\in\Argmin{V_0}$,
\begin{equation*}
\int^\infty_{\xi^0_\star}v(\xi)\mathrm{d}\xi<\infty
,\qquad
\lim_{\mathcal{H}\ni h\downarrow0}w(\xi^0_\star,h)=\lim_{\mathcal{H}\ni h\downarrow0}\int^\infty_{\xi^0_\star}w(\xi,h)\mathrm{d}\xi=0.
\end{equation*}

    \item\label{asp:misc:ii}
For any $h\in\mathcal{H}\cup\{0\}$, the law of $X_h$ admits a continuous pdf $f_{X_h}$ with respect to the Lebesgue measure. Moreover, the pdf $\{f_{X_h},h\in\mathcal{H}\}$ converge locally uniformly to $f_{X_0}$.

    \item\label{asp:misc:iii}
For any $R>0$,
\begin{equation*}
\inf_{\substack{h\in\mathcal{H}\cup\{0\}\\\xi\in B(\xi^0_\star,R)}}{f_{X_h}(\xi)}>0.
\end{equation*}

    \item\label{asp:misc:iv}
The pdf $\{f_{X_h},h\in\mathcal{H}\cup\{0\}\}$ are uniformly bounded and uniformly Lipschitz, namely,
\begin{equation*}
\sup_{h\in\mathcal{H}\cup\{0\}}{(\|f_{X_h}\|_\infty+[f_{X_h}]_\mathrm{Lip})}<\infty,
\end{equation*}
where $[f_{X_h}]_\text{\rm Lip}$ denotes the Lipschitz constant of $f_{X_h}$, $h\in\mathcal{H}\cup\{0\}$.
\end{enumerate}
\end{assumption}

\begin{remark}\label{rmk:nested:misc}
\begin{enumerate}[\it (i)]
\item (\cite[Remarks~2.3 \&~2.6]{mlsa}).\;
Assumptions~\ref{asp:misc}(\ref{asp:misc:i}) and~\ref{asp:misc}(\ref{asp:misc:ii}) are guaranteed under the conditions of~\cite[Propositions~5.1(a) \&~5.1(b)]{Giorgi2020}.
The last part of Assumption~\ref{asp:misc}(\ref{asp:misc:i}) reads $w(\xi^0_\star,h)=\int_{\xi^0_\star}^\infty w(\xi,h)\mathrm{d}\xi=\oo(1)$ as $\mathcal{H}\ni h\downarrow0$, $\xi^0_\star\in\Argmin{V_0}$.

Assumption~\ref{asp:misc}(\ref{asp:misc:iii}) is natural in view of Assumption~\ref{asp:misc}(\ref{asp:misc:ii}) and an increasing $F_{X_0}$.
It may be weakened to $f_{X_0}(\xi^0_\star)>0$, but the current formulation eases the obtention of the Lyapunov properties~\cite[Lemma~D.1]{mlsa} that play a central role in controlling the statistical error~\cite[Theorem~2.7]{mlsa} (Lemma~\ref{lmm:error}(\ref{lmm:error:statistical}), especially the constant defined in \eqref{e:thelambda} whose value must be positive).

Assumption~\ref{asp:misc}(\ref{asp:misc:iv}) is natural in view of Assumption~\ref{asp:misc}(\ref{asp:misc:ii}).

\item
\label{rmk:V''(xi*)>0}
Assumptions~\ref{asp:misc}(\ref{asp:misc:ii}) and~\ref{asp:misc}(\ref{asp:misc:iii}) imply that, for any $h\in\mathcal{H}\cup\{0\}$ and any $\xi^h_\star\in\Argmin{V_h}$, $f_{X_h}(\xi^h_\star)>0$, hence $V_h''(\xi^h_\star)>0$ and $\Argmin{V_h}$ is reduced to a singleton $\{\xi^h_\star\}$.
\end{enumerate}
\end{remark}

The global error of the NSA scheme \eqref{eq:sa:nested:alg:xi}-\eqref{eq:sa:nested:alg:C} writes as a sum of statistical and bias errors:
\begin{align*}
\xi^h_n-\xi^0_\star
&=\big(\xi^h_n-\xi^h_\star\big)+(\xi^h_\star-\xi^0_\star),
\\
\chi^h_n-\chi^0_\star
&=\big(\chi^h_n-\chi^h_\star\big)+(\chi^h_\star-\chi^0_\star).
\end{align*}

\begin{assumption}\label{asp:supE[sup]}
Denoting $c_\alpha=1\vee\frac\alpha{1-\alpha}$, it holds that
\begin{equation*}
\sup_{h\in\mathcal{H}}{\mathbb{E}\Big[|\xi^h_0|^2\exp\Big(\frac4{1-\alpha}c_\alpha \sup_{h\in\mathcal{H}}\|f_{X_h}\|_\infty|\xi_0^h|\Big)\Big]}<\infty.
\end{equation*}
\end{assumption}
Set
\begin{equation}\label{e:thelambda}
\lambda=\inf_{h\in\mathcal{H}\cup\{0\}}\frac38V_h''(\xi^h_\star)\wedge \|V_h''\|_\infty\frac{V_h''(\xi^h_\star)^4}{4[V_h'']_\mathrm{Lip}^2}>0.
\end{equation}
We refer to Lemma~\ref{lmm:error} for bias and statistical error controls for the NSA scheme.
Our main CLT for the NSA scheme follows.
\begin{theorem}
\label{thm:nested:clt}
Suppose that Assumptions~\ref{asp:misc} and~\ref{asp:supE[sup]} hold, and that $\mathbb{E}[|\varphi(Y, Z)|^{2+\delta}]<\infty$ for some $\delta>0$.
If $\gamma_n=\gamma_1n^{-\beta}$, $n\geq1$, $\beta\in\big(\frac12,1\big]$, with $\lambda\gamma_1>1$ if $\beta=1$, then
\begin{equation*}
\begin{pmatrix}
h^{-\beta}\big(\xi^h_{\lceil h^{-2}\rceil}-\xi^h_\star\big)\\
h^{-1}\big(\chi^h_{\lceil h^{-2}\rceil}-\chi^h_\star\big)
\end{pmatrix}
\stackrel[]{\mathcal{L}}{\longrightarrow}\mathcal{N}(0,\Sigma_\beta)
\quad\mbox{ as }\quad\mathcal{H}\ni h\downarrow0,
\end{equation*}
where
\begin{equation}
\label{eq:Sigma}
\Sigma_\beta=\begin{pmatrix}
\frac{\alpha\gamma_1}{2f_{X_0}(\xi^0_\star)-\gamma_1^{-1}(1-\alpha)\mathds{1}_{\beta=1}}
&\alpha\frac{\chi^0_\star-\xi^0_\star}{f_{X_0}(\xi^0_\star)}\mathds{1}_{\beta=1}\\
\alpha\frac{\chi^0_\star-\xi^0_\star}{f_{X_0}(\xi^0_\star)}\mathds{1}_{\beta=1}
&\frac{\Var((X_0-\xi^0_\star)^+)}{(1-\alpha)^2}
\end{pmatrix}.
\end{equation}
\end{theorem}

\begin{proof}
See Appendix~\ref{prf:nested:clt}.
\end{proof}

\begin{corollary}
\label{crl:nested:clt}
Within the framework of Theorem~\ref{thm:nested:clt},
\begin{equation*}
\begin{pmatrix}
h^{-\beta}\big(\xi^h_{\lceil h^{-2}\rceil}-\xi^0_\star\big)\\
h^{-1}\big(\chi^h_{\lceil h^{-2}\rceil}-\chi^0_\star\big)
\end{pmatrix}
\stackrel[]{\mathcal{L}}{\longrightarrow}\mathcal{N}\left(\begin{pmatrix}
-\frac{v(\xi^0_\star)}{f_{X_0}(\xi\star)}\mathds{1}_{\beta=1}\\
-\int_{\xi^0_\star}^\infty\frac{v(\xi)}{1-\alpha}\mathrm{d}\xi
\end{pmatrix},
\Sigma_\beta\right)
\quad\mbox{ as }\quad\mathcal{H}\ni h\downarrow0,
\end{equation*}
where $\Sigma_\beta$ is given by \eqref{eq:Sigma}.
\end{corollary}

\begin{proof}
For $h\in\mathcal{H}$,
\begin{equation*}
\begin{pmatrix}
h^{-\beta}\big(\xi^h_{\lceil h^{-2}\rceil}-\xi^0_\star\big)\\
h^{-1}\big(\chi^h_{\lceil h^{-2}\rceil}-\chi^0_\star\big)
\end{pmatrix}
=\begin{pmatrix}
h^{-\beta}\big(\xi^h_{\lceil h^{-2}\rceil}-\xi^h_\star\big)\\
h^{-1}\big(\chi^h_{\lceil h^{-2}\rceil}-\chi^h_\star\big)
\end{pmatrix}
+\begin{pmatrix}
h^{-\beta}(\xi^h_\star-\xi^0_\star)\\
h^{-1}(\chi^h_\star-\chi^0_\star)
\end{pmatrix}.
\end{equation*}
By Lemma~\ref{lmm:error}(\ref{lmm:error:weak}),
\begin{equation*}
\begin{pmatrix}
h^{-\beta}(\xi^h_\star-\xi^0_\star)\\
h^{-1}(\chi^h_\star-\chi^0_\star)
\end{pmatrix}
\to
-\begin{pmatrix}
\frac{v(\xi^0_\star)}{f_{X_0}(\xi\star)}\mathds{1}_{\beta=1}\\
\int_{\xi^0_\star}^\infty\frac{v(\xi)}{1-\alpha}\mathrm{d}\xi
\end{pmatrix}
\quad\mbox{ as }\quad\mathcal{H}\ni h\downarrow0.
\end{equation*}
Combining this limit with Theorem~\ref{thm:nested:clt} yields the desired result.
\end{proof}

Some comments are in order.

\begin{remark}\label{rmk:nsa:cv}
\begin{enumerate}[(i)]
\item\label{rmk:beta}
Unlike the non-asymptotic controls derived in~\cite[Theorem~2.7]{mlsa}, which are valid for all $\beta\in(0,1]$, the CLTs above are only valid for $\beta\in\big(\frac12,1\big]$.

As in~\cite{mlsa}, the bias parameter $h\in\mathcal{H}$ is used to control the bias error, the iterations amount and the convergence speed of NSA altogether.

Note that in Corollary~\ref{crl:nested:clt}, the bias parameter $h\in\mathcal{H}$ tends to $0$ while simultaneously the SA iterations amount $n=\lceil h^{-2}\rceil$ tends to $\infty$.

While the ES converges in $\OO(h)$ for any $\beta\in\big(\frac12,1\big]$, the VaR converges in $\OO(h^\beta)$, only reaching optimality at $\beta=1$. This comes at the cost of an additional non trivial constraint on $\gamma_1$ stipulating that $\lambda\gamma_1>1$, the value of $\lambda$ being inaccessible in view of \eqref{e:thelambda}.
Such restriction is classical in the stochastic approximation literature~\cite{D96}.

\item\label{rmk:nsa:cv:ii}
The VaR error's variance factor in \eqref{eq:Sigma} is affected by the choice of $\gamma_1$. 

When $\beta\in\big(\frac12,1\big)$, this factor is linear in $\gamma_1$, suggesting to take $\gamma_1$ small. But, by \eqref{eq:sa:nested:alg:xi}, $\lim_{\gamma_1\downarrow0}\xi^h_n=\xi^h_0$, i.e.~the algorithm does not learn for too small a value of $\gamma_1$. A fine-tuning of $\gamma_1$ must therefore be carried out.

The optimal setting $\beta=1$ affects this factor differently.
While it evolves in $\OO(\gamma_1)$ for $\gamma_1\uparrow\infty$, it evolves in $\OO\big(\big(\gamma_1-\frac{1-\alpha}{2f_{X_0}(\xi^0_\star)}\big)^{-1}\big)$ when $\gamma_1\downarrow\frac{1-\alpha}{2f_{X_0}(\xi^0_\star)}$, reflecting a high numerical instability.
Taking $\gamma_1=\frac{1-\alpha}{f_{X_0}(\xi^0_\star)}$ to set this factor to its minimum $\frac{\alpha(1-\alpha)}{f_{X_0}(\xi^0_\star)^2}$ is infeasible as neither $f_{X_0}$ nor $\xi^0_\star$ are a priori accessible.
Hence an upstream fine-tuning of $\gamma_1$ is required.

\item
The ES error's variance factor is independent of $\gamma_1$, reflecting a stronger numerical stability for the ES estimation.
This is to be nuanced, as ES dependency on $\gamma_1$ is dissimulated in secondary error terms in \eqref{eq:chih-chi*} that vanish faster than the main martingale component giving the CLT~\ref{crl:nested:clt}.

The choice $\beta=1$ affects the quality of the ES estimation, as the off-diagonal terms correlate positively the numerical instability of the VaR with the ES.

\item\label{rmk:nsa:cv:iv}
Although analytic, the covariance matrix \eqref{eq:Sigma} is not computable, as it depends on $f_{X_0}$, $\xi^0_\star$ and $\chi^0_\star$ that are a priori unknown. It must therefore be statistically estimated using several samplings of the renormalized estimation errors, in order to get confidence regions.
Alternatively, as in~\cite[Remark~3.4]{10.1214/21-EJP648}, the ES variance factor can be estimated using the Monte Carlo approximation
\begin{equation}
\label{eq:Var(X0-xi0)}
\Var\big((X_0-\xi^0_\star)^+\big)\approx\frac1{\lceil h^{-2}\rceil}\sum_{k=1}^{\lceil h^{-2}\rceil}\big((X_h^{(k)}-\xi^h_{k-1})^+\big)^2-\Big(\frac1{\lceil h^{-2}\rceil}\sum_{k=1}^{\lceil h^{-2}\rceil}(X_h^{(k)}-\xi^h_{k-1})^+\Big)^2.
\end{equation}
Besides, $f_{X_h}$ could be kernel-fitted by a $\widehat{f}_{X_h}$ so that e.g.~$f_{X_0}(\xi^0_\star)\approx\widehat{f}_{X_h}(\xi^h_{\lceil h^{-2}\rceil})$.
\end{enumerate}
\end{remark}

\begin{remark}\label{rmk:1.6:v-ix}
\begin{enumerate}[(i)]
\item\label{rmk:unbiased:1}
\cite[Theorem~1]{e3e3cb93-9ed3-368b-85cf-e5a71e8cb33a} provides a two-time-scale CLT for the unbiased case where exact simulations of $X_0$ are assumed to exist. Note that plugging a biased estimator $X_h$ therein does not achieve the sought concurrent control on the bias and iterations amount, as elaborated in Remark~\ref{rmk:nsa:cv}(\ref{rmk:beta}).

Besides, our method avoids the requirement of~\cite[Assumption~A4(iii)]{e3e3cb93-9ed3-368b-85cf-e5a71e8cb33a} that the martingale in Step~7 of Theorem~\ref{thm:nested:clt}'s proof, Appendix~\ref{prf:nested:clt}, which controls Corollary~\ref{crl:avg:nested:clt}, be $\frac2\beta$-integrable.

We also refer to~\cite[Theorem~4.1]{ebc8dfc8-2324-3c9a-9ce8-b1bdaf4c993d} on a two-time-scale unbiased CLT, with a linear dependency on the iterates for the update rules.

\item\label{rmk:unbiased:2}
\cite[Theorem~1.3]{10.1214/21-EJS1908} relies on~\cite[Theorem~1]{e3e3cb93-9ed3-368b-85cf-e5a71e8cb33a} to analyze a two-time-scale VaR-ES unbiased scheme.
Unlike \eqref{eq:sa:nested:alg:C}, their ES recursion writes
\begin{equation}\label{eq:ES:MC}
\chi^0_{n+1}=(1-b_{n+1})\chi^0_n+b_{n+1}X_0^{(n+1)}\mathds{1}_{X_0^{(n+1)}>\xi^0_n},
\end{equation}
echoing the ES definition \eqref{eq:ES}. The $\{\xi^0_k,k\geq1\}$ are supposingly iterates produced by \eqref{eq:sa:nested:alg:xi} by swapping $X_h$ with an exact simulation of $X_0$ and the $\{b_n,n\geq1\}$ are damping weights.

\item
The techniques employed in the aforementioned literature linearize the statistical error into a martingale and a drift terms. An $L^2(\mathbb{P})$-control guarantees that the drift vanishes faster than the martingale, which in turn yields the CLT via Lyapunov's theorem.

Due to our nested strategy, the linearization derived in Appendix~\ref{prf:nested:clt} entails two martingale and two drift terms. The first pair of drift and martingale terms deal with the effect of the bias parameter $h$ and the second one concerns the iterations amount $n=\lceil h^{-2}\rceil$.

Besides, we invoke the celebrated martingale arrays CLT~\cite[Corollary~3.1]{nla.cat-vn2887492}, a doubly indexed generalization of Lyapunov's theorem that nicely suits our multilevel analyses in subsequent sections.

\item
In a different direction, \cite[Theorem~2]{costa2024cvrpenalizedportfoliooptimization} proves the convergence of a stochastic mirror descent algorithm for portfolio optimization under an ES constraint with the presence of a time discretization bias. The biases are assumed to be adaptively controlled by bias parameters $\{h_n,n\geq1\}$ such that $\sum_{n=1}^\infty h_n<\infty$.

This is not the case with our scheme. Although the CLT varies both the bias and iterations amounts, running the algorithm fixes a single $h$ throughout all iterations. We refer to Remark~\ref{rmk:nested:cost}(\ref{rmk:2time}) for further comments in this regard.

\item
\cite[Theorem~2.7]{10.1214/15-AAP1109} presents a biased CLT similar to ours, however, it relies on the mean-reverting property~\cite[Assumption~(HMR)]{10.1214/15-AAP1109}, stipulating that the objective function being minimized is strongly convex. This is not the case with \eqref{eq:sa:nested:opt}.
Moreover, the aforementioned result concerns only the VaR component of our NSA scheme.

\end{enumerate}
\end{remark}

\subsection{Complexity Analysis}
Next, we discuss the complexity of NSA stemming from our CLTs.

\begin{theorem}
\label{cost:nsa}
Let $\varepsilon\in(0,1)$ be a prescribed accuracy. Within the framework of Theorem~\ref{thm:nested:clt}, setting
\begin{equation*}
h=\frac1{\big\lceil\varepsilon^{-\frac1\beta}\big\rceil}\sim\varepsilon^\frac1\beta
\quad\mbox{ and }\quad
n=\lceil h^{-2}\rceil=\lceil\varepsilon^{-\frac1\beta}\rceil^2\sim\varepsilon^{-\frac2\beta}
\end{equation*}
achieves a convergence rate in distribution for the NSA scheme \eqref{eq:sa:nested:alg:xi}-\eqref{eq:sa:nested:alg:C} of order $\varepsilon$ as $\varepsilon\downarrow0$.
The corresponding complexity satisfies
\begin{equation*}
\Cost^\beta_\text{\rm\tiny NSA}\leq
C\varepsilon^{-\frac3\beta},
\end{equation*}
for some positive constant $C<\infty$.
The optimal complexity, reached for $\beta=1$ under the constraint $\lambda\gamma_1>1$, satisfies
\begin{equation*}
\Cost^1_\text{\rm\tiny NSA}\leq C\varepsilon^{-3}.
\end{equation*}
\end{theorem}

\begin{proof}
The complexity of NSA satisfies
\begin{equation*}
\Cost^\beta_\text{\rm\tiny NSA}\leq C\frac{n}h,
\end{equation*}
for some positive constant $C<\infty$.
By Corollary~\ref{crl:nested:clt}, the convergence rate of NSA, with bias parameter $h$ and iterations amount $n=\lceil h^{-2}\rceil$, is of order $h^\beta$ as $\mathcal{H}\ni h\downarrow0$. The result follows by setting $h$ such that $h^\beta\leq\varepsilon$.
\end{proof}

\begin{remark}\label{rmk:nested:cost}
\begin{enumerate}[(i)]
\item
As in Corollary~\ref{crl:nested:clt}, NSA's optimal complexity is achieved when $\beta=1$, under the constraint $\lambda\gamma_1>1$.

\item
\label{rmk:2time}
As noted in~\cite[Remark~2.9]{mlsa}, unlike our NSA scheme, \cite[Algorithm~1]{barrera:hal-01710394} uses a single time scale for both the VaR and ES and, at the $n$-th iteration, uses a bias parameter $h_n$ such that $\mathcal{H}\ni h_n\downarrow0$ as $n\uparrow\infty$.

Although their approach results in an unbiased CLT~\cite[Theorem~3.2]{barrera:hal-01710394} (relying on~\cite[Theorem~2.1]{fort2015central}), it scores a complexity of order $\varepsilon^{-4}$.

As supported by Corollary~\ref{crl:nested:clt} and Theorem~\ref{cost:nsa}, our NSA scheme of single bias parameter $h$, not only achieves optimal convergence rate for the ES due to the two-time-scale strategy, but also runs in $\varepsilon^{-3}$ computational time.
\end{enumerate}
\end{remark}

\section{Averaged Nested Stochastic Approximation Algorithm}

According to Theorem~\ref{thm:nested:clt} and Corollary~\ref{crl:nested:clt}, the best CLT convergence rate is achieved for $\beta=1$, i.e.~by setting $\gamma_n=\gamma_1n^{-1}$ for the VaR NSA scheme.
This choice is only possible under the additional constraint $\lambda\gamma_1>1$, where $\lambda$ is an explicit yet inaccessible constant \eqref{e:thelambda}. In practice, $\gamma_1$ must be fine-tuned, adding a computational burden to the algorithm's usage.

To bypass this tuning issue, we explore the effect of the Polyak-Ruppert averaging principle~\cite{doi:10.1137/0330046,10.1214/15-AAP1109,ruppert1991handbook} on the VaR NSA scheme.
We follow in the footsteps of~\cite{10.1214/15-AAP1109,barrera:hal-01710394} and define, for a bias parameter $h\in\mathcal{H}$, the averaged VaR estimator
\begin{equation}
\label{eq:sa:avg:nested:alg:xi}
    \overline{\xi}^h_n=\frac1n\sum_{k=1}^n\xi^h_k =\Big(1-\frac1n\Big)\overline{\xi}^h_{n-1}+\frac1n\xi^h_n,
    \quad n\geq1,
\end{equation}
where $\overline{\xi}^h_0=0$ and $\{\xi^h_n,n\geq0\}$ are obtained via the NSA scheme \eqref{eq:sa:nested:alg:xi} with step sizes $\{\gamma_n=\gamma_1n^{-\beta},n\geq1\}$, $\gamma_1>0$, $\beta\in\big(\frac12,1\big)$.

\begin{remark}
Unlike~\cite[Theorem~3.3]{barrera:hal-01710394}, we do not average out the ES NSA estimators $\{\chi^h_n,0\leq n\leq\lceil h^{-2}\rceil\}$, inasmuch as in view of Theorem~\ref{thm:nested:clt}, their convergence rate of order $h$ is already optimal and their variance factor is independent of $\gamma_1$.
\end{remark}

\subsection{Convergence Rate Analysis}

\begin{theorem}
\label{thm:avg:nested:clt}
Within the framework of Theorem~\ref{thm:nested:clt}, if $\gamma_n=\gamma_1n^{-\beta}$, $n\geq1$, $\gamma_1>0$, $\beta\in\big(\frac12,1\big)$, then
\begin{equation*}
h^{-1}
\begin{pmatrix}
\overline{\xi}^h_{\lceil h^{-2}\rceil}-\xi^h_\star\\
\chi^h_{\lceil h^{-2}\rceil}-\chi^h_\star
\end{pmatrix}
\stackrel[]{\mathcal{L}}{\longrightarrow}\mathcal{N}(0,\overline\Sigma)
\quad\mbox{ as }\quad\mathcal{H}\ni h\downarrow0,
\end{equation*}
where
\begin{equation}
\label{eq:avg:Sigma}
\overline\Sigma=
\begin{pmatrix}
\frac{\alpha(1-\alpha)}{f_{X_0}(\xi^0_\star)^2}
&\frac\alpha{1-\alpha}\frac{\mathbb{E}[(X_0-\xi^0_\star)^+]}{f_{X_0}(\xi^0_\star)}\\
\frac\alpha{1-\alpha}\frac{\mathbb{E}[(X_0-\xi^0_\star)^+]}{f_{X_0}(\xi^0_\star)}
&\frac{\Var((X_0-\xi^0_\star)^+)}{(1-\alpha)^2}
\end{pmatrix}.
\end{equation}
\end{theorem}

\begin{proof}
See Appendix~\ref{prf:avg:nested:clt}.
\end{proof}

\begin{corollary}
\label{crl:avg:nested:clt}
Within the framework of Theorem~\ref{thm:avg:nested:clt},
\begin{equation*}
h^{-1}
\begin{pmatrix}
\overline{\xi}^h_{\lceil h^{-2}\rceil}-\xi^0_\star\\
\chi^h_{\lceil h^{-2}\rceil}-\chi^0_\star
\end{pmatrix}
\stackrel[]{\mathcal{L}}{\longrightarrow}\mathcal{N}\left(\begin{pmatrix}
-\frac{v(\xi^0_\star)}{f_{X_0}(\xi\star)}\\
-\int_{\xi^0_\star}^\infty\frac{v(\xi)}{1-\alpha}\mathrm{d}\xi
\end{pmatrix},
\overline\Sigma\right)
\quad\mbox{ as }\quad\mathcal{H}\ni h\downarrow0,
\end{equation*}
where $\overline\Sigma$ is given by \eqref{eq:avg:Sigma}.
\end{corollary}

\begin{proof}
For $h\in\mathcal{H}$, we decompose
\begin{equation*}
\begin{pmatrix}
\overline{\xi}^h_{\lceil h^{-2}\rceil}-\xi^0_\star\\
\chi^h_{\lceil h^{-2}\rceil}-\chi^0_\star
\end{pmatrix}
=\begin{pmatrix}
\overline{\xi}^h_{\lceil h^{-2}\rceil}-\xi^h_\star\\
\chi^h_{\lceil h^{-2}\rceil}-\chi^h_\star
\end{pmatrix}
+\begin{pmatrix}
\xi^h_\star-\xi^0_\star\\
\chi^h_\star-\chi^0_\star
\end{pmatrix}.
\end{equation*}
By Lemma~\ref{lmm:error}(\ref{lmm:error:weak}),
\begin{equation*}
h^{-1}\begin{pmatrix}
\xi^h_\star-\xi^0_\star\\
\chi^h_\star-\chi^0_\star
\end{pmatrix}
\to
-\begin{pmatrix}
\frac{v(\xi^0_\star)}{f_{X_0}(\xi\star)}\\
\int_{\xi^0_\star}^\infty\frac{v(\xi)}{1-\alpha}\mathrm{d}\xi
\end{pmatrix}
\quad\mbox{ as }\quad\mathcal{H}\ni h\downarrow0.
\end{equation*}
Using Theorem~\ref{thm:avg:nested:clt} and the above result concludes the proof.
\end{proof}

Let us comment the obtained CLT.

\begin{remark}
\begin{enumerate}[(i)]
\item
Unlike NSA, ANSA scores a convergence speed in $\OO(h)$ for both the VaR and ES, for any $\beta\in\big(\frac12,1\big)$.
The choice of $\beta$ in the range $\big(\frac12,1\big)$ is free of any constraint on the choice of $\gamma_1$.

\item
The covariance matrix \eqref{eq:avg:Sigma} is independent of $\gamma_1$, reflecting a strong numerical stability for ANSA. The VaR and ES estimates remain nonetheless asymptotically correlated.

The optimal variance factor for the NSA VaR estimate, derived in Remark~\ref{rmk:nsa:cv}(\ref{rmk:nsa:cv:ii}) for $\beta=1$, coincides with ANSA's variance factor for the VaR.

The ES variance factor is identical for both NSA and ANSA, since ANSA reuses the ES NSA scheme.

\item
Although the covariance matrix \eqref{eq:avg:Sigma} is not computable, it is approximatable.
See Remark~\ref{rmk:nsa:cv}(\ref{rmk:nsa:cv:iv}) for additional comments as well as an estimation formula of the ES variance factor.
\end{enumerate}
\end{remark}

\begin{remark}
\begin{enumerate}[(i)]
\item
\cite[Theorem~1.3]{10.1214/21-EJS1908} analyzes a two-time-scale unbiased algorithm where the ES update injects the averaged VaR estimator into \eqref{eq:ES:MC}. We refer to Remarks~\ref{rmk:1.6:v-ix}(\ref{rmk:unbiased:1})-(\ref{rmk:unbiased:2}) on the suitability of such result to our biased case.
The technique used to prove the aforementioned theorem relies on a spectral analysis of the algorithm's linearization, coupled with a tightness study of its diagonalized interpolation.

We also refer to~\cite{GADAT2023312} for non-asymptotic analyses of Polyak-Ruppert stochastic algorithms for strongly convex as well as globally {\L}ojasiewicz functions with a {\L}ojasiewicz exponent in $\big[0,\frac12\big]$.

\item
\cite[Theorem~2.8]{10.1214/15-AAP1109} establishes a CLT for biased Polyak-Ruppert schemes under a mean-reverting property~\cite[Assumption~(HMR)]{10.1214/15-AAP1109} that only deals with the VaR component of our ANSA scheme.

\item
See Remark~\ref{rmk:nested:cost}(\ref{rmk:2time}) for a comment on the absence of a bias in~\cite[Theorem~3.3]{barrera:hal-01710394} and its reliance on the unbiased averaged CLT~\cite[Theorem~3.2]{fort2015central}.

\end{enumerate}
\end{remark}

\subsection{Complexity Analysis}

\begin{theorem}
\label{cost:ansa}
Let $\varepsilon\in(0,1)$ be a prescribed accuracy. Within the framework of Theorem~\ref{thm:avg:nested:clt},
setting
\begin{equation*}
h=\frac1{\lceil\varepsilon^{-1}\rceil}\sim\varepsilon
\quad\mbox{ and }\quad
n=\lceil h^{-2}\rceil=\lceil\varepsilon^{-1}\rceil^2\sim\varepsilon^{-2}
\end{equation*}
achieves a convergence rate in distribution for the ANSA scheme \eqref{eq:sa:avg:nested:alg:xi}-\eqref{eq:sa:nested:alg:C} of order $\varepsilon$ as $\varepsilon\downarrow0$.
The corresponding complexity satisfies
\begin{equation*}
\Cost_\text{\rm\tiny ANSA}\leq C\varepsilon^{-3},
\end{equation*}
for some positive constant $C<\infty$.
\end{theorem}

\begin{proof}
The complexity of ANSA satisfies
\begin{equation*}
\Cost_\text{\rm\tiny ANSA}
\leq C\frac{n}h,
\end{equation*}
for some positive constant $C<\infty$.
By Corollary~\ref{crl:avg:nested:clt}, the convergence rate of NSA with bias parameter $h$ and iterations amount $n=\lceil h^{-2}\rceil$ is of order $h$ as $\mathcal{H}\ni h\downarrow0$, hence the result by setting $h$ such that $h\leq\varepsilon$.
\end{proof}

\begin{remark}
Theorems~\ref{cost:nsa} and~\ref{cost:ansa} reveal that, unlike NSA, ANSA achieves a complexity of order $\varepsilon^{-3}$ regardless of the value of $\beta\in\big(\frac12,1\big)$. This allows to circumvent the constraint on $\gamma_1$ that appears for the NSA scheme in the optimal case $\beta=1$.
\end{remark}

\section{Multilevel Stochastic Approximation Algorithm}

Viewing the complexity of order $\varepsilon^{-3}$, that is attained by NSA and ANSA, as a baseline for biased VaR and ES estimation, \cite{mlsa} suggests to adopt a multilevel strategy to further reduce it to an order of $\varepsilon^{-(2+\delta)}$, for an adequate $\delta\in(0,1)$.
In this section, we present an MLSA scheme for estimating the VaR and ES, then state an associated CLT before studying its complexity.

Let $h_0=\frac1K\in\mathcal{H}$. For some fixed integer $M>1$ and a number of levels $L\in\mathbb{N}^*$, we define the bias parameters
\begin{equation*}
\Big\{h_\ell\coloneqq\frac{h_0}{M^\ell}=\frac1{KM^\ell},0\leq\ell\leq L\Big\}\in\mathcal{H}^{L+1}.
\end{equation*}
Rather than simulating the biased solutions $\xi^{h_L}_\star$ and $\chi^{h_L}_\star$ as NSA and ANSA do, MLSA simulates their telescopic summations
\begin{align}
\xi^{h_L}_\star
&=\xi^{h_0}_\star+\sum_{\ell=1}^L\xi^{h_\ell}_\star-\xi^{h_{\ell-1}}_\star,
\label{eq:xi*^hL}
\\
\chi^{h_L}_\star
&=\chi^{h_0}_\star+\sum_{\ell=1}^L\chi^{h_\ell}_\star-\chi^{h_{\ell-1}}_\star.
\label{eq:C*^hL}
\end{align}
Indeed, following~\cite{mlsa,10.1214/15-AAP1109}, assuming the $F_{X_{h_\ell}}$ continuous and increasing for all $\ell\geq0$, for a sequence of positive integer iterations amounts $\mathbf{N}=\{N_\ell,0\leq\ell\leq L\}$, the MLSA estimators of the VaR and ES are given by
\begin{align}
\xi^\text{\tiny\rm ML}_\mathbf{N}
&=\xi^{h_0,0}_{N_0}+\sum_{\ell=1}^L\xi^{h_\ell,\ell}_{N_\ell}-\xi^{h_{\ell-1},\ell}_{N_\ell},
\label{eq:xi:ML}
\\
\chi^\text{\tiny\rm ML}_\mathbf{N}
&=\chi^{h_0,0}_{N_0}+\sum_{\ell=1}^L\chi^{h_\ell,\ell}_{N_\ell}-\chi^{h_{\ell-1},\ell}_{N_\ell}.
\label{eq:C:ML}
\end{align}
Each level $0\leq\ell\leq L$ is simulated independently.
First, we simulate the level $0$ pair $(\xi^{h_0,0}_{N_0},\chi^{h_0,0}_{N_0})$ using $N_0$ iterations of the NSA scheme \eqref{eq:sa:nested:alg:xi}-\eqref{eq:sa:nested:alg:C} with bias $h_0$.
Then, at each level $1\leq\ell\leq L$, given innovations $\{(X^{(n)}_{h_{\ell-1},\ell},X^{(n)}_{h_\ell,\ell}),1\leq n\leq N_\ell\}\stackrel{\text{\rm\tiny i.i.d.}}{\sim}(X_{h_{\ell-1},\ell},X_{h_\ell,\ell})$, for $j\in\{\ell-1,\ell\}$, we perform $N_\ell$ iterations of the dynamics
\begin{align}
\xi^{h_j,\ell}_{n+1}
&=\xi^{h_j,\ell}_n-\gamma_{n+1}H_1(\xi^{h_j,\ell}_n,X_{h_j,\ell}^{(n+1)}),
\label{eq:sa:ml:alg:xi}
\\
\chi^{h_j,\ell}_{n+1}
&=\chi^{h_j,\ell}_n-\frac1{n+1}H_2(\chi^{h_j,\ell}_n,\xi^{h_j,\ell}_n,X_{h_j,\ell}^{(n+1)}),
\label{eq:sa:ml:alg:chi}
\end{align}
starting from real-valued random variables $\xi^{h_{\ell-1},\ell}_0$, $\xi^{h_\ell,\ell}_0$ and $\chi^{h_{\ell-1},\ell}_0=\chi^{h_\ell,\ell}_0=0$, the pair $(\xi^{h_{\ell-1,\ell}}_0,\xi^{h_\ell,\ell}_0)$ being independent from $\{(X^{(n)}_{h_{\ell-1},\ell},X^{(n)}_{h_\ell,\ell}),1\leq n\leq N_\ell\}$.
Note importantly that $X_{h_{\ell-1},\ell}$ and $X_{h_\ell,\ell}$ must be perfectly correlated in the following sense: after computing $X_{h_{\ell-1,\ell}}$ via
\begin{equation*}
X_{h_{\ell-1},\ell}=\frac1{KM^{\ell-1}}\sum_{k=1}^{KM^{\ell-1}}\varphi(Y_\ell,Z_\ell^{(k)}),
\end{equation*}
where $\{Z_\ell^{(k)},1\leq k\leq KM^{\ell-1}\}\stackrel{\text{\rm\tiny i.i.d.}}{\sim}Z$ are independent from $Y_\ell\sim Y$, $X_{h_\ell,\ell}$ is computed by sampling additional random variables $\big\{Z_\ell^{(k)},KM^{\ell-1}<k\leq KM^\ell\big\}\stackrel{\text{\rm\tiny i.i.d.}}{\sim}Z$ independently from $\{Z_\ell^{(k)},1\leq k\leq KM^{\ell-1}\}$ and $Y_\ell$, and taking
\begin{equation*}
X_{h_\ell,\ell}=\frac1MX_{h_{\ell-1},\ell}+\frac1{KM^\ell}\sum_{k=KM^{\ell-1}+1}^{KM^\ell}\varphi(Y_\ell,Z_\ell^{(k)}).
\end{equation*}

\subsection{Convergence Rate Analysis}

\begin{lemma}\label{lmm:technical}
Suppose that $\varphi(Y,Z)\in L^2(\mathbb{P})$.
\begin{enumerate}[(i)]
    \item\label{lmm:technical-i}
Then,
\begin{equation*}
h_\ell^{-\frac12}(X_{h_\ell}-X_{h_{\ell-1}})
\eqqcolon G_\ell
\stackrel[]{\mathcal{L}}{\longrightarrow}
G\coloneqq
\big((M-1)\Var(\varphi(Y,Z)|Y)\big)^\frac12\mathfrak{N}
\quad\mbox{ as }\quad
\ell\uparrow\infty,
\end{equation*}
where $\mathfrak{N}\sim\mathcal{N}(0,1)$ is independent of $Y$.

    \item\label{lmm:technical-ii}
Assume that, for all $\ell\geq1$, $F_{X_{h_{\ell-1}}|G_\ell}$ is $\Pas$ continuously differentiable with derivative $f_{X_{h_{\ell-1}}|G_\ell}$, and that $\{(x,g)\mapsto f_{X_{h_{\ell-1}}|G_\ell=g}(x),\ell\geq1\}$ are bounded uniformly in $\ell\geq1$ and converge locally uniformly to some bounded and continuous function $(x,g)\mapsto f_g(x)$.
Then, for any $\xi\in\mathbb{R}$,
\begin{equation*}
{h_\ell^{-\frac12}\mathbb{E}\big[\big|\mathds{1}_{X_{h_\ell}>\xi}-\mathds{1}_{X_{h_{\ell-1}}>\xi}\big|\big]}
\to\mathbb{E}[|G|f_G(\xi)]
\quad\mbox{ as }\quad\ell\uparrow\infty.
\end{equation*}

    \item\label{lmm:technical-iii}
Assume that $F_{X_0}$ is continuous. Then, for any $\xi\in\mathbb{R}$,
\begin{equation*}
h_\ell^{-\frac12}\big((X_{h_\ell}-\xi)^+-(X_{h_{\ell-1}}-\xi)^+\big)
\stackrel[]{\mathcal{L}}{\longrightarrow}\mathds{1}_{X_0>\xi}\,G
\quad\mbox{ as }\quad\ell\uparrow\infty.
\end{equation*}
\end{enumerate}
\end{lemma}

\begin{proof}
See Appendix~\ref{prf:technical}.
\end{proof}

\begin{remark}
The framework of Lemma~\ref{lmm:technical}(\ref{lmm:technical-ii}) strengthens that of~\cite[Proposition~3.2(ii)]{mlsa} (recalled in Lemma~\ref{lmm:recall}(\ref{lmm:lipschitz}))
by supposing the existence and convergence of pdf for the conditional laws $X_{h_{\ell-1}}\,|\,G_\ell$, $\ell\geq1$.
As shown in Step~9 of Theorem~\ref{thm:ml:clt}'s proof, Appendix~\ref{prf:ml:clt}, such framework allows to retrieve an asymptotic distribution for the martingale part of the VaR error decomposition \eqref{eq:xiML-xi*(hL)}.
While remaining weaker than~\cite[Assumption~2.5]{Gordy}, it reaches optimal complexity according to~\cite[Remark~3.3 \& Theorem~3.9(i)]{mlsa}.
\end{remark}

\begin{assumption}[{\cite[Assumption~3.4]{mlsa}}]\label{asp:fh-f0}
There exist positive constants $C,\delta_0<\infty$ such that, for any $h\in\mathcal{H}$ and any compact set $\mathcal{K}\subset\mathbb{R}$,
\begin{equation*}
\sup_{\xi\in\mathcal{K}}{|f_{X_h}(\xi)-f_{X_0}(\xi)|}\leq Ch^{\frac14+\delta_0}.
\end{equation*}
\end{assumption}

\begin{remark}[{\cite[Remark~3.5]{mlsa}}]
Assumption~\ref{asp:fh-f0} relaxes the postulate of~\cite[Proposition~5.1(a)]{Giorgi2020}.
\end{remark}

We state below a CLT for the MLSA scheme.
According to context, we use the notation $N_\ell$ to designate both $N_\ell$ and $\big\lceil N_\ell\big\rceil$ interchangeably.

\begin{theorem}
\label{thm:ml:clt}
Suppose that the frameworks of Theorem~\ref{thm:nested:clt} and Lemma~\ref{lmm:technical}(\ref{lmm:technical-ii}) hold and that Assumption~\ref{asp:fh-f0} is satisfied.
If $\gamma_n=\gamma_1n^{-\beta}$, $n\geq1$, $\gamma_1>0$, $\beta\in\big(\frac12,1\big]$, with $\lambda\gamma_1>1$ if $\beta=1$, then, setting
\begin{equation}
\label{eq:N_ell}
N_\ell=h_L^{-\frac2\beta}\bigg(\sum_{\ell'=0}^Lh_{\ell'}^{-\frac{2\beta-1}{2(1+\beta)}}\bigg)^\frac1\beta h_\ell^\frac3{2(1+\beta)},
\quad0\leq\ell\leq L,
\end{equation}
it holds that
\begin{equation*}
\begin{pmatrix}
h_L^{-1}\big(\xi^\text{\tiny\rm ML}_\mathbf{N}-\xi^{h_L}_\star\big)\\
h_L^{-\frac1\beta-\frac{2\beta-1}{4\beta(1+\beta)}}\big(\chi^\text{\tiny\rm ML}_\mathbf{N}-\chi^{h_L}_\star\big)
\end{pmatrix}
\stackrel[]{\mathcal{L}}{\longrightarrow}
\mathcal{N}(0,\Sigma^\text{\tiny\rm ML}_\beta)
\quad\mbox{ as }\quad
L\uparrow\infty,
\end{equation*}
where $\Sigma^\text{\tiny\rm ML}_\beta=$
\begin{equation}
\label{eq:Sigma:ML}
\begin{pmatrix}
\frac{\gamma_1\mathbb{E}[|G|f_G(\xi^0_\star)]}{(1-\alpha)(2f_{X_0}(\xi^0_\star)-(1-\alpha)\gamma_1^{-1}\mathds{1}_{\beta=1})}
&0\\
0
&\frac{h_0^\frac{2\beta-1}{2(1+\beta)}\big(M^\frac{2\beta-1}{2(1+\beta)}-1\big)^\frac1\beta}{(1-\alpha)^2}
\Big(\frac{h_0^{-1}\Var((X_{h_0}-\xi^{h_0}_\star)^+)}{M^\frac{2\beta-1}{2\beta(1+\beta)}}+\frac{\Var(\mathds{1}_{X_0>\xi^0_\star}\,G)}{M^\frac{2\beta-1}{2(1+\beta)}-1}\Big)
\end{pmatrix},
\end{equation}
with $G$ being defined in Lemma~\ref{lmm:technical}(\ref{lmm:technical-i}).
\end{theorem}

\begin{proof}
See Appendix~\ref{prf:ml:clt}.
\end{proof}

\begin{remark}\label{rmk:ml}
\begin{enumerate}[(i)]
\item
As noted in Remark~\ref{rmk:nested:misc}(\ref{rmk:beta}), while the $L^2(\mathbb{P})$-controls of~\cite[Theorem~3.6]{mlsa} are available for all $\beta\in(0,1]$, the CLT~\ref{thm:ml:clt} is rather available for $\beta\in\big(\frac12,1\big]$.

\item
We used in \eqref{eq:N_ell} the optimal iterations amounts for the VaR MLSA scheme according to~\cite[Theorem~3.9(i)]{mlsa}. They were obtained by optimizing the MLSA complexity while constraining the VaR estimation error to a prescribed accuracy $\varepsilon$.

\item
The VaR and ES evolve at different speeds: in $\OO(h_L)$ for the VaR and in $\OO\big(h_L^{\frac1\beta+\frac{2\beta-1}{4\beta(1+\beta)}}\big)$ for the ES as $L\uparrow\infty$. The optimal ES convergence speed in $\OO\big(h_L^\frac98\big)$ is attained for $\beta=1$ under the condition $\lambda\gamma_1>1$.

\item
The VaR variance factor depends on $\gamma_1$ in the same way the NSA covariance matrix \eqref{eq:Sigma} does (Remark~\ref{rmk:nested:misc}(\ref{rmk:nsa:cv:ii})).

To ensure the additional condition $\lambda\gamma_1>1$ in the optimal case $\beta=1$, a rule of thumb would be to choose $\gamma_1$ large enough with regard to the variables of the problem. But this is not satisfactory, inasmuch as the VaR component of the covariance matrix $\Sigma_1$ in \eqref{eq:Sigma} evolves in $\OO(\gamma_1)$ for $\gamma_1$ large, damaging the algorithm's convergence.

Setting it too small blows up the VaR convergence factor, as it evolves in $\OO\big(\big(\gamma_1-\frac{1-\alpha}{2f_{X_0}(\xi^0_\star)}\big)^{-1}\big)$ if $\gamma_1\downarrow\frac{1-\alpha}{2f_{X_0}(\xi^0_\star)}$.

Such a behavior translates a high numerical instability for the VaR estimation, as has been observed empirically in~\cite{mlsa}.
The optimal $\gamma_1$, although explicit, is not calculable since the value of $f_{X_0}(\xi^0_\star)$ is inaccessible.
Thus $\gamma_1$ must be carefully fine-tuned via a grid search.

\item\label{rmk:h0:L}
The ES variance factor depends on $\beta$ and is minimal for $\beta=1$.
It is independent of $\gamma_1$, which should entail a strong numerical stability of the ES MLSA scheme, as observed empirically in~\cite{mlsa}.

The ES variance factor is composed of two terms: the first one stems from the initial level $\ell=0$, dependent on $X_{h_0}$ and $\xi^{h_0}_\star$,  and the second one stems from the asymptotic behavior of the algorithm, dependent on $G$, $X_0$ and $\xi^0_\star$.

Setting $h_0$ sufficiently small should help reduce $\Var((X_{h_0}-\xi^{h_0}_\star)^+)$, hence the impact of the initial level simulation on the ES variance.
However, as clarified in Theorem~\ref{thm:cost:ml}, for a prescribed accuracy $\varepsilon\in(0,1)$, $h_0$ must be floored by $\varepsilon$, and $L$ set as in \eqref{eq:L}. Reducing $h_0$ would consequently reduce $L$, hence exiting the asymptotic regime.
$h_0$ should therefore be fine-tuned.
See~\cite[Remark~2.6]{10.1214/15-AAP1109} for further comments on setting $h_0$.

\item
The VaR-ES asymptotic MLSA correlation is null, which is linked to the speed differential between the VaR and ES. It expresses an independent behavior of the VaR and ES as $L\uparrow\infty$. The strong stability of the ES is thus hardly affected by the instability of the VaR. This may seem surprising since the scheme \eqref{eq:sa:ml:alg:chi} clearly incorporates the VaR estimates in the ES approximation.
However, it can be explained by the robustifying effect of the averaging property of the ES estimates \eqref{eq:sa:ml:alg:chi}.

\item\label{rmk:sim}
The covariance matrix \eqref{eq:Sigma:ML} is not calculable, as so many unknowable variables intervene in its definition. It should consequently be estimated using multiple runs of the MLSA scheme.

The ES variance factor could notwithstanding be estimated by approximating $\Var((X_{h_0}-\xi^{h_0}_\star)^+)$ as in \eqref{eq:Var(X0-xi0)} with
\begin{equation*}
\Var((X_{h_0}-\xi^{h_0}_\star)^+)\approx\frac1{N_0}\sum_{k=1}^{N_0}\big((X_{h_0}^{(k)}-\xi^{h_0}_{k-1})^+\big)^2-\bigg(\frac1{N_0}\sum_{k=1}^{N_0}(X_{h_0}^{(k)}-\xi^{h_0}_{k-1})^+\bigg)^2,
\end{equation*}
$G$ with
\begin{equation*}
G\approx G_{h_L}=\bigg\{(M-1)\bigg(\frac1{\lceil h_L^{-1}\rceil}\sum_{k=1}^{\lceil h_L^{-1}\rceil}\varphi(Y,Z^{(k)})^2-X_{h_L}^2\bigg)\bigg\}^\frac12\mathcal{N}(0,1),
\end{equation*}
and $\Var(\mathds{1}_{X_0>\xi^0_\star}\,G)$ with
\begin{equation*}
\Var(\mathds{1}_{X_0>\xi^0_\star}\,G)\approx\frac1{N_L}\sum_{k=1}^{N_L}\big(\mathds{1}_{X_{h_L}^{(k)}>\xi^{h_L}_{k-1}}G_{h_L}^{(k)}\big)^2-\bigg(\frac1{N_L}\sum_{k=1}^{N_L}\mathds{1}_{X_{h_L}^{(k)}>\xi^{h_L}_{k-1}}G_{h_L}^{(k)}\bigg)^2.
\end{equation*}

\item
\cite[Theorem~2.11]{10.1214/15-AAP1109} establishes a CLT for biased multilevel schemes under a mean-reverting property~\cite[Assumption~(HMR)]{10.1214/15-AAP1109} that only concerns the VaR component of our MLSA scheme.
Additionally, the martingale array studied in the proof of the aforementioned theorem is different from the one in Step~9 of the proof of Theorem~\ref{thm:ml:clt}, Appendix~\ref{prf:ml:clt}. Indeed, due to the discontinuity of the gradient $H_1$ \eqref{eq:H1}, the summation order of the martingale array had to be swapped to allow invoking the CLT for martingale arrays~\cite[Corollary~3.1]{nla.cat-vn2887492}.

\end{enumerate}
\end{remark}

\subsection{Complexity Analysis}

The global error of the MLSA scheme \eqref{eq:xi:ML}-\eqref{eq:C:ML} decomposes into statistical and bias errors:
\begin{align}
\xi^\text{\tiny\rm ML}_\mathbf{N}-\xi^0_\star
&=\big(\xi^\text{\tiny\rm ML}_\mathbf{N}-\xi^{h_L}_\star\big)+(\xi^{h_L}_\star-\xi^0_\star),
\label{eq:global:error:VaR:MLSA}
\\
\chi^\text{\tiny\rm ML}_\mathbf{N}-\chi^0_\star
&=\big(\chi^\text{\tiny\rm ML}_\mathbf{N}-\chi^{h_L}_\star\big)+(\chi^{h_L}_\star-\chi^0_\star).
\label{eq:global:error:ES:MLSA}
\end{align}

\begin{theorem}\label{thm:cost:ml}
Let $\varepsilon\in(0,1)$ be a prescribed accuracy.
Within the framework of Theorem~\ref{thm:ml:clt}, if $h_0>\varepsilon$, then,
setting
\begin{equation}\label{eq:L}
L=\bigg\lceil\frac{\ln{(h_0/\varepsilon)}}{\ln{M}}\bigg\rceil\sim\frac{\left|\ln{\varepsilon}\right|}{\ln{M}}
\end{equation}
achieves a global convergence rate in distribution for the MLSA scheme \eqref{eq:xi:ML}-\eqref{eq:C:ML} of order $\varepsilon$ as $\varepsilon\downarrow0$.
The corresponding complexity satisfies
\begin{equation*}
\Cost^\beta_\text{\rm\tiny MLSA}\leq C\varepsilon^{-1-\frac3{2\beta}},
\end{equation*}
for some positive constant $C<\infty$.
The optimal complexity, reached for $\beta=1$ under the constraint $\lambda\gamma_1>1$, satisfies
\begin{equation*}
\Cost^1_\text{\rm\tiny MLSA}\leq C\varepsilon^{-\frac52}.
\end{equation*}
\end{theorem}

\begin{proof}
The complexity of MLSA satisfies
\begin{equation*}
\Cost^\beta_\text{\rm\tiny MLSA}
\leq C\sum_{\ell=0}^L\frac{N_\ell}{h_\ell},
\end{equation*}
for some positive constant $C<\infty$.
In view of \eqref{eq:global:error:VaR:MLSA}-\eqref{eq:global:error:ES:MLSA}, Lemma~\ref{lmm:error}(\ref{lmm:error:weak}) and Theorem~\ref{thm:ml:clt} guarantee that the convergence rate of the MLSA scheme is of order $h_L$ as $L\uparrow\infty$. Thus, to achieve an error of order $\varepsilon$, we must choose $h_L\leq\varepsilon$, hence \eqref{eq:L} follows. The remaining complexity computations are standard and are thereby omitted.
\end{proof}

\begin{remark}
MLSA is computationally optimal for $\beta=1$ under the condition $\lambda\gamma_1>1$. It runs optimally in $\OO(\varepsilon^{-\frac52})$ time, an order of magnitude faster than NSA and ANSA that run optimally in $\OO(\varepsilon^{-3})$ time.
\end{remark}

\section{Averaged Multilevel Stochastic Approximation Algorithm}
\label{sec:mlsa:avg}

As shown in Theorem~\ref{thm:cost:ml}, the MLSA scheme \eqref{eq:xi:ML}-\eqref{eq:C:ML} is optimal for $\beta=1$ under the non trivial constraint $\lambda\gamma_1>1$.
The hyperparameter $\gamma_1$ must be carefully fine-tuned, adding a significant tuning phase to MLSA's execution time. To address this limitation, we look into applying the Polyak-Ruppert averaging principle to the multilevel paradigm.

For each level $0\leq\ell\leq L$ and $j\in\{(\ell-1)^+,\ell\}$, setting $\bar\xi^{h_j,\ell}_0=0$, we consider the estimate $\overline{\xi}^{h_j,\ell}_{N_\ell}$ calculated by averaging out the simulations $\{\xi^{h_j,\ell}_n,1\leq n\leq N_\ell\}$ from \eqref{eq:sa:ml:alg:xi}.
The AMLSA estimator of the VaR is defined by
\begin{equation}
\label{eq:xi:ML:avg}
\overline{\xi}^\text{\tiny\rm ML}_\mathbf{N}=\overline{\xi}^{h_0,0}_{N_0}
+\sum_{\ell=1}^L\overline{\xi}^{h_\ell,\ell}_{N_\ell}-\overline{\xi}^{h_{\ell-1},\ell}_{N_\ell}.
\end{equation}

\subsection{Convergence Rate Analysis}

\begin{theorem}
\label{thm:avg:ml:clt}
Suppose that the framework of Theorem~\ref{thm:ml:clt} holds and that $\delta_0\geq\frac18$. If $\gamma_n=\gamma_1n^{-\beta}$, $n\geq1$, $\gamma_1>0$, $\beta\in\big(\frac89,1\big)$, then, setting
\begin{equation}
\label{eq:N_ell:avg}
N_\ell
=h_L^{-2}\bigg(\sum_{\ell'=0}^Lh_{\ell'}^{-\frac14}\bigg)h_\ell^\frac34,
\quad0\leq\ell\leq L,
\end{equation}
it holds that
\begin{equation*}
\begin{pmatrix}
h_L^{-1}\big(\overline{\xi}^\text{\tiny\rm ML}_\mathbf{N}-\xi^{h_L}_\star\big)\\
h_L^{-\frac98}\big(\chi^\text{\tiny\rm ML}_\mathbf{N}-\chi^{h_L}_\star\big)
\end{pmatrix}
\stackrel[]{\mathcal{L}}{\longrightarrow}
\mathcal{N}(0,\overline\Sigma^\text{\tiny\rm ML})
\quad\mbox{ as }\quad
L\uparrow\infty,
\end{equation*}
where
\begin{equation}\label{eq:avg:Sigma:ML}
\overline\Sigma^\text{\tiny\rm ML}=
\begin{pmatrix}
\frac{\mathbb{E}[|G|f_G(\xi^0_\star)]}{(1-\alpha)^2(1-M^{-1/4})}
&0\\
0
&\frac{h_0^{-3/8}(1-M^{-1/4})^{1/2}\Var((X_{h_0}-\xi^{h_0}_\star)^+)}{(1-\alpha)^2}+\frac{h_0^{1/4}\Var(\mathds{1}_{X_0>\xi^0_\star}G)}{(1-\alpha)^2M^{1/4}}
\end{pmatrix},
\end{equation}
with $G$ being defined in Lemma~\ref{lmm:technical}(\ref{lmm:technical-i}).
\end{theorem}

\begin{proof}
See Appendix~\ref{prf:avg:ml:clt}.
\end{proof}

\begin{remark}
\begin{enumerate}[(i)]
    \item
    The iterations amounts in \eqref{eq:N_ell:avg} are determined in Appendix~\ref{apx:B2}.

    \item
    In contrast with classical Polyak-Ruppert type algorithms, Theorem~\ref{thm:avg:ml:clt} does not support all $\beta\in\big(\frac12,1\big)$, but instead restricts $\beta$ to $\big(\frac89,1\big)$.
    The lower threshold on the admissible values of $\beta$ manifests when studying the statistical error of the VaR.

    \item
    The convergence rate for the VaR is in $\OO(h_L)$, as with MLSA. The convergence rate for the ES is in $\OO(h_L^\frac98)$ for any $\beta\in\big(\frac89,1\big)$, without any extra condition on $\gamma_1$.
    This rate matches the optimal convergence rate of the ES MLSA scheme, albeit with the additional constraint $\lambda\gamma_1>1$.

    \item
     As with MLSA, the ES variance factor is composed of two terms, the first one affected by the level $0$ simulations and the second one affected by the asymptotics of the ES. A careful fine-tuning of $h_0$ is hence necessary to reduce the impact of the initial level simulation on the ES.
     See Remark~\ref{rmk:ml}(\ref{rmk:h0:L}) for complementary comments in this regard.

    \item
    The asymptotic covariance matrix is independent of $\gamma_1$ (and $\beta$), suggesting numerical stability for the VaR and ES estimations for any $\beta\in\big(\frac89,1\big)$. It also suggests asymptotic decorrelation of the VaR and ES schemes, reinforcing their robustness.

    \item
    The covariance matrix \eqref{eq:avg:Sigma:ML} is not determinable numerically, but it is estimatable using multiple runs of the AMLSA scheme.

    The ES variance factor could be estimated with the help of the formulas detailed in Remark~\ref{rmk:ml}(\ref{rmk:sim}).
\end{enumerate}
\end{remark}

\subsection{Complexity Analysis}

Recalling the ES global error decomposition \eqref{eq:global:error:ES:MLSA}, we decompose similarly the VaR's global error into a statistical and a bias error:
\begin{equation}
\overline{\xi}^\text{\tiny\rm ML}_\mathbf{N}-\xi^0_\star
=\big(\overline{\xi}^\text{\tiny\rm ML}_\mathbf{N}-\xi^{h_L}_\star\big)+(\xi^{h_L}_\star-\xi^0_\star).
\label{eq:global:error:VaR:AMLSA}
\end{equation}

\begin{theorem}\label{thm:cost:aml}
Let $\varepsilon\in(0,1)$ be some prescribed accuracy.
Within the framework of Theorem~\ref{thm:avg:ml:clt}, if $h_0>\varepsilon$, then, setting $L$ as in \eqref{eq:L} achieves a global convergence rate in distribution for the AMLSA scheme \eqref{eq:xi:ML:avg}-\eqref{eq:C:ML} of order $\varepsilon$ as $\varepsilon\downarrow0$.
The corresponding complexity satisfies
\begin{equation*}
\Cost_\text{\rm\tiny AMLSA}\leq C\varepsilon^{-\frac52},
\end{equation*}
for some positive constant $C<\infty$.
\end{theorem}

\begin{proof}
The complexity of AMLSA satisfies
\begin{equation*}
\Cost_\text{\rm\tiny AMLSA}
\leq C\sum_{\ell=0}^L\frac{N_\ell}{h_\ell},
\end{equation*}
for some positive constant $C<\infty$.
Lemma~\ref{lmm:error}(\ref{lmm:error:weak}) and Theorem~\ref{thm:ml:clt} guarantee that, with regard to \eqref{eq:global:error:VaR:AMLSA}-\eqref{eq:global:error:ES:MLSA}, the convergence rate of the AMLSA scheme is of order $h_L$ as $L\uparrow\infty$. Therefore, to attain an error of order $\varepsilon$, we must ensure that $h_L\leq\varepsilon$, hence \eqref{eq:L} follows. The remaining cost computations are standard and are therefore skipped.
\end{proof}

\begin{remark}
The complexity of AMLSA is independent of $\beta\in\big(\frac89,1\big)$ and bears no additional condition on $\gamma_1>0$.
\end{remark}

To ease comparison, the previous convergence rate and complexity results are gathered in Tables~\ref{tbl:cvr} and~\ref{tbl:cpx}.

\begin{table}[H]
    \centering
    \begin{tabular}{|l|l|c|c|l|}
         \hline
         \multirow{2}{*}{Algorithm} & \multirow{2}{10mm}{\centering CLT} & \multicolumn{2}{c|}{Convergence Rate} & \multirow{2}{10mm}{Conditions} \\
         \cline{3-4}
         & & VaR & ES & \\
         \specialrule{0.1em}{0em}{0em}
         \multirow{2}{*}{NSA \eqref{eq:sa:nested:alg:xi}-\eqref{eq:sa:nested:alg:C}}
         & \multirow{2}{*}{Corollary~\ref{crl:nested:clt}}
         & \multirow{2}{*}{\centering $h^\beta$}
         & \multirow{2}{*}{\centering $h$}
         & $\mathcal{H}\ni h\downarrow0$, $\beta\in\big(\frac12,1\big]$,\\
         & & & & with $\lambda\gamma_1>1$ if $\beta=1$\\
         \hline
         ANSA \eqref{eq:sa:avg:nested:alg:xi}-\eqref{eq:sa:nested:alg:C}
         & Corollary~\ref{crl:avg:nested:clt}
         & $h$
         & $h$
         & $\mathcal{H}\ni h\downarrow0$, $\beta\in\big(\frac12,1\big)$\\
         \hline
         \multirow{2}{*}{MLSA \eqref{eq:sa:ml:alg:xi}-\eqref{eq:sa:ml:alg:chi}}
         & \multirow{2}{*}{Theorem~\ref{thm:ml:clt}}
         & \multirow{2}{*}{\centering $h_L$}
         & \multirow{2}{*}{\centering $h_L^{\frac1\beta+\frac{2\beta-1}{4\beta(1+\beta)}}$}
         & $L\uparrow\infty$, $\beta\in\big(\frac12,1\big]$,\\
         & & & & with $\lambda\gamma_1>1$ if $\beta=1$\\
         \hline
         AMLSA \eqref{eq:xi:ML:avg}-\eqref{eq:sa:ml:alg:chi}
         & Theorem~\ref{thm:avg:ml:clt}
         & $h_L$
         & $h_L^\frac98$
         & $L\uparrow\infty$, $\beta\in\big(\frac89,1\big)$\\
         \hline
    \end{tabular}
    \caption{Summary table of the convergence rate results.
    $h\in\mathcal{H}$ is a bias parameter and $L$ is a number of levels.}
    \label{tbl:cvr}
\end{table}

\begin{table}[H]
    \centering
    \begin{tabular}{|c|c|c|}
         \hline
         Algorithm & Complexity Theorem & Complexity \\
         \specialrule{0.1em}{0em}{0em}
         NSA \eqref{eq:sa:nested:alg:xi}-\eqref{eq:sa:nested:alg:C}
         & Theorem~\ref{cost:nsa}
         & $\varepsilon^{-\frac3\beta}$, $\beta\in\big(\frac12,1\big]$, with $\lambda\gamma_1>1$ if $\beta=1$\\
         \hline
         ANSA \eqref{eq:sa:avg:nested:alg:xi}-\eqref{eq:sa:nested:alg:C}
         & Theorem~\ref{cost:ansa}
         & $\varepsilon^{-3}$, $\beta\in\big(\frac12,1\big)$\\
         \hline
         MLSA \eqref{eq:sa:ml:alg:xi}-\eqref{eq:sa:ml:alg:chi}
         & Theorem~\ref{thm:cost:ml}
         & $\varepsilon^{-1-\frac3{2\beta}}$, $\beta\in\big(\frac12,1\big]$, with $\lambda\gamma_1>1$ if $\beta=1$\\
         \hline
         AMLSA \eqref{eq:xi:ML:avg}-\eqref{eq:sa:ml:alg:chi}
         & Theorem~\ref{thm:cost:aml}
         & $\varepsilon^{-\frac52}$, $\beta\in\big(\frac89,1\big)$\\
         \hline
    \end{tabular}
    \caption{Summary table of the complexity results. $\varepsilon\in(0,1)$ designates a prescribed accuracy.}
    \label{tbl:cpx}
\end{table}

\section{Financial Case Study}
\label{sec:swap}
The goal of this numerical study is to check and compare empirically the validity of the Corollaries~\ref{crl:nested:clt} and~\ref{crl:avg:nested:clt} and Theorems~\ref{thm:ml:clt} and~\ref{thm:avg:ml:clt}.
For benchmarking purposes, we consider a numerical setting that allows to retrieve the VaR and ES of a financial loss either analytically, using an unbiased scheme or using a biased one like the ones studied in the previous sections.
The code for this numerical case study is available at \href{https://github.com/azarlouzi/avg_mlsa}{\texttt{github.com/azarlouzi/avg\_mlsa}}.

The risk-free rate is $r$ and derivative pricing is done under the probability measure $\mathbb{P}$.
We consider a swap of strike $\bar{K}$ and maturity $T$ on some underlying (FX or interest) rate.
The swap is issued at par.
The rate's risk neutral model $\{S_t,0\leq t\leq T\}$ is a Bachelier process of inital value $S_0$, drift $\kappa$ and volatility $\sigma$.
The swap pays at coupon
dates $0<T_1<\dots<T_d=T$ the cash flows $\Delta T_i(S_{T_{i-1}}-\bar{K})$, where $\Delta T_i=T_i-T_{i-1}$ with $T_0=0$.
The swap's nominal $\bar{N}$ is set such that each leg is worth $100$ at inception.
For $t\in[0,T]$, we denote $i_t$ the integer such that $t\in[T_{i_t-1},T_{i_t})$, if $t\in[0,T)$, and $+\infty$ otherwise.

Hence
\begin{equation}
\label{eq:rate}
\mathrm{d}S_t=\kappa S_t\mathrm{d}t+\sigma\mathrm{d}W_t,
\quad\mbox{ i.e. }\quad
S_t=S_0\mathrm{e}^{\kappa t}+\sigma\int_0^t\mathrm{e}^{-\kappa(s-t)}\mathrm{d}W_s,
\end{equation}
where $\{W_t,0\leq t\leq T\}$ is a standard Brownian motion.s
The fair value of the swap at time $t\in[0,T]$ is
\begin{equation*}
P_t
=\bar{N}\,\mathbb{E}\bigg[\sum_{i=i_t}^d\mathrm{e}^{-r(T_i-t)}\Delta T_i(S_{T_{i-1}}-\bar{K})\bigg|\mathcal{F}_t\bigg].
\end{equation*}
The loss at some short time horizon $\delta\in(0,T_1)$ on a short position on the swap is
\begin{equation*}
X_0=\mathrm{e}^{-r\delta}P_\delta.
\end{equation*}

We are interested in retrieving the VaR $\xi^0_\star$ and ES $\chi^0_\star$ of such a position, at some confidence level $\alpha\in(0,1)$.

\subsection{Analytical and Simulation Formulas}

The swap being issued at par, i.e.~$P_0=0$, it follows that
\begin{equation}
\label{eq:strike}
\bar{K}=S_0\frac{\sum_{i=1}^d\mathrm{e}^{-rT_i}\Delta T_i\,\mathrm{e}^{\kappa T_{i-1}}}{\sum_{i=1}^d\mathrm{e}^{-rT_i}\Delta T_i}.
\end{equation}
Note that $i_\delta=1$, so that, by \eqref{eq:rate}-\eqref{eq:strike} and the fact that $\pm\int_0^\delta\mathrm{e}^{-\kappa s}\mathrm{d}W_s\sim\mathcal{N}\big(0,\int_0^\delta\mathrm{e}^{-2\kappa s}\mathrm{d}s\big)$,
\begin{equation}
\label{eq:X0=Phi(Y):bis}
X_0\stackrel{\mathcal{L}}{=}\eta Y,
\quad\mbox{ where }\quad
\eta=\bar{N}\sigma\sqrt{\frac{1-\mathrm{e}^{-2\kappa\delta}}{2\kappa}}\sum_{i=2}^d\mathrm{e}^{-rT_i}\Delta T_i\,\mathrm{e}^{\kappa T_{i-1}}
\end{equation}
and $Y\sim\mathcal{N}(0,1)$ is independent of $\{W_t,0\leq t\leq T\}$.
This allows to simulate $X_0$ exactly.

The values of $\xi^0_\star$ and $\chi^0_\star$ can be obtained analytically. Indeed,
\begin{equation}
\label{eq:xi:*}
\alpha=\mathbb{P}(X_0\leq\xi^0_\star)=\mathbb{P}(\eta Y\leq\xi^0_\star),
\quad\mbox{ i.e. }\quad
\xi^0_\star=\eta F^{-1}(\alpha),
\end{equation}
where $F$ denotes the standard Gaussian cdf.
Likewise,
\begin{equation}
\label{eq:chi:*}
\chi^0_\star=\mathbb{E}[X_0|X_0\geq\xi^0_\star]
=\frac\eta{1-\alpha}\,\mathbb{E}[Y\mathds{1}_{\eta Y\geq\xi^0_\star}],
\quad\mbox{ i.e. }\quad
\chi^0_\star=\frac\eta{1-\alpha}f\Big(\frac{\xi^0_\star}\eta\Big),
\end{equation}
where $f$ designates the standard Gaussian pdf.

We also have
\begin{equation*}
X_0
=\bar{N}\sigma\mathbb{E}\bigg[\sum_{i=2}^d\mathrm{e}^{-rT_i}\Delta T_i\,\mathrm{e}^{\kappa T_{i-1}}\int_0^{T_{i-1}}\mathrm{e}^{-\kappa s}\mathrm{d}W_s\bigg|\mathcal{F}_\delta\bigg]
\stackrel{\mathcal{L}}{=}\mathbb{E}[\varphi(Y,Z)|Y],
\end{equation*}
where
\begin{align*}
Y&=\sqrt{\frac{1-\mathrm{e}^{-2\kappa\delta}}{2\kappa}}\,U_0
\sim\pm\int_0^\delta\mathrm{e}^{-\kappa s}\mathrm{d}W_s,\\
Z_1&=\sqrt{\frac{1-\mathrm{e}^{-2\kappa(T_1-\delta)}}{2\kappa}}\,U_1
\sim\pm\int_\delta^{T_1}\mathrm{e}^{-\kappa s}\mathrm{d}W_s,\\
Z_i&=\sqrt{\frac{1-\mathrm{e}^{-2\kappa\Delta T_i}}{2\kappa}}\,U_i
\sim\pm\int_{T_{i-1}}^{T_i}\mathrm{e}^{-\kappa s}\mathrm{d}W_s,
\quad 2\leq i\leq d-1,\\
\varphi(y,z)&=\bar{N}\sigma\sum_{i=2}^d\mathrm{e}^{-rT_i}\Delta T_i\,\mathrm{e}^{\kappa T_{i-1}}\bigg(y+\sum_{j=1}^{i-1}z_j\bigg),
\quad y\in\mathbb{R},z=(z_1,\dots,z_{d-1})\in\mathbb{R}^{d-1},
\end{align*}
with $\{U_i,0\leq i\leq d-1\}\stackrel{\text{\rm\tiny i.i.d.}}{\sim}\mathcal{N}(0,1)$.
On this basis, $X_h$ can be simulated as
\begin{equation*}
X_h=\frac1K\sum_{k=1}^K\varphi(Y,Z^{(k)}),
\end{equation*}
where $h=\frac1K\in\mathcal{H}$ and $\{Z^{(k)},1\leq k\leq K\}\stackrel{\text{\rm\tiny i.i.d.}}{\sim}Z$ are independent from $Y$.

\subsection{Numerical Results}

We aim to compare numerically the asymptotic error distributions of the schemes introduced in Sections~\ref{sec:nested}--\ref{sec:mlsa:avg}.
We fix a small prescribed accuracy $\varepsilon\in(0,1)$, set $h$ as in Theorems~\ref{cost:nsa} and~\ref{cost:ansa} for the NSA and ANSA schemes and set $L$ and $\mathbf{N}=\{N_\ell,0\leq\ell\leq L\}$ as in Theorems~\ref{thm:cost:ml} and~\ref{thm:cost:aml} for the MLSA and AMLSA schemes.
We run each scheme $5000$ times and plot the corresponding joint distributions of the renormalized $(\text{VaR},\text{ES})$ estimation errors as in Corollaries~\ref{crl:nested:clt} and~\ref{crl:avg:nested:clt} and Theorems~\ref{thm:ml:clt} and~\ref{thm:avg:ml:clt}.
The unbiased SA and averaged SA (ASA) schemes of~\cite[Theorem~3.4]{10.1214/21-EJP648} and~\cite[Theorem~2.4]{BardouFrikhaPages+2009+173+210}, that are based on simulating the loss $X_0$ exactly via \eqref{eq:X0=Phi(Y):bis}, are also run for benchmarking purposes.
But, unlike~\cite[Theorem~2.4]{BardouFrikhaPages+2009+173+210}, we do not average out the ES component. Qualitatively however, the resulting outputs should be comparable.

For the case study, we set $S_0=1$, $r=2\%$, $\kappa=12\%$, $\sigma=20\%$, $T=1\,\mathrm{year}$, $\Delta T_i=3\,\mathrm{months}$, $\delta=7\;\mathrm{days}$ and $\alpha=85\%$. We use a $30/360$ day count fraction convention.
\eqref{eq:xi:*} and \eqref{eq:chi:*} yield $\xi^0_\star\approx2.19$ and $\chi^0_\star\approx3.29$.
The biased risk measures $\xi^{h_L}_\star$ and $\chi^{h_L}_\star$ needed for the CLTs of Theorems~\ref{thm:ml:clt} and~\ref{thm:avg:ml:clt} are computed by averaging out $200$ outcomes of the NSA scheme with bias $h=h_L$ and $n=10^5$ iterations, yielding $\xi^{h_L}_\star\approx2.17$ and $\chi^{h_L}_\star\approx3.41$.
We set $\varepsilon=\frac1{256}$ and $\beta=0.9$ for all six SA schemes.
The choice of $\beta$ is such that it is in the interval $\big(\frac89,1\big)$ for which Theorem~\ref{thm:avg:ml:clt} is proven. It is equal for all schemes to allow comparison.
We adopt $\{\gamma_n=1\times n^{-0.9},n\geq1\}$ for the unbiased SA and ASA algorithms, $\{\gamma_n=0.1\times(250+n)^{-0.9},n\geq1\}$ for the NSA and ANSA algorithms, and $h_0=\frac1{32}$, $M=2$, $L=3$ and $\{\gamma_n=0.1\times(1500+n)^{-0.9},n\geq1\}$ for the MLSA and AMLSA algorithms.
We also compute the theoretical ES variance factors given in the CLTs via the Monte Carlo formulas in Remarks~\ref{rmk:nsa:cv}(\ref{rmk:nsa:cv:iv}) and~\ref{rmk:ml}(\ref{rmk:sim}).

For each algorithm, we fit a bivariate Gaussian pdf on the renormalized estimation errors. The parameters of the fitted Gaussian laws are reported in Table~\ref{tbl:gaussian}.
We use the obtained Gaussian covariance matrices to derive confidence ellipses at $95\%$, using the formulas given in~\cite{errorellipses}.
Figure~\ref{fig:euption} plots, for each algorithm, a $95\%$-confidence ellipse and Gaussian marginals of the joint distributions of the renormalized $(\text{VaR},\text{ES})$ estimation errors.
We also plot in black crosses the ES Gaussian laws estimated by Monte Carlo, and report the corresponding variances in Table~\ref{tbl:fit:vs:MC}.

\begin{figure}[H]
\centering
\begin{minipage}{\textwidth}
\centerline{
\begin{subfigure}{.475\textwidth}
\flushright
\includegraphics[width=.975\textwidth]{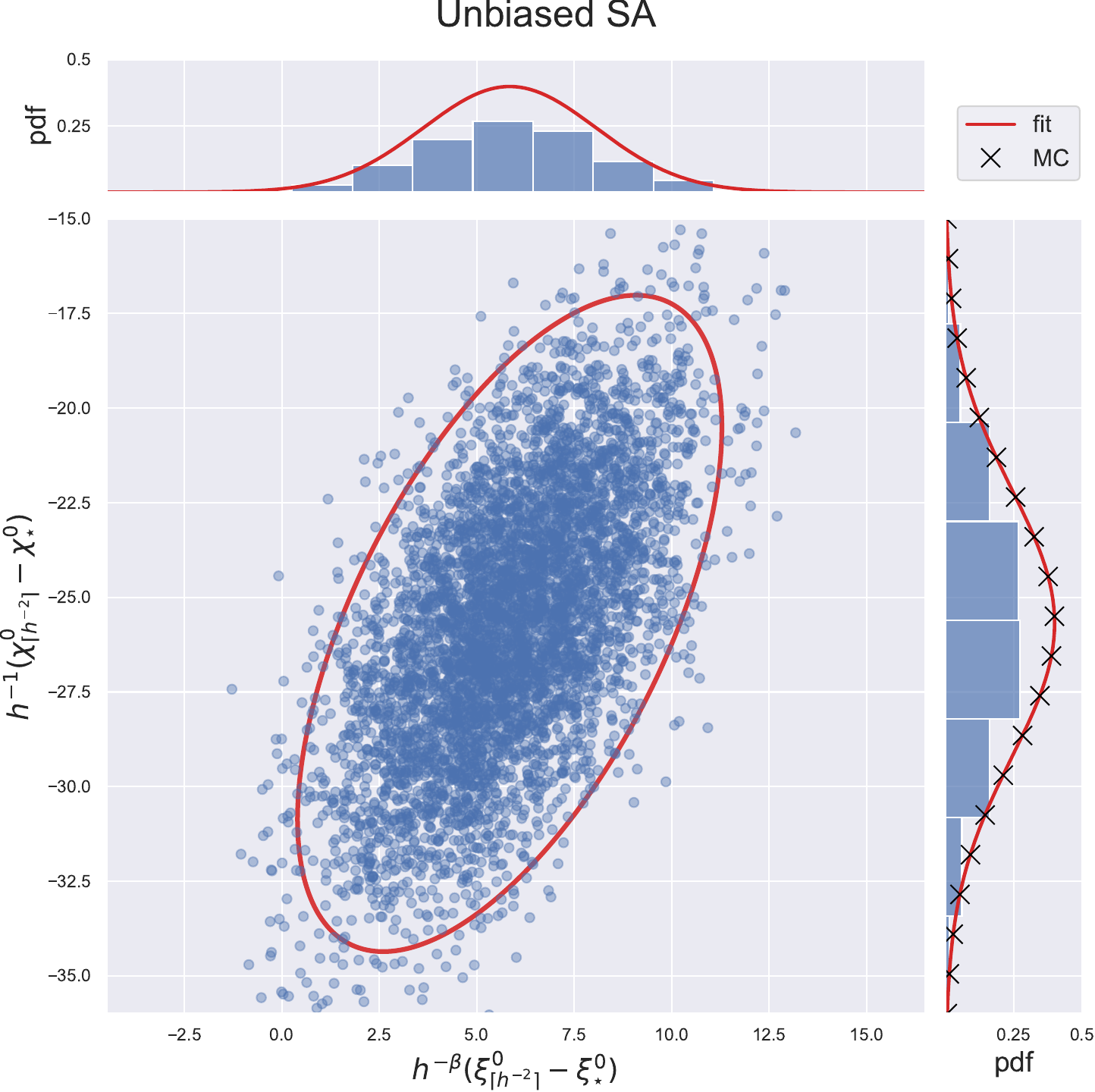}
\end{subfigure}

\hfill

\begin{subfigure}{.475\textwidth}
\flushleft
\includegraphics[width=.975\textwidth]{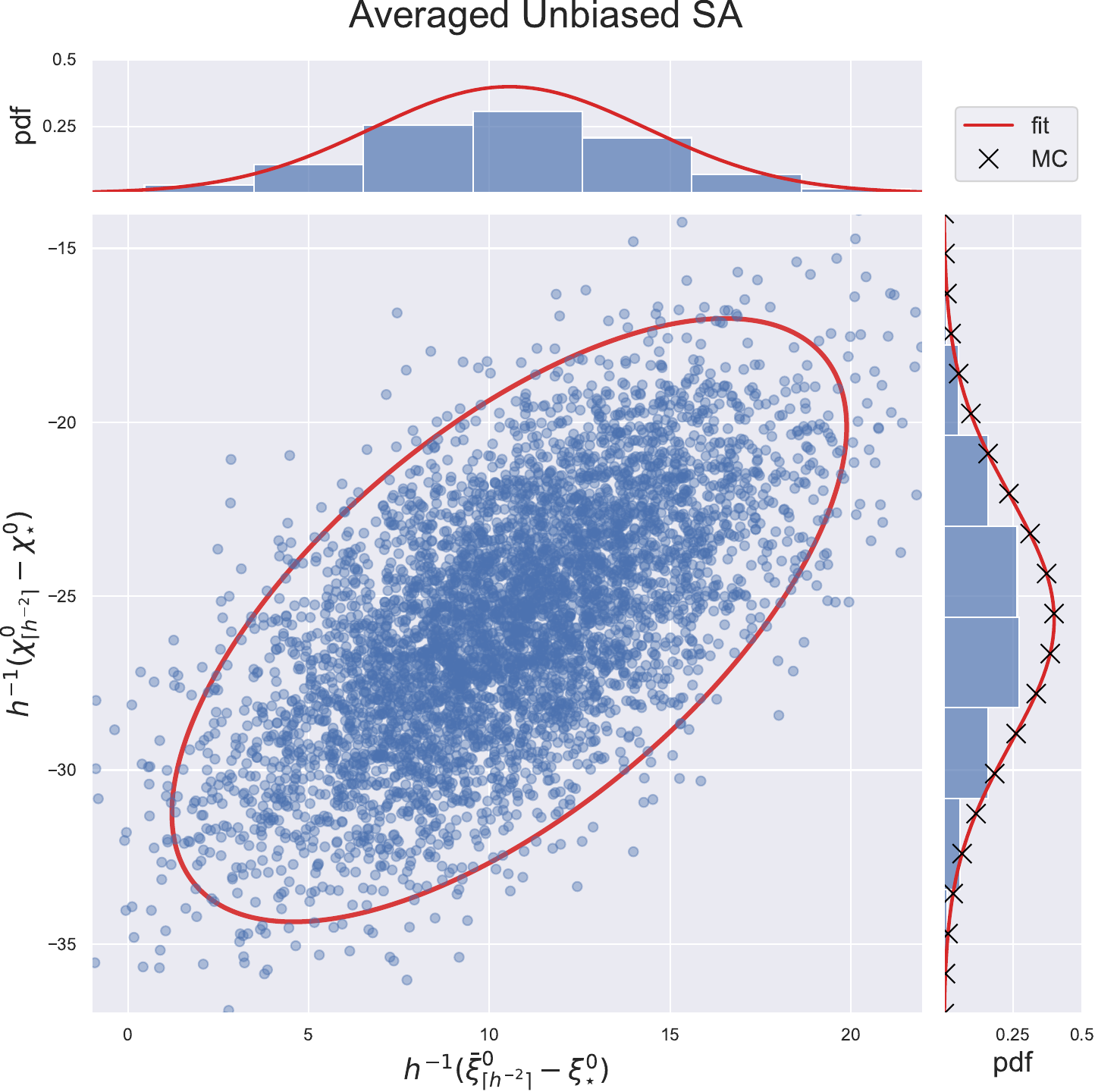}
\end{subfigure}
}
\end{minipage}

\bigskip

\begin{minipage}{\textwidth}
\centerline{
\begin{subfigure}{.475\textwidth}
\flushright
\includegraphics[width=.975\textwidth]{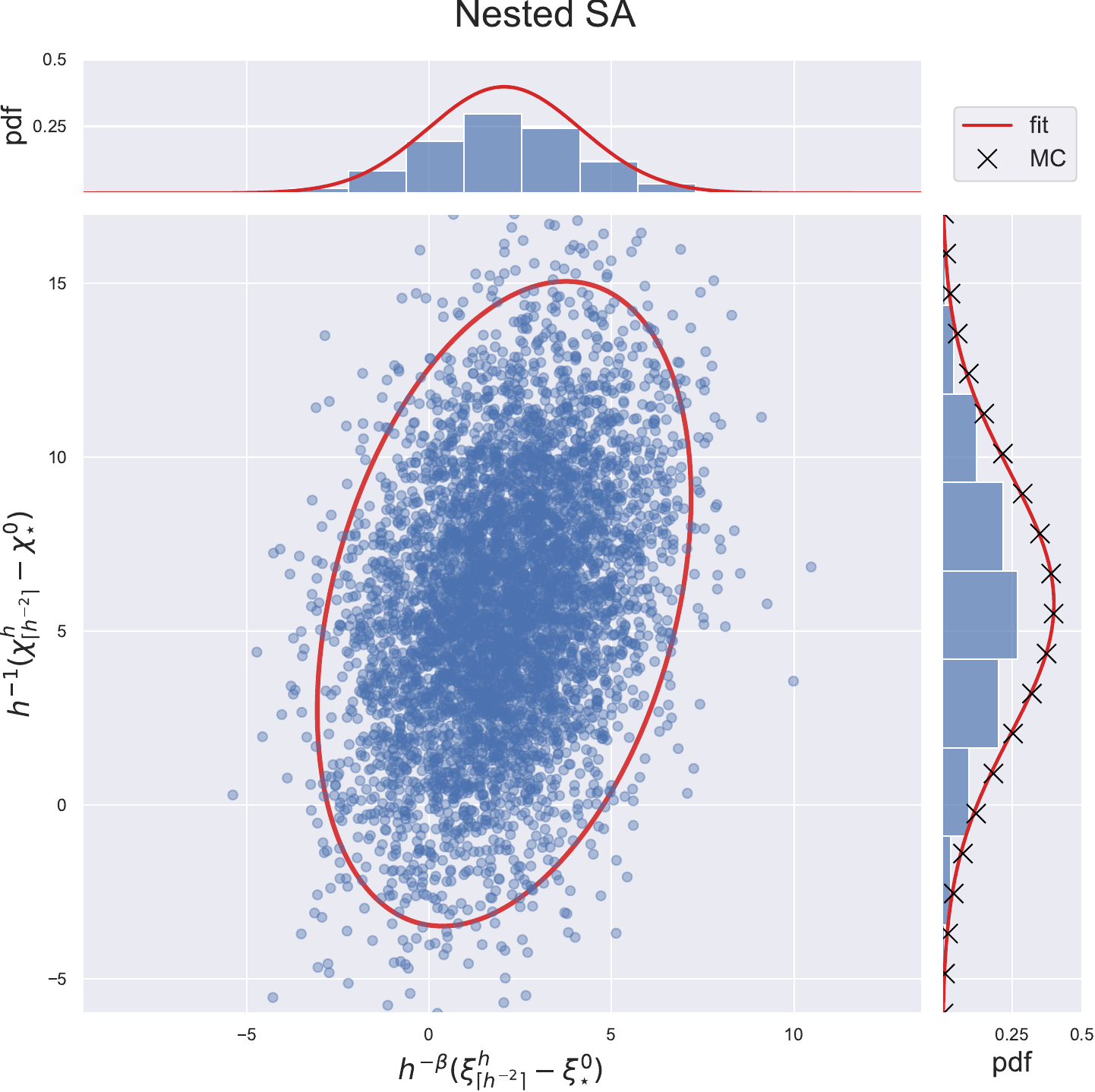}
\end{subfigure}

\hfill

\begin{subfigure}{.475\textwidth}
\flushleft
\includegraphics[width=.975\textwidth]{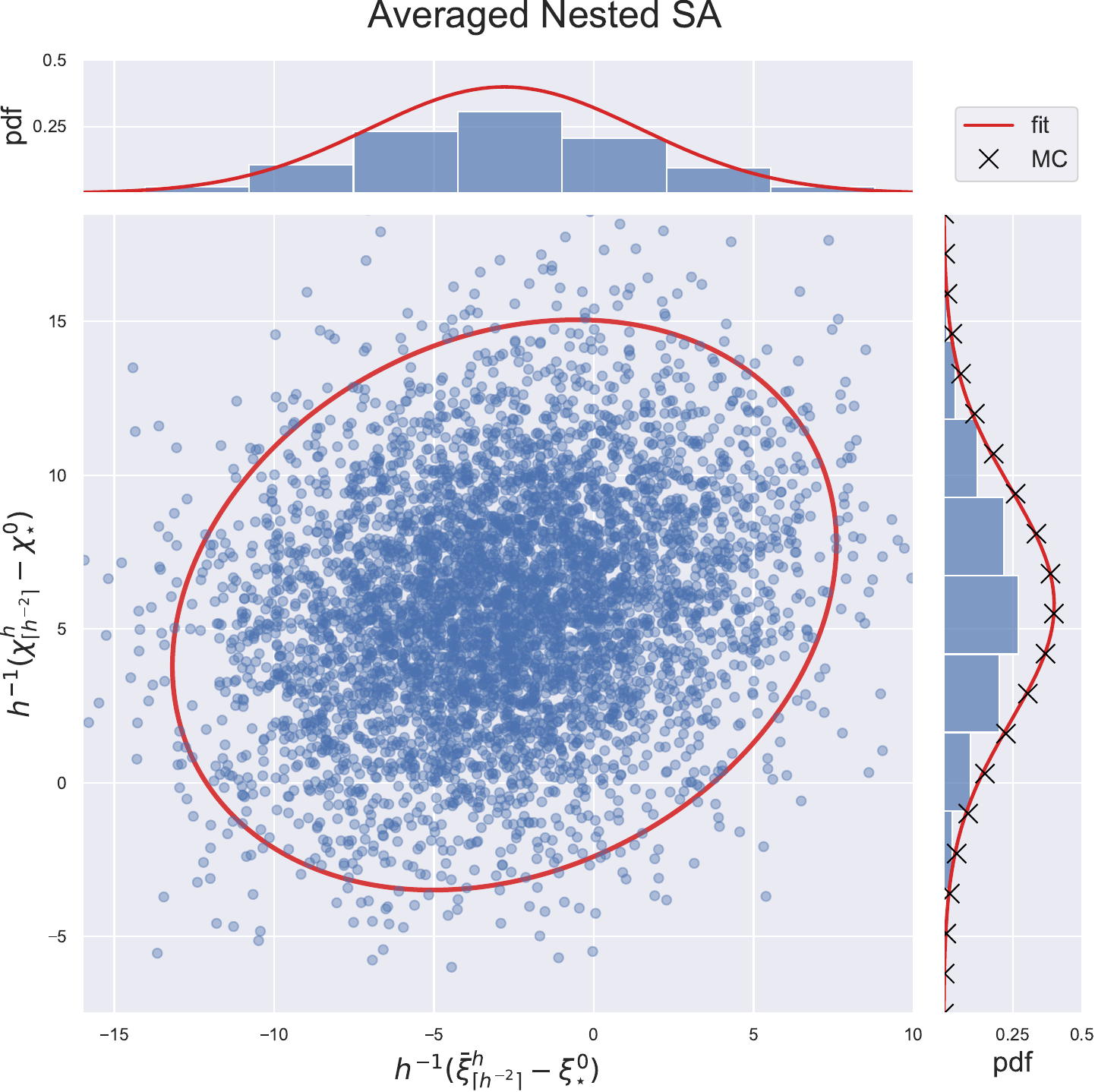}
\end{subfigure}
}
\end{minipage}

\bigskip

\begin{minipage}{\textwidth}
\centerline{
\begin{subfigure}{.475\textwidth}
\flushright
\includegraphics[width=.975\textwidth]{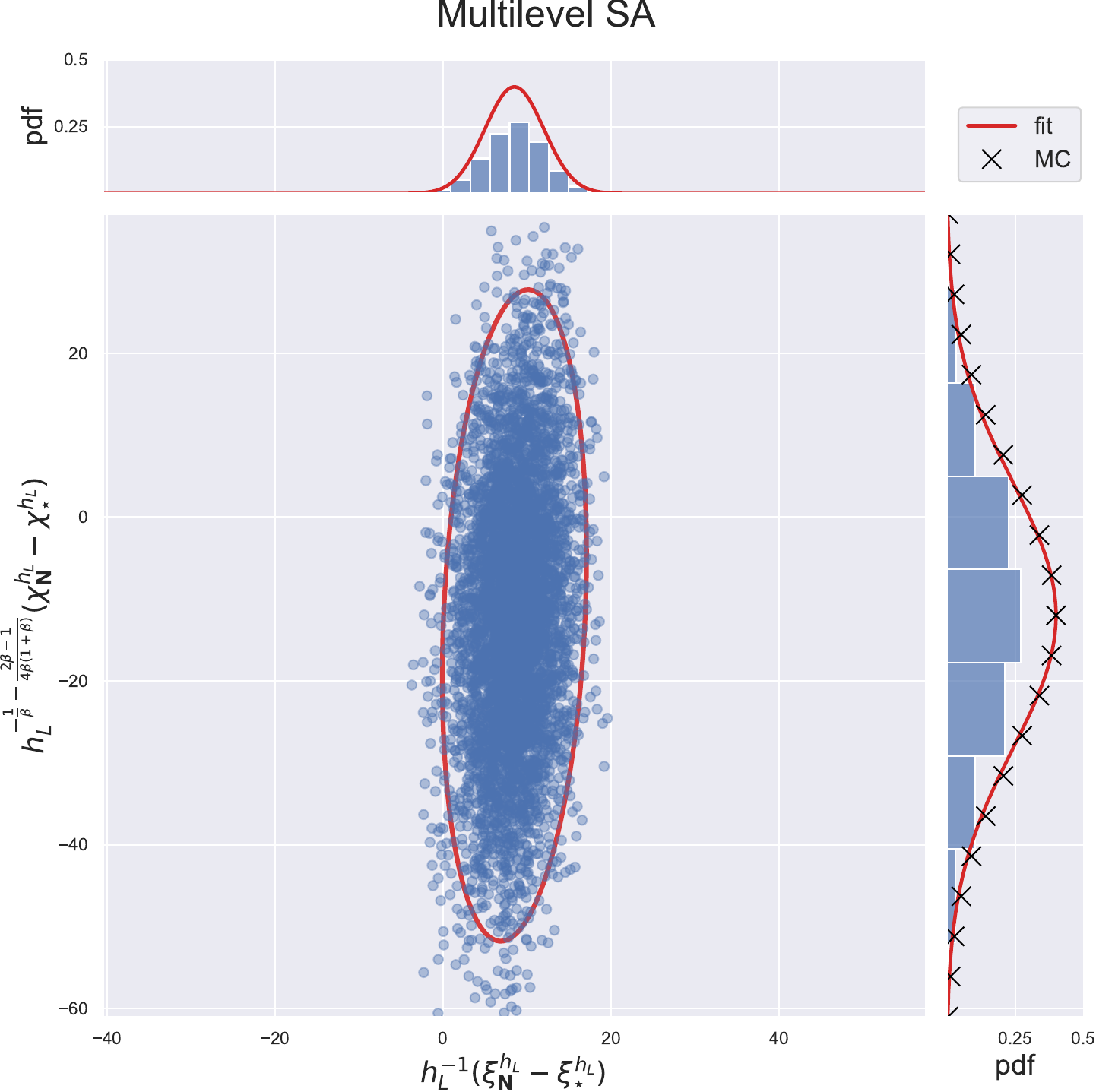}
\end{subfigure}

\hfill

\begin{subfigure}{.475\textwidth}
\flushleft
\includegraphics[width=.975\textwidth]{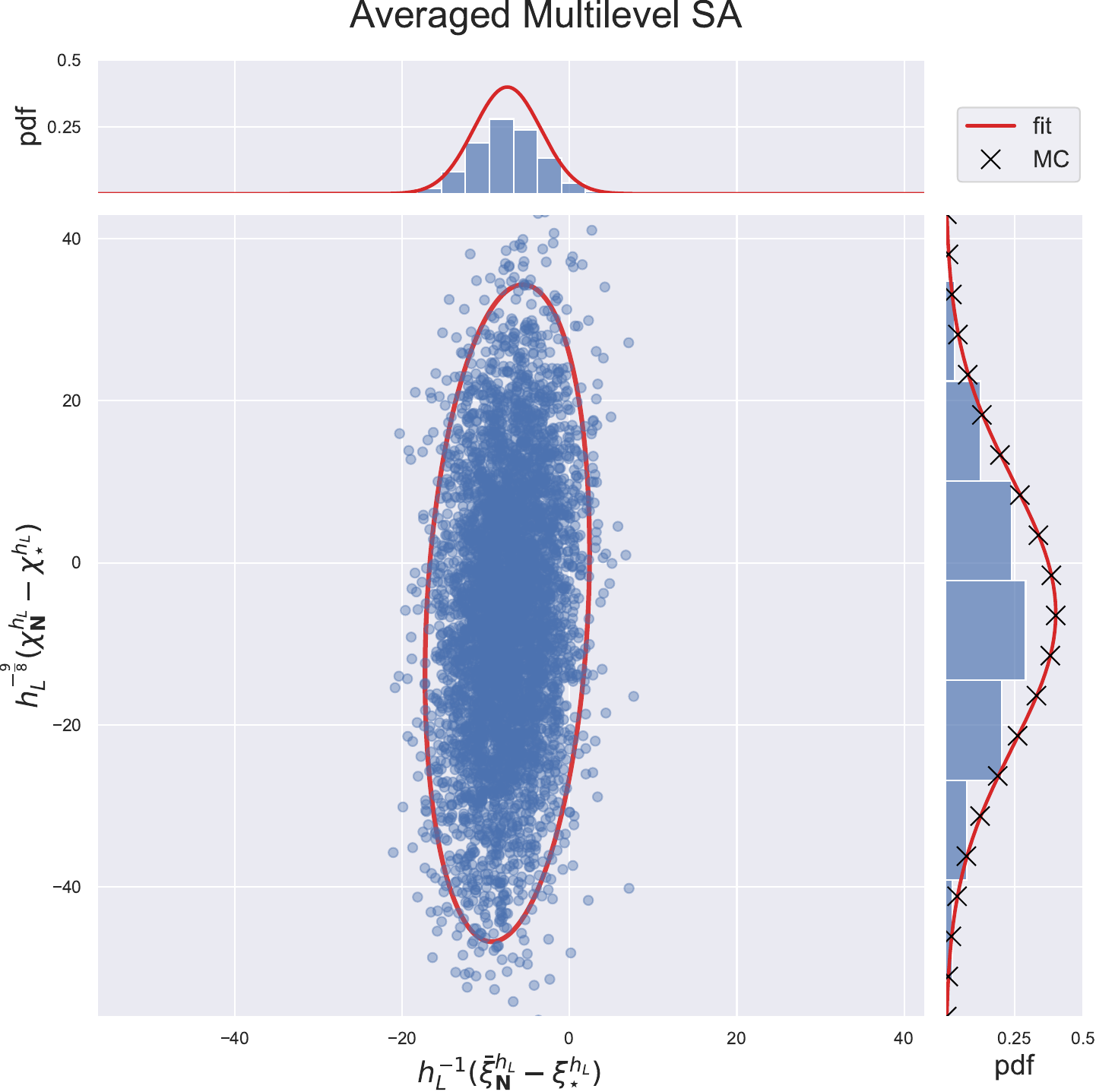}
\end{subfigure}
}
\end{minipage}
\caption{Joint distributions of the renormalized VaR and ES estimation errors.}
\label{fig:euption}
\end{figure}

\begin{table}[H]
    \centering
    \begin{tabular}{|c|c|}
        \hline
        Unbiased SA & Averaged Unbiased SA \\
         $\mu=\begin{pmatrix}
            5.84\\
            -25.69
         \end{pmatrix}$,
         $\Sigma=\begin{pmatrix}
             4.94 & 4.70 \\
             4.70 & 12.55
         \end{pmatrix}$
         & $\mu=\begin{pmatrix}
             10.55\\
             -25.69
         \end{pmatrix}$,
         $\Sigma=\begin{pmatrix}
            14.54 & 8.65 \\
            8.65 & 12.55
         \end{pmatrix}$\\
         \hline
         Nested SA & Averaged Nested SA \\
         $\mu=\begin{pmatrix}
             2.06\\
             5.78
         \end{pmatrix}$,
         $\Sigma=\begin{pmatrix}
             4.38 & 2.62 \\
             2.62 & 14.33
         \end{pmatrix}$
         & $\mu=\begin{pmatrix}
             -2.80\\
             5.78
         \end{pmatrix}$,
         $\Sigma=\begin{pmatrix}
              18.00 & 3.46 \\
              3.46 & 14.34
         \end{pmatrix}$\\
         \hline
         Multilevel SA & Averaged Multilevel SA \\
         $\mu=\begin{pmatrix}
             8.47\\
             -12.01
         \end{pmatrix}$,
         $\Sigma=\begin{pmatrix}
             12.30 & 10.82 \\
             10.82 & 264.10
         \end{pmatrix}$
         & $\mu=\begin{pmatrix}
             -7.44\\
             -6.22
         \end{pmatrix}$,
         $\Sigma=\begin{pmatrix}
             16.13 & 12.65 \\
             12.65 & 274.30
         \end{pmatrix}$\\
         \hline
    \end{tabular}
    \caption{Empirical means ($\mu$) and covariances ($\Sigma$) fitted on the renormalized VaR-ES estimation errors in Figure~\ref{fig:euption}.}
    \label{tbl:gaussian}
\end{table}

\begin{table}[H]
    \centering
    \begin{tabular}{|c!{\vrule width 1pt}c|c|c|c|c|c|}
        \hline
        Scheme & SA & ASA & NSA & ANSA & MLSA & AMLSA \\
        \specialrule{0.1em}{0em}{0em}
        Fitted & 12.55 & 12.55 & 14.34 & 14.34 & 264.10 & 274.30 \\
        \hline
        MC     & 12.61 & 12.61 & 15.28 & 15.28 & 292.50 & 269.39 \\
        \hline
    \end{tabular}
    \caption{Comparison of empirical and Monte Carlo estimations of the ES variance.}
    \label{tbl:fit:vs:MC}
\end{table}

Qualitatively, the fitted joint distributions appear centrosymmetric and unimodal for all algorithms, suggesting a Gaussian behavior in line with the CLTs proven in the previous sections and with~\cite[Theorem~3.4]{10.1214/21-EJP648} and~\cite[Theorem~2.4]{BardouFrikhaPages+2009+173+210}. This Gaussianity is further endorsed by the marginals' histograms that match the overlayed bell curves.
Unlike Theorem~\ref{thm:nested:clt}, the axes tilt of the fitted ellipse on NSA suggests some asymptotic correlation between the VaR and ES. This may be due to the choice $\beta=0.9$ that is too close $1$, where correlation theoretically exists.
By contrast, consistent with the Theorems~\ref{thm:avg:nested:clt}, \ref{thm:ml:clt} and~\ref{thm:avg:ml:clt}, the ANSA algorithm shows an asymptotic correlation between the VaR and ES, while the MLSA and AMLSA algorithms display barely any.
Finally, as noted for the NSA algorithm, the unbiased SA and ASA algorithms present some asymptotic correlation, that may stem from the closeness of the pick $\beta=0.9$ to $1$, for which correlation occurs.
Numerically, the averaged algorithms turned out significantly more stable than the non-averaged ones. In particular, they required much less time to parametrize the step size initialization $\gamma_1$.

Quantitatively, the off-diagonal correlation term may not be null as Theorems~\ref{thm:nested:clt}, \ref{thm:ml:clt} and~\ref{thm:avg:ml:clt} show, but it remains of low magnitude with respect to the diagonal values.
The biases in MLSA and AMLSA may reflect unattained asymptotic regime, but they remain also an order of magnitude lower than the covariance diagonal terms. The VaR MLSA variance is small due to the choice of a learning rate with small initialization $\gamma_1$, as predicted by Theorem~\ref{thm:ml:clt}.
Table~\ref{tbl:fit:vs:MC} demonstrates that the ES variances predicated by our CLTs match quite accurately the empirically observed ones.

\section*{Conclusion}

In~\cite{mlsa}, a nested stochastic approximation algorithm for VaR and ES estimation, as well as a multilevel acceleration thereof, were presented and compared in terms of non-asymptotic $L^2(\mathbb{P})$-errors.
The present article complements the aforementioned one by analyzing the corresponding asymptotic error distributions, as required for delimiting VaR and ES trust regions and confidence intervals.
Further averaged extensions of these algorithms are also presented and shown to achieve better convergence rates than their original counterparts.
A financial case study where exact VaR and ES values as well as unbiased SA schemes are available for benchmarking purposes, supports our theoretical findings.

All things considered, the analyses in this paper hint at resorting to a multilevel scheme rather than a simply nested one in order to reduce complexity, and at applying an averaged scheme rather than a non-averaged one in order to increase numerical stability.
For some prescribed accuracy $\varepsilon>0$, the optimal complexity attained by the presented multilevel algorithms is in $\OO(\varepsilon^{-\frac52})$.
Given a learning rate $\{\gamma_n=\gamma_1n^{-\beta},n\geq1\}$, $\gamma_1>0$, $\beta\in\big(\frac12,1\big]$, this complexity is achieved by MLSA for $\beta=1$ under the constraint $\lambda\gamma_1>1$, where the constant $\lambda>0$ in \eqref{e:thelambda} is explicit but tedious to compute.
This complexity is also achieved by AMLSA for $\beta\in\big(\frac89,1\big)$, without any constraint on $\gamma_1$.

However, $\OO(\varepsilon^{-\frac52})$ remains higher than the theoretical optimum of $\OO(\varepsilon^{-2})$ that multilevel algorithms are known for~\cite{10.1214/15-AAP1109,10.1287/opre.1070.0496}.
This gap in performance stems from the discontinuity of the gradient used in the updating formula \eqref{eq:sa:nested:alg:xi}, registering an $\OO(1)$ error whenever the generated loss $X_h$ is too close to the estimate $\xi^h_n$ but falls on the opposite side of the discontinuity with respect to the simulation target $X_0$~\cite{doi:10.1137/21M1447064,giles2023efficient}.
This limitation should be addressed in future research.

\appendix

\section{Auxiliary Results}

\begin{lemma}
\label{lmm:limsup}
Let $\{\gamma_n=\gamma_1n^{-\beta},n\geq1\}$, $\gamma_1>0$, $\beta\in(0,1]$.
\begin{enumerate}[(i)]
\item\label{lmm:limsup:i}
Let $\{e_n,n\geq1\}$ be a non-negative sequence. For $b\geq0$ and $\lambda>0$, if $\lambda\gamma_1>b$ when $\beta=1$, then
\begin{equation*}
\limsup_{n\to\infty}\gamma_n^{-b}\sum_{k=1}^n\gamma_k^{1+b}e_k\exp\bigg(-\lambda\sum_{j=k+1}^n\gamma_j\bigg)
\leq\frac1C\limsup_{n\to\infty}e_n,
\end{equation*}
where
\begin{equation*}
C\coloneqq
\begin{cases}
   \lambda-b/\gamma_1
      &\text{if $\beta=1$ and $\lambda\gamma_1>b$,}\\
   \lambda
      &\text{if $\beta\in(0,1)$.}
\end{cases}
\end{equation*}
\item\label{lmm:limsup:ii}
For $b\geq0$ and $\lambda>0$, if $\lambda\gamma_1>b$ when $\beta=1$, then
\begin{equation*}
\limsup_{n\to\infty}{\gamma_n^{-b}\exp\bigg(-\lambda\sum_{j=1}^n\gamma_j\bigg)}=0.
\end{equation*}
\end{enumerate}
\end{lemma}
Lemma~\ref{lmm:limsup}(\ref{lmm:limsup:i}) is a special case of~\cite[Lemma~5.9]{fort2015central} and Lemma~\ref{lmm:limsup}(\ref{lmm:limsup:ii}) is a special case of Lemma~\ref{lmm:limsup}(\ref{lmm:limsup:i}).

\begin{lemma}[{\cite[Lemma~3.2 \& Proposition~5.2]{Giorgi2020}}]
\label{lmm:aux}\
\begin{enumerate}[(i)]
    \item\label{lmm:aux-i}
If $\mathbb{E}[|\varphi(Y, Z)-\mathbb{E}[\varphi(Y,Z)|Y]|^p]<\infty$ for some $p>1$, then
\begin{equation*}
\mathbb{E}[|X_h-X_{h'}|^p]\leq C|h-h'|^\frac{p}2,
\quad
h,h'\in\mathcal{H}\cup\{0\}.
\end{equation*}

    \item
If, additionally, the $X_h$ admit pdf $f_{X_h}$ that are bounded uniformly in $h\in\mathcal{H}\cup\{0\}$, then, for any $\xi\in\mathbb{R}$,
\begin{equation*}
\mathbb{E}\big[\big|\mathds{1}_{X_h>\xi}-\mathds{1}_{X_{h'}>\xi}\big|\big]\leq C|h-h'|^\frac{p}{2(1+p)},
\quad
h,h'\in\mathcal{H}\cup\{0\}.
\end{equation*}
\end{enumerate}
\end{lemma}

\begin{lemma}[{\cite[Proposition~2.4 \& Theorem~2.7(i)]{mlsa}}]
\label{lmm:error}\
\begin{enumerate}[(i)]
    \item\label{lmm:error:weak}
Suppose that $\varphi(Y,Z)\in L^1(\mathbb{P})$, that Assumptions~\ref{asp:misc}(\ref{asp:misc:i}) and~\ref{asp:misc}(\ref{asp:misc:ii}) hold and that the pdf $f_{X_0}$ is positive. Then, as $\mathcal{H}\ni h\downarrow0$,
\begin{equation*}
\xi^h_\star-\xi^0_\star
=-\frac{v(\xi^0_\star)}{f_{X_0}(\xi^0_\star)}h+\oo(h)
,\qquad
\chi^h_\star-\chi^0_\star
=-h\int_{\xi^0_\star}^\infty \frac{v(\xi)}{1-\alpha}\mathrm{d}\xi+\oo(h).
\end{equation*}
    \item\label{lmm:error:statistical}
Suppose that $\varphi(Y,Z)\in L^2(\mathbb{P})$ and that Assumptions~\ref{asp:misc} and~\ref{asp:supE[sup]} hold.
If $\gamma_n=\gamma_1n^{-\beta}$, $n\geq1$, $\gamma_1>0$, $\beta\in(0,1]$,
with $\lambda\gamma_1>1$ if $\beta=1$, then there exists a positive constant $C<\infty$ such that, for any positive integer $n$,
\begin{equation*}
\sup_{h\in\mathcal{H}}\mathbb{E}[(\xi^h_n-\xi^h_\star)^2]\leq C\gamma_n.
\end{equation*}
\end{enumerate}
\end{lemma}

\section{Proof of Theorem~\ref{thm:nested:clt}}
\label{prf:nested:clt}

We denote by $C<\infty$ a positive constant that may change from line to line but does not depend upon $h$ or $n$.
We follow a similar strategy to~\cite[Theorem~2.7]{10.1214/15-AAP1109}.
\\

Let $h\in\mathcal{H}$ and define $\{\mathcal{F}^h_n,n\geq0\}$ as the filtration given by $\mathcal{F}^h_0=\sigma(\xi^h_0,\chi^h_0)$ and $\mathcal{F}^h_n=\sigma(\xi^h_0,\chi^h_0,X_h^{(1)},\dots,X_h^{(n)})$, $n\geq1$.

The sequence $\{\xi^h_n,n\geq0\}$ of dynamics \eqref{eq:sa:nested:alg:xi} decomposes into
\begin{equation}
\label{decomposition:var:sa}
\xi^h_n-\xi^h_\star
=\big(1-\gamma_nV_0''(\xi^0_\star)\big)(\xi^h_{n-1}-\xi^h_\star)+\gamma_ng^h_n+\gamma_nr^h_n+\gamma_n\rho^h_n+\gamma_ne^h_n,
\end{equation}
where
\begin{align}
g^h_n
&=\big(V_0''(\xi^0_\star)-V_h''(\xi^h_\star)\big)(\xi^h_{n-1}-\xi^h_\star),
\label{eq:ghn}\\
r^h_n
&=V_h''(\xi^h_\star)(\xi^h_{n-1}-\xi^h_\star)-V_h'(\xi^h_{n-1}),
\label{eq:rhn}\\
\rho^h_n
&=V_h'(\xi^h_{n-1})-V_h'(\xi^h_\star)-\big(H_1(\xi^h_{n-1},X_h^{(n)})-H_1(\xi^h_\star,X_h^{(n)})\big),
\label{eq:rhohn}\\
e^h_n
&=V_h'(\xi^h_\star)-H_1(\xi^h_\star,X_h^{(n)})=-H_1(\xi^h_\star,X_h^{(n)}).
\label{eq:ehn}
\end{align}
Hence, unrolling \eqref{decomposition:var:sa},
\begin{equation}
\label{eq:xih-xi*}
\begin{aligned}
\xi^h_n-\xi^h_\star=
(\xi^h_0&-\xi^h_\star)\Pi_{1:n}
+\sum_{k=1}^n\gamma_k\Pi_{k+1:n}g^h_k\\
&+\sum_{k=1}^n\gamma_k\Pi_{k+1:n}r^h_k
+\sum_{k=1}^n\gamma_k\Pi_{k+1:n}\rho^h_k
+\sum_{k=1}^n\gamma_k\Pi_{k+1:n}e^h_k,
\end{aligned}
\end{equation}
where
\begin{equation}
\Pi_{k:n}=\prod_{j=k}^n\big(1-\gamma_jV_0''(\xi^0_\star)\big),
\label{eq:Pi}
\end{equation}
with the convention $\prod_\varnothing=1$.

Likewise, the sequence $\{\chi^h_n,n\geq0\}$ of dynamics \eqref{eq:sa:nested:alg:C} rewrites
\begin{equation}
\label{eq:chih-chi*}
\begin{aligned}
\chi^h_n-\chi^h_\star
&=\frac1n\sum_{k=1}^n\Big(\xi^h_{k-1}+\frac1{1-\alpha}(X_h^{(k)}-\xi^h_{k-1})^+\Big)-V_h(\xi^h_\star)\\
&=\frac1n\sum_{k=1}^n\theta^h_k+\frac1n\sum_{k=1}^n\zeta^h_k+\frac1n\sum_{k=1}^n\eta^h_k,
\end{aligned}
\end{equation}
where
\begin{align}
\theta^h_k
&=\xi_{k-1}^h-\xi^h_\star
+\frac1{1-\alpha}\big((X_h^{(k)}-\xi_{k-1}^h)^+-(X_h^{(k)}-\xi^h_\star)^+\big)-\big(V_h(\xi_{k-1}^h)-V_h(\xi^h_\star)\big),
\label{eq:thetahk}\\
\zeta^h_k
&=V_h(\xi_{k-1}^h)-V_h(\xi^h_\star),\nonumber\\
\label{eq:etahk}
\eta^h_k
&=\frac1{1-\alpha}\big((X_h^{(k)}-\xi^h_\star)^+-\mathbb{E}[(X_h^{(k)}-\xi^h_\star)^+]\big).
\end{align}

We first provide a useful upper bound on $\Pi_{k:n}$. Since $\lim_{k\to\infty}{\gamma_k}=0$, there exists $k_0\geq0$ such that $(1-\gamma_jV_0''(\xi^0_\star))>0$, $j\geq k_0$.
Thus, using the inequality $1+x\leq\exp(x)$, $x\in\mathbb{R}$, we obtain that for $n$ large enough,
\begin{equation}
\label{upper:estimate:pi:i:n}
\begin{aligned}
|\Pi_{k:n}|
&=|\Pi_{k:k_0-1}|\prod_{j=k_0 \vee k}^n\big(1-\gamma_jV_0''(\xi^0_\star)\big)\\
&\leq|\Pi_{k:k_0-1}|\exp\bigg(-V_0''(\xi^0_\star)\sum_{j=k_0 \vee k}^n\gamma_j\bigg)
\leq K\exp\bigg(-V_0''(\xi^0_\star)\sum_{j=k}^n\gamma_j\bigg)
\end {aligned}
\end{equation}
where $K=1\vee\max_{1\leq k\leq k_0}|\Pi_{k:k_0-1}|\exp\big(V_0''(\xi^0_\star)\sum_{j=k_0\wedge k}^{k-1}\gamma_j\big)$, with the convention $\sum_\varnothing=0$.

Set $n=\lceil h^{-2}\rceil$.
We study below the contribution of each term on the right hand sides of \eqref{eq:xih-xi*} and \eqref{eq:chih-chi*} to the asymptotic estimation error.
\\

\noindent
\emph{\textbf{Step~1. Study of $\big\{h^{-\beta}(\xi^h_0-\xi^h_\star)\Pi_{1:\lceil h^{-2}\rceil},h\in\mathcal{H}\big\}$.}}
\newline\newline
By assumption and via \eqref{e:thelambda},
\begin{equation}\label{gamma:condition:beta:equal:1}
2\gamma_1V_0''(\xi^0_\star)
\geq2\gamma_1\inf_{h\in\mathcal{H}\cup\{0\}}V_h''(\xi^h_\star)
\geq\frac{16}3\gamma_1\inf_{h\in\mathcal{H}\cup\{0\}}\lambda_h
\geq\gamma_1\lambda>1
\quad\mbox{ if }\quad
\beta=1.
\end{equation}
We deduce via the inequality $\gamma_1^\frac12\gamma_{\lceil h^{-2}\rceil}^{-\frac12}\geq h^{-\beta}$, \eqref{gamma:condition:beta:equal:1}, \eqref{upper:estimate:pi:i:n} and Lemma~\ref{lmm:limsup}(\ref{lmm:limsup:ii}) that
\begin{equation*}
\limsup_{\mathcal{H}\ni h\downarrow0}h^{-\beta}|\Pi_{1:\lceil h^{-2}\rceil}|
\leq C\gamma_1^\frac12\limsup_{n\uparrow\infty}\gamma_n^{-\frac12}\mathrm{e}^{-V_0''(\xi^0_\star)\sum_{j=1}^n\gamma_j}=0.
\end{equation*}

Besides, since $\lim_{\mathcal{H}\ni h\downarrow0}\xi^h_\star=\xi^0_\star$, $\{\xi^h_\star,h\in\mathcal{H}\}$ is bounded, recalling Assumption~\ref{asp:supE[sup]}, $\sup_{h\in\mathcal{H}}{\mathbb{E}[|\xi^h_0-\xi^h_\star|]}\leq\sup_{h\in\mathcal{H}}\mathbb{E}[|\xi^h_0|]+\sup_{h\in\mathcal{H}}{|\xi^h_\star|}<\infty$.

All in all,
\begin{equation*}
h^{-\beta}(\xi^h_0-\xi^h_\star)\Pi_{1:\lceil h^{-2}\rceil}\stackrel{L^1(\mathbb{P})}{\longrightarrow}0
\quad\mbox{ as }\quad
\mathcal{H}\ni h\downarrow0.
\end{equation*}

\noindent
\emph{\textbf{Step~2. Study of $\big\{h^{-\beta}\sum_{k=1}^{\lceil h^{-2}\rceil}\gamma_k\Pi_{k+1:\lceil h^{-2}\rceil}g^h_k,h\in\mathcal{H}\big\}$.}}
\newline\newline
By Lemma~\ref{lmm:error}(\ref{lmm:error:statistical}) and \eqref{upper:estimate:pi:i:n},
\begin{equation*}
\mathbb{E}\bigg[\bigg|\sum_{k=1}^n\gamma_k\Pi_{k+1:n}g^h_k\bigg|\bigg]
\leq C|V_0''(\xi^0_\star)-V_h''(\xi^h_\star)|\sum_{k=1}^n\gamma_k^\frac32\mathrm{e}^{-V_0''(\xi^0_\star)\sum_{j=k+1}^n\gamma_j}.
\end{equation*}
Recalling that $h^{-\beta}\leq\gamma_1^\frac12\gamma_{\lceil h^{-2}\rceil}^{-\frac12}$ and that $2\gamma_1V_0''(\xi^0_\star)>1$ if $\beta=1$ by \eqref{gamma:condition:beta:equal:1}, Lemma~\ref{lmm:limsup}(\ref{lmm:limsup:i}) yields
\begin{equation*}
\limsup_{\mathcal{H}\ni h\downarrow0}h^{-\beta}\sum_{k=1}^{\lceil h^{-2}\rceil}\gamma_k^\frac32\mathrm{e}^{-V_0''(\xi^0_\star)\sum_{j=k+1}^{\lceil h^{-2}\rceil}\gamma_j}
\leq C\gamma_1^\frac12\limsup_{n\uparrow\infty}\gamma_n^{-\frac12}\sum_{k=1}^n\gamma_k^\frac32\mathrm{e}^{-V_0''(\xi^0_\star)\sum_{j=k+1}^n\gamma_j}
\leq C.
\end{equation*}
Besides, by Assumption~\ref{asp:misc}(\ref{asp:misc:ii}), $\big\{V_h''=(1-\alpha)^{-1}f_{X_h},h\in\mathcal{H}\big\}$ converges locally uniformly to $V_0''$ as $\mathcal{H}\ni h\downarrow0$. Moreover, by Lemma~\ref{lmm:error}(\ref{lmm:error:weak}), $\lim_{\mathcal{H}\ni h\downarrow0}\xi^h_\star=\xi^0_\star$. Thus $\lim_{\mathcal{H}\ni h\downarrow0}V_h''(\xi^h_\star)=V_0''(\xi^0_\star)$.
Finally,
\begin{equation*}
h^{-\beta}\sum_{k=1}^{\lceil h^{-2}\rceil}\gamma_k\Pi_{k+1:\lceil h^{-2}\rceil}g^h_k
\stackrel{L^1(\mathbb{P})}{\longrightarrow}0
\quad\mbox{ as }\quad
\mathcal{H}\ni h\downarrow0.
\end{equation*}

\noindent
\emph{\textbf{Step~3. Study of $\big\{h^{-\beta}\sum_{k=1}^{\lceil h^{-2}\rceil}\gamma_k\Pi_{k+1:\lceil h^{-2}\rceil}r^h_k,h\in\mathcal{H}\big\}$.}}
\newline\newline
Using that $V_h'(\xi^h_\star)=0$, by a first order Taylor expansion,
\begin{equation}
\label{eq:Vh':taylor:1}
V_h'(\xi)=V_h''(\xi^h_\star)(\xi-\xi^h_\star)
+(\xi-\xi^h_\star)\int_0^1\big(V_h''\big(\xi^h_\star+t(\xi-\xi^h_\star)\big)-V_h''(\xi^h_\star)\big)\mathrm{d}t.
\end{equation}
The uniform Lipschitz regularity of $\{V_h''=(1-\alpha)^{-1}f_{X_h},h\in\mathcal{H}\}$, by Assumption~\ref{asp:misc}(\ref{asp:misc:iv}), and \eqref{eq:Vh':taylor:1} yield
\begin{equation*}
\mathbb{E}[|r^h_k|]
\leq\frac{\sup_{h'\in\mathcal{H}}[f_{X_{h'}}]_\text{\rm Lip}}{2(1-\alpha)}\,\mathbb{E}[(\xi_{k-1}^h-\xi^h_\star)^2],
\end{equation*}
so that, by Lemma~\ref{lmm:error}(\ref{lmm:error:statistical}),
\begin{equation}\label{eq:E[|r|]<}
\mathbb{E}[|r^h_k|]\leq C\gamma_k.
\end{equation}
Recalling that $2\gamma_1V_0''(\xi^0_\star)>1$ if $\beta=1$ by \eqref{gamma:condition:beta:equal:1}, using \eqref{upper:estimate:pi:i:n} and Lemma~\ref{lmm:limsup}(\ref{lmm:limsup:i}),
\begin{equation*}
\limsup_{\mathcal{H}\ni h\downarrow0}\mathbb{E}\bigg[\bigg|h^{-\beta}\sum_{k=1}^{\lceil h^{-2}\rceil}\gamma_k\Pi_{k+1:{\lceil h^{-2}\rceil}}r^h_k\bigg|\bigg]
\leq C\limsup_{n\uparrow\infty}\gamma_n^{-\frac12}\sum_{k=1}^n\gamma^2_k\mathrm{e}^{-V_0''(\xi^0_\star)\sum_{j=k+1}^n\gamma_j}
=0,
\end{equation*}
i.e.
\begin{equation*}
h^{-\beta}\sum_{k=1}^{\lceil h^{-2}\rceil}\gamma_k\Pi_{k+1:\lceil h^{-2}\rceil}r^h_k
\stackrel{L^1(\mathbb{P})}{\longrightarrow}0
\quad\mbox{ as }\quad
\mathcal{H}\ni h\downarrow0.
\end{equation*}

\noindent
\emph{\textbf{Step~4. Study of $\big\{h^{-\beta}\sum_{k=1}^{\lceil h^{-2}\rceil}\gamma_k\Pi_{k+1:\lceil h^{-2}\rceil}\rho^h_k,h\in\mathcal{H}\big\}$.}}
\newline\newline
Note that
\begin{equation*}
\begin{aligned}
\mathbb{E}[|\rho^h_k|^2]
&=\mathbb{E}\big[\Var\big(H_1(\xi^h_{k-1},X_h^{(k)})-H_1(\xi^h_\star,X_h^{(k)})\big|\mathcal{F}^h_{k-1}\big)\big]\\
&\leq\mathbb{E}\big[\big(H_1(\xi^h_{k-1},X_h^{(k)})-H_1(\xi^h_\star,X_h^{(k)})\big)^2\big].
\end{aligned}
\end{equation*}
By the uniform boundedness of $\{f_{X_h},h\in\mathcal{H}\}$, guaranteed under Assumption~\ref{asp:misc}(\ref{asp:misc:iv}),
\begin{equation*}
\begin{aligned}
\mathbb{E}\big[\big(H_1(\xi^h_{k-1},X_h^{(k)})&-H_1(\xi^h_\star,X_h^{(k)})\big)^2\big]\\
&=\frac1{(1-\alpha)^2}\,\mathbb{E}\Big[\Big|\mathds{1}_{X_h^{(k)} \geq \xi^h_{k-1}}-\mathds{1}_{X_h^{(k)} \geq \xi^h_\star}\Big|\Big]\\
&=\frac1{(1-\alpha)^2}\,\mathbb{E}\Big[\mathbb{E}\Big[\mathds{1}_{\xi^h_{k-1}\leq X_h^{(k)}<\xi^h_\star}+\mathds{1}_{\xi^h_\star\leq X_h^{(k)}<\xi^h_{k-1}}\Big|\mathcal{F}^h_{k-1}\Big]\Big]\\
&=\frac1{(1-\alpha)^2}\,\mathbb{E}\big[\big|F_{X_h}(\xi^h_{k-1})-F_{X_h}(\xi^h_\star)\big|\big]\\
&\leq\frac{\sup_{h'\in\mathcal{H}}\|f_{X_{h'}}\|_\infty}{(1-\alpha)^2}\,\mathbb{E}[(\xi^h_{k-1}-\xi^h_\star)^2]^\frac12.
\end{aligned}
\end{equation*}
Observe that, by Assumption~\ref{asp:supE[sup]} and the fact that $\lim_{\mathcal{H}\ni h\downarrow0}\xi^h_\star=\xi^0_\star$,
\begin{equation*}
\sup_{h\in\mathcal{H}}\mathbb{E}[(\xi^h_0-\xi^h_\star)^2]\leq2(\sup_{h\in\mathcal{H}}{\mathbb{E}[|\xi^h_0|^2]}+\sup_{h\in\mathcal{H}}{|\xi^h_\star|^2})<\infty.
\end{equation*}
Thus, using Lemma~\ref{lmm:error}(\ref{lmm:error:statistical}),
\begin{equation}\label{eq:E[|rho|2]<}
\mathbb{E}[|\rho^h_k|^2]
\leq\mathbb{E}[(H_1(\xi^h_{k-1},X_h^{(k)})-H_1(\xi^h_\star,X_h^{(k)}))^2]
\leq C\gamma_k^\frac12.
\end{equation}

By \eqref{eq:H1} and \eqref{eq:rhohn}, $\mathbb{E}[\rho^{h}_k|\mathcal{F}_{k-1}^h]=0$, so that $\{\rho^h_k,k\geq1\}$ are $\{\mathcal{F}^h_k,k\geq1\}$-martingale increments. Hence, via \eqref{eq:E[|rho|2]<},
\begin{equation*}
\mathbb{E}\bigg[\bigg(\sum_{k=1}^n\gamma_k\Pi_{k+1:n}\rho^h_k\bigg)^2\bigg]
=\sum_{k=1}^n\gamma_k^2|\Pi_{k+1:n}|^2\,\mathbb{E}[|\rho^h_k|^2]
\leq\sum_{k=1}^n\gamma_k^\frac52|\Pi_{k+1:n}|^2.
\end{equation*}
Since $ h^{-\beta}\leq \gamma_1^\frac12\gamma_{\lceil h^{-2}\rceil}^{-\frac12}$ and $2\gamma_1V_0''(\xi^0_\star)>1$ if $\beta=1$ by \eqref{gamma:condition:beta:equal:1}, via \eqref{upper:estimate:pi:i:n} and Lemma~\ref{lmm:limsup}(\ref{lmm:limsup:i}),
\begin{equation*}
\limsup_{\mathcal{H}\ni h\downarrow0}\mathbb{E}\bigg[\bigg(h^{-\beta}\sum_{k=1}^{\lceil h^{-2}\rceil}\gamma_k\Pi_{k+1:\lceil h^{-2}\rceil}\rho^h_k\bigg)^2\bigg]
\leq C\limsup_{n\uparrow\infty}\gamma_n^{-1}\sum_{k=1}^n\gamma_k^\frac52\mathrm{e}^{-2V_0''(\xi^0_\star)\sum_{j=1}^n\gamma_j}
=0.
\end{equation*}
Eventually,
\begin{equation*}
h^{-\beta}\sum_{k=1}^{\lceil h^{-2}\rceil}\gamma_k\Pi_{k+1:\lceil h^{-2}\rceil}\rho^h_k
\stackrel{L^2(\mathbb{P})}{\longrightarrow}0
\quad\mbox{ as }\quad
\mathcal{H}\ni h\downarrow0.
\end{equation*}

\noindent
\emph{\textbf{Step~5. Study of $\Big\{\frac{h^{-1}}{\lceil h^{-2}\rceil}\sum_{k=1}^{\lceil h^{-2}\rceil}\theta^h_k,h\in\mathcal{H}\Big\}$.}}
\newline\newline
Via \eqref{eq:sa:nested:opt} and \eqref{eq:thetahk}, $\mathbb{E}[\theta_k^h|\mathcal{F}_{k-1}^h]=0$, i.e.~$\{\theta^h_k,k\geq1\}$ are $\{\mathcal{F}^h_k,k\geq1\}$-martingale increments. Hence, by Lemma~\ref{lmm:error}(\ref{lmm:error:statistical}),
\begin{equation}
\label{eq:E[sum(thetahk)]<}
\begin{aligned}
\mathbb{E}\bigg[\bigg(\frac1{\sqrt{n}}\sum_{k=1}^n\theta^h_k\bigg)^2\bigg]
&=\frac1n\sum_{k=1}^n\mathbb{E}[|\theta^h_k|^2]\\
&\leq\frac1n\sum_{k=1}^n\mathbb{E}\Big[\Big|\xi_{k-1}^h-\xi^h_\star+\frac1{1-\alpha}\big((X_h^{(k)}-\xi_{k-1}^h)^+-(X_h^{(k)}-\xi^h_\star)^+\big)\Big|^2\Big]\\
&\leq \frac{C}n\sum_{k=0}^{n-1}\mathbb{E}[(\xi^h_k-\xi^h_\star)^2]\\
&\leq C\bigg(\frac1n\,\mathbb{E}[(\xi^h_0-\xi^h_\star)^2]+\frac1n\sum_{k=1}^{n-1}\gamma_k\bigg).
\end{aligned}
\end{equation}
Using that $\sup_{h\in\mathcal{H}}\mathbb{E}[(\xi^h_0-\xi^h_\star)^2]<\infty$ and a comparison between series and integrals,
\begin{equation*}
\frac{h^{-1}}{\lceil h^{-2}\rceil}\sum_{k=1}^{\lceil h^{-2}\rceil}\theta^h_k
\stackrel{L^2(\mathbb{P})}{\longrightarrow}0
\quad\mbox{ as }\quad
\mathcal{H}\ni h\downarrow0.
\end{equation*}

\noindent
\emph{\textbf{Step~6. Study of $\Big\{\frac{h^{-1}}{\lceil h^{-2}\rceil}\sum_{k=1}^{\lceil h^{-2}\rceil}\zeta^h_k,h\in\mathcal{H}\Big\}$.}}
\newline\newline
Using a second order Taylor expansion, the uniform boundedness of $\{V_h''=(1-\alpha)^{-1}f_{X_h},h\in\mathcal{H}\}$ and Lemma~\ref{lmm:error}(\ref{lmm:error:statistical}),
\begin{equation}
\label{eq:E[sum(zetahk)]<}
\begin{aligned}
\mathbb{E}\bigg[\bigg|\frac1{\sqrt{n}}\sum_{k=1}^n\zeta^h_k\bigg|\bigg]
&\leq\frac{\sup_{h'\in\mathcal{H}}\|f_{X_{h'}}\|_\infty}{2(1-\alpha)}\bigg(\frac1{\sqrt{n}}\,\mathbb{E}[(\xi^h_0-\xi^h_\star)^2]
   +\frac1{\sqrt{n}}\sum_{k=2}^n\mathbb{E}[(\xi_{k-1}^h-\xi^h_\star)^2]\bigg)\\
&\leq C\bigg(\frac1{\sqrt{n}}\,\mathbb{E}[(\xi^h_0-\xi^h_\star)^2]
    +\frac1{\sqrt{n}}\sum_{k=1}^{n-1}\gamma_k\bigg).
\end{aligned}
\end{equation}
By a comparison between series and integrals, $\lim_{n\uparrow\infty}\frac1{\sqrt{n}}\sum_{k=1}^{n-1}\gamma_k=0$. Recalling that $\sup_{h\in\mathcal{H}}\mathbb{E}[(\xi^h_0-\xi^h_\star)^2]<\infty$, we get
\begin{equation*}
\frac{h^{-1}}{\lceil h^{-2}\rceil}\sum_{k=1}^{\lceil h^{-2}\rceil}\zeta^h_k
\stackrel{L^1(\mathbb{P})}{\longrightarrow}0
\quad\mbox{ as }\quad
\mathcal{H}\ni h\downarrow0.
\end{equation*}

\noindent
\emph{\textbf{Step~7. Study of $\Big\{\Big(h^{-\beta}\sum_{k=1}^{\lceil h^{-2}\rceil}\gamma_k\Pi_{k+1:\lceil h^{-2}\rceil}e^h_k, \frac{h^{-1}}{\lceil h^{-2}\rceil}\sum_{k=1}^{\lceil h^{-2}\rceil}\eta^h_k\Big), h\in\mathcal{H}\Big\}$.}}
\newline\newline
We seek here to apply the CLT~\cite[Corollary~3.1]{nla.cat-vn2887492} to the martingale array \begin{equation*}
\bigg\{\bigg(h^{-\beta}\sum_{k=1}^{\lceil h^{-2}\rceil}\gamma_k\Pi_{k+1:\lceil h^{-2}\rceil}e^h_k, \frac{h^{-1}}{\lceil h^{-2}\rceil}\sum_{k=1}^{\lceil h^{-2}\rceil}\eta^h_k\bigg),h\in\mathcal{H}\bigg\}.
\end{equation*}
\noindent
\emph{\textbf{Step~7.1. Verification of the conditional Lindeberg condition.}}
\newline
We first check the conditional Lindeberg condition.
By assumption, recalling that $\{\xi^h_\star,h\in\mathcal{H}\}$ is bounded, we have
\begin{equation*}
\sup_{h\in\mathcal{H}}\mathbb{E}[|X_h-\xi^h_\star|^{2+\delta}]\leq 2^{1+\delta}\big(\sup_{h\in\mathcal{H}}\mathbb{E}[|X_h|^{2+\delta}]+\sup_{h\in\mathcal{H}}|\xi^h_\star|^{2+\delta}\big)<\infty.
\end{equation*}

On the one hand, via \eqref{eq:H1}, \eqref{upper:estimate:pi:i:n} and Lemma~\ref{lmm:limsup}(\ref{lmm:limsup:i}), using that $(2+\delta)V_0''(\xi^0_\star)\gamma_1>1+\frac\delta2$ if $\beta=1$ by \eqref{gamma:condition:beta:equal:1},
\begin{equation*}
\begin{aligned}
\limsup_{\mathcal{H}\ni h\downarrow0}
\sum_{k=1}^{\lceil h^{-2}\rceil}
&\mathbb{E}\big[\big|h^{-\beta}\gamma_k\Pi_{k+1:\lceil h^{-2}\rceil}H_1(\xi^h_\star,X_h^{(k)})\big|^{2+\delta}\big]\\
&\leq\gamma_1^{1+\frac\delta2}c_\alpha^{2+\delta}
\limsup_{n\uparrow\infty}\gamma_n^{-(1+\frac{\delta}2)}\sum_{k=1}^n\gamma_k^{2+\delta}\mathrm{e}^{-(2+\delta)V_0''(\xi^0_\star)\sum_{j=k+1}^n\gamma_j}
=0,
\end{aligned}
\end{equation*}
where $c_\alpha=1\vee\alpha/(1-\alpha)$.

On the other hand,
\begin{align*}
\sum_{k=1}^{\lceil h^{-2}\rceil}
\mathbb{E}\Big[\Big|\frac{h^{-1}}{\lceil h^{-2}\rceil}\eta^h_k\Big|^{2+\delta}\Big]
\leq C\,h^{2+\delta}\sum_{k=1}^{\lceil h^{-2}\rceil}\mathbb{E}\big[\big|X_h-\xi^h_\star\big|^{2+\delta}\big] 
&\leq C \sup_{h'\in\mathcal{H}}\mathbb{E}\big[\big|X_{h'}-\xi^{h'}_\star\big|^{2+\delta}\big] h^{\delta}, 
\end{align*}
so that
\begin{equation*}
\limsup_{\mathcal{H}\ni h\downarrow0}
\sum_{k=1}^{\lceil h^{-2}\rceil}
\mathbb{E}\big[\big|\frac{h^{-1}}{\lceil h^{-2}\rceil}\eta^h_k\big|^{2+\delta}\big]=0.
\end{equation*}

The two previous limits thus guarantee the conditional Lindeberg condition.
\\

\noindent
\emph{\textbf{Step~7.2. Convergence of the conditional covariance matrices.}}
\newline
We now prove the convergence of the conditional covariance matrices $\big\{S_h=(S^{i,j}_h)_{1\leq i,j\leq 2},\\h\in\mathcal{H}\big\}$ defined by
\begin{align*}
S^{1,1}_h
& \coloneqq\sum_{k=1}^{\lceil h^{-2}\rceil}h^{-2\beta}\gamma_k^2\Pi_{k+1:\lceil h^{-2}\rceil}^2\mathbb{E}[H_1(\xi^h_\star,X_h)^2]
=\frac\alpha{1-\alpha}\,h^{-2\beta}\gamma_{\lceil h^{-2}\rceil}\Sigma_{\lceil h^{-2}\rceil},\\
S^{2,2}_h & \coloneqq \sum_{k=1}^{\lceil h^{-2}\rceil}\mathbb{E}\bigg[\Big(\frac{h^{-1}}{\lceil h^{-2}\rceil}\eta^h_k\Big)^2\bigg|\mathcal{F}^h_{k-1}\bigg]
=\frac{h^{-2}}{\lceil h^{-2}\rceil}\frac{\Var\big((X_h-\xi^h_\star)^+\big)}{(1-\alpha)^2},\\
S^{1,2}_h=S^{2,1}_h& \coloneqq \frac{h^{-(1+\beta)}}{\lceil h^{-2}\rceil}\sum_{k=1}^{\lceil h^{-2}\rceil}\gamma_k\Pi_{k+1:\lceil h^{-2}\rceil}\,\mathbb{E}\big[e^h_k\eta^h_k\big|\mathcal{F}^h_{k-1}\big],
\end{align*}
with
\begin{equation}
\label{eq:Sigma_n}
\Sigma_n=\frac1{\gamma_n}\sum_{k=1}^n\gamma_k^2\Pi_{k+1:n}^2.
\end{equation}

\noindent
\emph{\textbf{Step~7.2.1. Convergence of $\big\{S^{1,1}_h,h\in \mathcal{H}\big\}$.}}
\newline
Using that $\lim_{\mathcal{H}\ni h\downarrow0}h^{-2\beta}\gamma_{\lceil h^{-2}\rceil}=\gamma_1$,
\begin{equation}\label{eq:helper:1}
\lim_{\mathcal{H}\ni h\downarrow0} S^{1,1}_h=\frac{\gamma_1\alpha}{1-\alpha}\lim_{n\uparrow\infty}\Sigma_n,
\end{equation}
provided that the latter limit exists. We decompose
\begin{equation*}
\begin{aligned}
\Sigma_{n+1}
&=\gamma_{n+1}+\frac{\gamma_n}{\gamma_{n+1}}\big(1-\gamma_{n+1}V_0''(\xi^0_\star)\big)^2\,\Sigma_n\\
&=\Sigma_n+\frac{\gamma_n-\gamma_{n+1}}{\gamma_{n+1}}\Sigma_n+\gamma_n\gamma_{n+1}V_0''(\xi^0_\star)^2\,\Sigma_n+(\gamma_{n+1}-\gamma_n)+\gamma_n\big(1-2V_0''(\xi^0_\star)\Sigma_n\big).
\end{aligned}
\end{equation*}
Asymptotically,
\begin{equation*}
\frac{\gamma_n-\gamma_{n+1}}{\gamma_{n+1}}=\frac{\mathds{1}_{\beta=1}}{\gamma_1}\gamma_n+\oo(\gamma_n),
\quad
\gamma_{n+1}-\gamma_n=\oo(\gamma_n)
\quad\mbox{ and }\quad
\gamma_n\gamma_{n+1}=\oo(\gamma_n).
\end{equation*}
Hence
\begin{equation*}
\Sigma_{n+1}=\Sigma_n+\Big(1-\Big(2V_0''(\xi^0_\star)-\frac{\mathds{1}_{\beta=1}}{\gamma_1}\Big)\Sigma_n\Big)\gamma_n+(\Sigma_n+1)\oo(\gamma_n).
\end{equation*}
The best candidate limit for $\{\Sigma_n,n\geq1\}$ is
\begin{equation}
\label{eq:Sigma*}
\Sigma_\star\coloneqq\frac1{2V_0''(\xi^0_\star)-\frac{\mathds{1}_{\beta=1}}{\gamma_1}}.
\end{equation}
Letting
\begin{equation}\label{eq:Delta:Sigma}
\Delta\Sigma_n\coloneqq\Sigma_n-\Sigma_\star,
\end{equation}
we have
\begin{equation*}
\begin{aligned}
\Delta\Sigma_{n+1}
&=\Delta\Sigma_n+\Big(\frac{\gamma_n-\gamma_{n+1}}{\gamma_{n+1}}-\frac{\mathds{1}_{\beta=1}}{\gamma_1}\gamma_n\Big)\Delta\Sigma_n+\gamma_n\gamma_{n+1}\Delta\Sigma_n-\gamma_n\Big(2V_0''(\xi^0_\star)-\frac{\mathds{1}_{\beta=1}}{\gamma_1}\Big)\Delta\Sigma_n\\
&\quad+\Big(\frac{\gamma_n-\gamma_{n+1}}{\gamma_{n+1}}-\frac{\mathds{1}_{\beta=1}}{\gamma_1}\gamma_n\Big)\Sigma_\star+\gamma_n\gamma_{n+1}\Sigma_\star+(\gamma_{n+1}-\gamma_n)\\
&=\big(1-\mu\gamma_n+\oo(\gamma_n)\big)\Delta\Sigma_n+\oo(\gamma_n),
\end{aligned}
\end{equation*}
where $\mu\coloneqq2V_0''(\xi^0_\star)-\frac{\mathds{1}_{\beta=1}}{\gamma_1}>0$.
Let $\varepsilon>0$. There exists $n_0\geq0$ such that, for $n\geq n_0$, $1-(\mu+\varepsilon)\gamma_n>0$, hence
\begin{equation}
\label{eq:|DeltaSigma|}
|\Delta\Sigma_{n+1}|\leq\big(1-(\mu+\varepsilon)\gamma_n\big)|\Delta\Sigma_n|+\varepsilon\gamma_n.
\end{equation}
Thus, for $n\geq n_0$,
\begin{equation*}
|\Delta\Sigma_n|\leq|\Delta\Sigma_{n_0}|\exp\bigg(-(\mu+\varepsilon)\sum_{k=n_0}^n\gamma_k\bigg)+\varepsilon\sum_{k=n_0}^n\gamma_k\exp\bigg(-(\mu+\varepsilon)\sum_{j=k}^n\gamma_j\bigg),
\end{equation*}
so that by Lemma~\ref{lmm:limsup}, $\limsup_{n\uparrow\infty}|\Delta\Sigma_n|\leq C\varepsilon$, which yields
\begin{equation}
\label{eq:lim:Sigma:n}
\lim_{n\uparrow\infty}\Sigma_n
=\Sigma_\star.
\end{equation}
Combining \eqref{eq:helper:1} and \eqref{eq:lim:Sigma:n},
\begin{equation*}
S^{1,1}_{h}
\stackrel{\Pas}{\longrightarrow}
\frac\alpha{1-\alpha}\,\frac{\gamma_1}{2V_0''(\xi^0_\star)-\frac{\mathds{1}_{\beta=1}}{\gamma_1}}
\quad\mbox{ as }\quad
\mathcal{H}\ni h\downarrow0.
\end{equation*}

\noindent\emph{\textbf{Step~7.2.2. Convergence of $\big\{S^{2,2}_h,h\in\mathcal{H}\big\}$.}}
\newline
Given that $x\mapsto x^+$ is $1$-Lipschitz on $\mathbb{R}$,
\begin{equation*}
\mathbb{E}\big[\big|(X_h-\xi^h_\star)^+-(X_0-\xi^0_\star)^+\big|^2\big]
\leq2\big(\mathbb{E}[(X_h-X_0)^2]+(\xi^h_\star-\xi^0_\star)^2\big).
\end{equation*}
Note that
\begin{equation*}
\mathbb{E}[\Var(\varphi(Y,Z)|Y)]
\leq\mathbb{E}[\Var(\varphi(Y,Z)|Y)]+\Var(\mathbb{E}[\varphi(Y,Z)|Y])
=\Var(\varphi(Y,Z))
<\infty.
\end{equation*}
Hence
\begin{equation}\label{eq:Xh->X0:L2}
\mathbb{E}[(X_h-X_0)^2]
=h\,\mathbb{E}[\Var(\varphi(Y,Z)|Y)]
\to
0
\quad\mbox{ as }\quad
\mathcal{H}\ni h\downarrow0.
\end{equation}
Recalling that $\lim_{\mathcal{H}\ni h\downarrow0}\xi^h_\star=\xi^0_\star$,
it follows that $(X_h-\xi^h_\star)^+$ converges to $(X_0-\xi^0_\star)^+$ in $L^2(\mathbb{P})$ as $\mathcal{H}\ni h\downarrow0$, so that
\begin{equation*}
\Var\big((X_h-\xi^h_\star)^+\big)
\to
\Var\big((X_0-\xi^0_\star)^+\big)
\quad\mbox{ as }\quad
\mathcal{H}\ni h\downarrow0.
\end{equation*}
Eventually,
\begin{equation*}
S^{2,2}_h
\stackrel{\Pas}{\longrightarrow}
\frac{\Var\big((X_0-\xi^0_\star)^+\big)}{(1-\alpha)^2}
\quad\mbox{ as }\quad
\mathcal{H}\ni h\downarrow0.
\end{equation*}

\noindent\emph{\textbf{Step~7.2.3. Convergence of $\big\{S^{1,2}_h,h\in\mathcal{H}\big\}$.}}
\newline
Recalling that $\{(X_h-\xi^h_\star)^+,h\in\mathcal{H}\}$ converges to $(X_0-\xi^0_\star)^+$ in $L^2(\mathbb{P})$ as $\mathcal{H}\ni h\downarrow0$, we deduce that
\begin{equation}
\mathbb{E}\big[e^h_k\eta^h_k\big|\mathcal{F}^h_{k-1}\big]
=\frac\alpha{(1-\alpha)^2}\,\mathbb{E}[(X_h-\xi^h_\star)^+]
\to\frac\alpha{(1-\alpha)^2}\,\mathbb{E}[(X_0-\xi^0_\star)^+]
\quad\mbox{ as }\quad
\mathcal{H}\ni h\downarrow0.
\label{eq:dom-cv}
\end{equation}
Besides,
$\lim_{\mathcal{H}\ni h\downarrow0}\frac{h^{-(1+\beta)}}{\lceil h^{-2}\rceil}=1_{\beta=1}$.
Thus, by Lemma~\ref{lmm:limsup}(\ref{lmm:limsup:i}),
\begin{equation*}
\limsup_{\mathcal{H}\ni h\downarrow0}\frac{h^{-(1+\beta)}}{\lceil h^{-2}\rceil}\sum_{k=1}^{\lceil h^{-2}\rceil}\gamma_k\Pi_{k+1:\lceil h^{-2}\rceil}\,\Big|\mathbb{E}\big[e^h_k\eta^h_k\big|\mathcal{F}^h_{k-1}\big]-\frac\alpha{(1-\alpha)^2}\mathbb{E}[(X_0-\xi^0_\star)^+]\Big|=0,
\end{equation*}
so that
\begin{equation}\label{eq:helper:2}
\lim_{\mathcal{H}\ni h\downarrow0}\frac{h^{-(1+\beta)}}{\lceil h^{-2}\rceil}\sum_{k=1}^{\lceil h^{-2}\rceil}\gamma_k\Pi_{k+1:\lceil h^{-2}\rceil}\mathbb{E}\big[e^h_k\eta^h_k\big|\mathcal{F}^h_{k-1}\big]
=\frac{\alpha\,\mathbb{E}[(X_0-\xi^0_\star)^+]\mathds{1}_{\beta=1}}{(1-\alpha)^2}\lim_{n\uparrow\infty}\widetilde\Sigma_n,
\end{equation}
provided that the latter limit exists, where
\begin{equation*}
\widetilde\Sigma_n\coloneqq\sum_{k=1}^n\gamma_k\Pi_{k+1:n}.
\end{equation*}
Observe that
\begin{equation*}
\widetilde\Sigma_{n+1}
=\gamma_{n+1}+\big(1-\gamma_{n+1}V_0''(\xi^0_\star)\big)\widetilde\Sigma_n
=\widetilde\Sigma_n+\gamma_{n+1}\big(1-V_0''(\xi^0_\star)\widetilde\Sigma_n\big).
\end{equation*}
Let $\widetilde\Sigma_\star\coloneqq\frac1{V_0''(\xi^0_\star)}$ and $\Delta\widetilde\Sigma_n\coloneqq\widetilde\Sigma_n-\widetilde\Sigma_\star$.
Then
\begin{equation*}
\Delta\widetilde\Sigma_{n+1}
=\big(1-\gamma_{n+1}V_0''(\xi^0_\star)\big)\Delta\widetilde\Sigma_n
=-\widetilde\Sigma_\star\prod_{k=1}^{n+1}\big(1-\gamma_kV_0''(\xi^0_\star)\big),
\end{equation*}
so that
\begin{equation*}
|\Delta\widetilde\Sigma_n|\leq | \widetilde\Sigma_\star| \exp\bigg(-V_0''(\xi^0_\star)\sum_{k=1}^n\gamma_k\bigg).
\end{equation*}
Since $\sum_{n\geq1}\gamma_n=\infty$, it ensues that
\begin{equation}\label{eq:helper:3}
\lim_{n\uparrow\infty}{\sum_{k=1}^n\gamma_k\Pi_{k+1:n}}=\widetilde\Sigma_\star.
\end{equation}
Finally, by \eqref{eq:helper:2} and \eqref{eq:helper:3},
\begin{equation*}
S^{1,2}_h=S^{2,1}_h
\stackrel{\Pas}{\longrightarrow}
\frac{\alpha\,\mathbb{E}[(X_0-\xi^0_\star)^+]\mathds{1}_{\beta=1}}{V_0''(\xi^0_\star)(1-\alpha)^2}
\quad\mbox{ as }\quad
\mathcal{H}\ni h\downarrow 0.
\end{equation*}

The proof is now complete.

\section{Proof of Theorem~\ref{thm:avg:nested:clt}}
\label{prf:avg:nested:clt}

We employ similar notation to Theorem~\ref{thm:nested:clt}'s proof, Appendix~\ref{prf:nested:clt}.
\\

By the decomposition \eqref{decomposition:var:sa},
\begin{equation}
\label{decomposition:var:sa:avg}
\xi^h_n-\xi^h_\star
=\frac1{V_h''(\xi^h_\star)}a^h_n+\frac1{V_h''(\xi^h_\star)}r^h_n+\frac1{V_h''(\xi^h_\star)}\rho^h_n+\frac1{V_h''(\xi^h_\star)}e^h_n,
\end{equation}
recalling the definitions \eqref{eq:rhn}--\eqref{eq:ehn}, with $\{a^h_n,n\geq1\}$ given by
\begin{equation}
\label{eq:ahn}
a^h_n=-\frac1{\gamma_n}\big(\xi^h_n-\xi^h_\star-(\xi^h_{n-1}-\xi^h_\star)\big).
\end{equation}
From \eqref{eq:sa:avg:nested:alg:xi} and \eqref{decomposition:var:sa:avg},
\begin{equation*}
\overline{\xi}^h_n-\xi^h_\star
=\frac1{V_h''(\xi^h_\star)n}\sum_{k=1}^n a^h_k+\frac1{V_h''(\xi^h_\star)n}\sum_{k=1}^n r^h_k+\frac1{V_h''(\xi^h_\star)n}\sum_{k=1}^n\rho^h_k+\frac1{V_h''(\xi^h_\star)n}\sum_{k=1}^ne^h_k.
\end{equation*}

Hereafter, we set $n = \lceil h^{-2}\rceil$. We study next the participation of each term of the above decomposition in the asymptotic estimation error.\\

\noindent
\emph{\textbf{Step~1. Study of $\Big\{\frac{h^{-1}}{V_h''(\xi^h_\star)\lceil h^{-2}\rceil}\sum_{k=1}^{\lceil h^{-2}\rceil}a^h_k,h\in\mathcal{H}\Big\}$.}}
\newline\newline
By summing by parts,
\begin{equation}
\label{eq:abel}
\begin{aligned}
\frac1n\sum_{k=1}^n\frac1{\gamma_k}\big(\xi^h_k-\xi^h_\star-(\xi^h_{k-1}-\xi^h_\star)\big)
&=\frac1n\Big(\frac1{\gamma_n}(\xi^h_n-\xi^h_\star)-\frac1{\gamma_1}(\xi^h_0-\xi^h_\star)\Big)\\
&\quad+\frac1n\sum_{k=2}^n
\Big(\frac1{\gamma_{k-1}}-\frac1{\gamma_k}\Big)(\xi^h_{k-1}-\xi^h_\star).
\end{aligned}
\end{equation}
We deal with each term on the right hand side above separately.
On the one hand, recalling that $\sup_{h\in\mathcal{H}}\mathbb{E}[|\xi^h_0-\xi^h_\star|]<\infty$, it follows from Lemma~\ref{lmm:error}(\ref{lmm:error:statistical}) that 
\begin{equation*}
\begin{aligned}
\mathbb{E}\bigg[\bigg|\frac1n\Big(\frac1{\gamma_n}(\xi^h_n-\xi^h_\star)&-\frac1{\gamma_1}(\xi^h_0-\xi^h_\star)\Big)\bigg|\bigg]\\
&\leq\mathbb{E}\bigg[\bigg|\frac1n\Big(\frac1{\gamma_n}(\xi^h_n-\xi^h_\star)-\frac1{\gamma_1}(\xi^h_0-\xi^h_\star)\Big)\bigg|^2\bigg]^\frac12\\
&\leq \frac{C}n\Big(\frac1{\gamma_n}\mathbb{E}[(\xi^h_n-\xi^h_\star)^2]^\frac12+\mathbb{E}[|\xi^h_0-\xi^h_\star|]\Big)\\
&\leq \frac{C}n\Big(\frac1{\sqrt{\gamma_n}}+1\Big).
\end{aligned}
\end{equation*}
On the other hand, via Lemma~\ref{lmm:error}(\ref{lmm:error:statistical}) again and a comparison between series and integrals,
\begin{equation}\label{eq:abel:2}
\begin{aligned}
\mathbb{E}\bigg[\bigg|\frac1n\sum_{k=2}^n\Big(\frac1{\gamma_{k-1}}&-\frac1{\gamma_k}\Big)(\xi^h_{k-1}-\xi^h_\star)\bigg|\bigg]\\
&\leq\mathbb{E}\bigg[\bigg(\frac1n\sum_{k=2}^n\Big(\frac1{\gamma_{k-1}}-\frac1{\gamma_k}\Big)(\xi^h_{k-1}-\xi^h_\star)\bigg)^2\bigg]^\frac12\\
&\leq\frac1n\sum_{k=2}^n\Big(\frac1{\gamma_{k-1}}-\frac1{\gamma_k}\Big)\mathbb{E}[(\xi^h_{k-1}-\xi^h_\star)^2]^\frac12\\
&\leq \frac{C}n\sum_{k=2}^n\Big(\frac1{\gamma_{k-1}}-\frac1{\gamma_k}\Big)\gamma_{k-1}^\frac12\\
&\leq \frac{C}{n\sqrt{\gamma_n}}.
\end{aligned}
\end{equation}
Gathering the previous upper bounds,
\begin{equation}
\label{eq:E[a:avg]<}
\mathbb{E}\bigg[\bigg|\frac1n\sum_{k=1}^n\frac1{\gamma_k}a^h_k\bigg|\bigg]
\leq \frac{C}{n\sqrt{\gamma_n}}.
\end{equation}
Now, recalling that $\lim_{\mathcal{H}\ni h\downarrow0}V_h''(\xi^h_\star)=V_0''(\xi^0_\star)$ and that $\beta\in\big(\frac12,1\big)$, we obtain
\begin{equation*}
\frac{h^{-1}}{V_h''(\xi^h_\star)\lceil h^{-2}\rceil}\sum_{k=1}^{\lceil h^{-2}\rceil}a^h_k
\stackrel{L^1(\mathbb{P})}{\longrightarrow}0
\quad\mbox{ as }\quad\mathcal{H}\ni h\downarrow0.
\end{equation*}

\noindent
\emph{\textbf{Step~2. Study of $\Big\{\frac{h^{-1}}{V_h''(\xi^h_\star)\lceil h^{-2}\rceil}\sum_{k=1}^{\lceil h^{-2}\rceil}r^h_k,h\in\mathcal{H}\Big\}$.}}
\newline\newline
Via \eqref{eq:E[|r|]<},
\begin{equation*}
\mathbb{E}\bigg[\bigg|\frac1{\sqrt{n}}\sum_{k=1}^n r^h_k\bigg|\bigg]
\leq \frac{C}{\sqrt{n}}\sum_{k=1}^n\gamma_k.
\end{equation*}
A comparison between series and integrals then yields
\begin{equation*}
\frac{h^{-1}}{V_h''(\xi^h_\star)\lceil h^{-2}\rceil}\sum_{k=1}^{\lceil h^{-2}\rceil}r^h_k
\stackrel{L^1(\mathbb{P})}{\longrightarrow}0
\quad\mbox{ as }\quad\mathcal{H}\ni h\downarrow0.
\end{equation*}

\noindent
\emph{\textbf{Step~3. Study of $\Big\{\frac{h^{-1}}{V_h''(\xi^h_\star)\lceil h^{-2}\rceil}\sum_{k=1}^{\lceil h^{-2}\rceil}\rho^h_k,h\in\mathcal{H}\Big\}$.}}
\newline\newline
Recalling that $\{\rho^h_k,k\geq1\}$ is a sequence of martingale increments, following \eqref{eq:E[|rho|2]<},
\begin{equation*}
\mathbb{E}\bigg[\bigg|\frac1{\sqrt{n}}\sum_{k=1}^n\rho^h_k\bigg|^2\bigg]
\leq \frac{C}n\sum_{k=1}^n \gamma_k^\frac12.
\end{equation*}
Thus, using a comparison between series and integrals,
\begin{equation*}
\frac{h^{-1}}{V_h''(\xi^h_\star)\lceil h^{-2}\rceil}\sum_{k=1}^{\lceil h^{-2}\rceil}\rho^h_k
\stackrel{L^2(\mathbb{P})}{\longrightarrow}0
\quad\mbox{ as }\quad\mathcal{H}\ni h\downarrow0.
\end{equation*}

\noindent
\emph{\textbf{Step~4. Study of $\Big\{ \Big(\frac{h^{-1}}{V_h''(\xi^h_\star)\lceil h^{-2}\rceil}\sum_{k=1}^{\lceil h^{-2}\rceil}e^h_k, \frac{h^{-1}}{\lceil h^{-2}\rceil}\sum_{k=1}^{\lceil h^{-2}\rceil}\eta^h_k\Big),h\in\mathcal{H}\Big\}$.}}
\newline\newline
We apply again the CLT~\cite[Corollary~3.1]{nla.cat-vn2887492} to the martingale arrays
\begin{equation*}
\bigg\{\bigg(\frac{h^{-1}}{V_h''(\xi^h_\star)\lceil h^{-2}\rceil}\sum_{k=1}^{\lceil h^{-2}\rceil}e^h_k ,\frac{h^{-1}}{\lceil h^{-2}\rceil}\sum_{k=1}^{\lceil h^{-2}\rceil}\eta^h_k\bigg),h\in\mathcal{H}\bigg\},
\end{equation*}
recalling the definition \eqref{eq:etahk}.
\\

\noindent
\emph{\textbf{Step~4.1. Verification of the conditional Lindeberg condition.}}
\newline
We start by checking the conditional Lindeberg condition.
Note that Step~7.1 of the proof of Theorem~\ref{thm:nested:clt}, Appendix~\ref{prf:nested:clt}, already guarantees that $\big\{ \frac{h^{-1}}{\lceil h^{-2}\rceil}\sum_{k=1}^{\lceil h^{-2}\rceil}\eta^h_k,h\in\mathcal{H} \big\}$ satisfies this condition.
Moreover, by assumption, via \eqref{eq:ehn} and \eqref{eq:H1},
\begin{equation*}
\begin{aligned}
\sum_{k=1}^{\lceil h^{-2}\rceil}\mathbb{E}\Big[\Big|\frac{h^{-1}}{V_h''(\xi^h_\star)\lceil h^{-2}\rceil}e^h_k\Big|^{2+\delta}\Big]
&\leq\sum_{k=1}^{\lceil h^{-2}\rceil}\frac{\mathbb{E}[|e^h_k|^{2+\delta}]}{|V_h''(\xi^h_\star)|^{2+\delta}\lceil h^{-2}\rceil^{1+\frac\delta2}}\\
&\leq\frac{c_\alpha^{2+\delta}}{|V_h''(\xi^h_\star)|^{2+\delta}\lceil h^{-2}\rceil^\frac\delta2}\to0
\quad\mbox{ as }\quad\mathcal{H}\ni h\downarrow0,
\end{aligned}
\end{equation*}
recalling that $c_\alpha=1\vee\alpha/(1-\alpha)$.
\\

\noindent
\emph{\textbf{Step~4.2. Convergence of the conditional covariance matrices.}}
\newline
We now prove the convergence of the conditional covariance matrices $\big\{S_h=(S_h^{i,j})_{1\leq i,j\leq 2},\\h\in\mathcal{H}\big\}$ defined by
\begin{align*}
S^{1,1}_h
& \coloneqq\sum_{k=1}^{\lceil h^{-2}\rceil}\frac{h^{-2}}{V_h''(\xi^h_\star)^2\lceil h^{-2}\rceil^2}\mathbb{E}[|e^h_k|^2|\mathcal{F}^h_{k-1}]
=\frac\alpha{1-\alpha}
\frac{h^{-2}}{V_h''(\xi^h_\star)^2\lceil h^{-2}\rceil},\\
S^{2,2}_h & \coloneqq \sum_{k=1}^{\lceil h^{-2}\rceil}\mathbb{E}\bigg[\Big(\frac{h^{-1}}{\lceil h^{-2}\rceil}\eta^h_k\Big)^2\bigg|\mathcal{F}^h_{k-1}\bigg]
=\frac{h^{-2}}{\lceil h^{-2}\rceil}\frac{\Var\big((X_h-\xi^h_\star)^+\big)}{(1-\alpha)^2},\\
S^{1,2}_h=S^{2,1}_h& \coloneqq\frac{h^{-2}}{V_h''(\xi^h_\star)\lceil h^{-2}\rceil^2}\sum_{k=1}^{\lceil h^{-2}\rceil}\mathbb{E}\big[e^h_k\eta^h_k\big|\mathcal{F}^h_{k-1}\big].
\end{align*}

First, recalling that $\lim_{\mathcal{H}\ni h\downarrow0}V_h''(\xi^h_\star)=V_0''(\xi^0_\star)=(1-\alpha)^{-1}f_{X_0}(\xi^0_\star)$, we get
\begin{equation*}
S^{1,1}_h
\stackrel{\Pas}{\longrightarrow}
\frac{\alpha(1-\alpha)}{f_{X_0}(\xi^0_\star)^2}
\quad\mbox{ as }\quad
\mathcal{H}\ni h\downarrow0.
\end{equation*}

Next, for the asymptotics of $\big\{S^{2, 2}_h, h \in \mathcal{H}\big\}$, we refer to Step~7.2 of the proof of Theorem~\ref{thm:nested:clt}, Appendix~\ref{prf:nested:clt}.

Finally, using that $\lim_{\mathcal{H}\ni h\downarrow0}V_h''(\xi^h_\star)=V_0''(\xi^0_\star)=(1-\alpha)^{-1}f_{X_0}(\xi^0_\star)$ and \eqref{eq:dom-cv}, we obtain by Ces\`aro's lemma that
\begin{equation*}
S^{1,2}_h=S^{2,1}_h
\stackrel{\Pas}{\longrightarrow}
\frac\alpha{1-\alpha}\frac{\mathbb{E}[(X_0-\xi^0_\star)^+]}{f_{X_0}(\xi^0_\star)}
\quad\mbox{ as }\quad\mathcal{H}\ni h\downarrow0.
\end{equation*}

The proof is now complete.

\section{Proof of Lemma~\ref{lmm:technical}}
\label{prf:technical}

\noindent
{(\ref{lmm:technical-i})}\
Let
\begin{equation*}
\Phi^{(k)}(Y)\coloneqq\varphi(Y,Z^{(k)})-\mathbb{E}[\varphi(Y,Z)|Y],
\end{equation*}
where $\{Z^{(k)},k\geq1\}\stackrel{\text{\rm\tiny i.i.d.}}{\sim}Z$ are independent of $Y$. We then write
\begin{equation*}
\begin{aligned}
X_{h_\ell}-X_{h_{\ell-1}}
&=\frac1{KM^\ell}\sum_{k=1}^{KM^\ell}\Phi^{(k)}(Y)-\frac1{KM^{\ell-1}}\sum_{k=1}^{KM^{\ell-1}}\Phi^{(k)}(Y)\\
&=\Big(1-\frac1M\Big)\bigg(\frac1{KM^{\ell-1}(M-1)}\sum_{k=KM^{\ell-1}+1}^{KM^\ell}\Phi^{(k)}(Y)-\frac1{KM^{\ell-1}}\sum_{k=1}^{KM^{\ell-1}}\Phi^{(k)}(Y)\bigg),
\end{aligned}
\end{equation*}
so that
\begin{equation*}
G_\ell
=\Big(1-\frac1M\Big)^\frac12
\bigg(\frac1{(KM^{\ell-1}(M-1))^\frac12}\sum_{k=KM^{\ell-1}+1}^{KM^\ell}\Phi^{(k)}(Y)
-\frac{(M-1)^\frac12}{(KM^{\ell-1})^\frac12}\sum_{k=1}^{KM^{\ell-1}}\Phi^{(k)}(Y)\bigg).
\end{equation*}
Conditionally on $Y$, we have the CLT
\begin{equation*}
U_\ell\coloneqq
\begin{pmatrix}
\frac1{(KM^{\ell-1}(M-1))^\frac12}\sum_{k=KM^{\ell-1}+1}^{KM^\ell}\Phi^{(k)}(Y)\\
\frac1{(KM^{\ell-1})^\frac12}\sum_{k=1}^{KM^{\ell-1}}\Phi^{(k)}(Y)
\end{pmatrix}
\stackrel[]{\mathcal{L}}{\longrightarrow}
\mathcal{N}\big(0,\mathbb{E}[\Phi^{(1)}(Y)^2|Y]I_2\big),
\quad\mbox{ as }\quad
\ell\uparrow\infty,
\end{equation*}
where $I_2$ stands for the $2\times2$ identity matrix.
Observing that $G_\ell=(1-1/M)^\frac12\langle u_M,U_\ell \rangle$, with $u_M\coloneqq\big(1,-(M-1)^\frac12\big)^\top$, it ensues that, conditionally on $Y$,
\begin{equation}
\label{eq:helper:4}
h_\ell^{-\frac12}(X_{h_\ell}-X_{h_{\ell-1}})
\stackrel[]{\mathcal{L}}{\longrightarrow}
\mathcal{N}\big(0,(M-1)\Var(\varphi(Y,Z)|Y)\big)
\quad\mbox{ as }\quad
\ell\uparrow\infty,
\end{equation}
where we used that, by the law of total variance,
\begin{equation}\label{eq:total:variance}
\mathbb{E}[\Var(\varphi(Y,Z)|Y)]
\leq\mathbb{E}[\Var(\varphi(Y,Z)|Y)]+\Var(\mathbb{E}[\varphi(Y,Z)|Y])
=\Var(\varphi(Y,Z))
<\infty,
\end{equation}
so that $\Var(\varphi(Y,Z)|Y)<\infty$ $\Pas$.
This concludes the proof.
\\

\noindent
{(\ref{lmm:technical-ii})}\
We have
\begin{equation*}
\mathbb{E}\big[\big|\mathds{1}_{X_{h_\ell}>\xi}-\mathds{1}_{X_{h_{\ell-1}}>\xi}\big|\big]
=\mathbb{P}(X_{h_{\ell-1}}\leq\xi < X_{h_\ell})+\mathbb{P}(X_{h_\ell}\leq\xi < X_{h_{\ell-1}}).
\end{equation*}
Introducing the random variable $G_\ell$, we compute
\begin{equation}
\label{eq:helper:5}
\begin{aligned}
\mathbb{P}(X_{h_{\ell-1}}\leq\xi<X_{h_\ell})
&=\mathbb{P}\big(X_{h_{\ell-1}}\leq\xi<X_{h_{\ell-1}}+h_\ell^\frac12G_\ell\big)\\
&=\mathbb{P}\big(X_{h_{\ell-1}}\leq\xi<X_{h_{\ell-1}}+h_\ell^\frac12G_\ell,G_\ell>0\big)\\
&=\mathbb{E}\big[\mathbb{P}\big(\xi-h_\ell^\frac12G_\ell<X_{h_{\ell-1}}\leq\xi,G_\ell>0\big|G_\ell\big)\big]\\
&=\mathbb{E}\big[\mathds{1}_{G_\ell>0}\big(F_{X_{h_{\ell-1}}\,|\,G_\ell}(\xi)-F_{X_{h_{\ell-1}}\,|\,G_\ell}(\xi-h_\ell^\frac12G_\ell)\big)\big]\\
&=h_\ell^\frac12\mathbb{E}[G_\ell^+f_{X_{h_{\ell-1}}|G_\ell}(\xi)]+h_\ell^\frac12\,r_\ell^+,
\end{aligned}
\end{equation}
where
\begin{equation*}
r_\ell^+\coloneqq\int_0^1\mathbb{E}\big[G_\ell^+\big(f_{X_{h_{\ell-1}}|G_\ell}(\xi-th_\ell^\frac12G_\ell)-f_{X_{h_{\ell-1}}|G_\ell}(\xi)\big)\big]\mathrm{d}t.
\end{equation*}
Similarly,
\begin{equation}
\label{eq:helper:6}
\mathbb{P}(X_{h_\ell}\leq\xi< X_{h_{\ell-1}})
=h_\ell^\frac12\mathbb{E}[G_\ell^-f_{X_{h_{\ell-1}}|G_\ell}(\xi)]+h_\ell^\frac12r_\ell^-,
\end{equation}
where
\begin{equation*}
r_\ell^-\coloneqq\int_0^1\mathbb{E}\big[G_\ell^-\big(f_{X_{h_{\ell-1}}|G_\ell}(\xi-th_\ell^\frac12G_\ell)-f_{X_{h_{\ell-1}}|G_\ell}(\xi)\big)\big]\mathrm{d}t.
\end{equation*}
Hence, combining \eqref{eq:helper:5} and \eqref{eq:helper:6},
\begin{equation*}
h_\ell^{-\frac12}\mathbb{E}\big[\big|\mathds{1}_{X_{h_\ell}>\xi}-\mathds{1}_{X_{h_{\ell-1}}>\xi}\big|\big]
=\mathbb{E}[|G_\ell|f_{X_{h_{\ell-1}}|G_\ell}(\xi)]+r_\ell^++r_\ell^-.
\end{equation*}

We will prove that the first term on the right hand side above converges to $\mathbb{E}[|G|f_G(\xi)]$ and that the remaining two vanish as $\ell\uparrow\infty$.

For a fixed $C>0$, we decompose the first term into
\begin{equation}
\begin{aligned}
\label{decomposition:first:term:local:variance:ML}
\mathbb{E}[|G_\ell|f_{X_{h_{\ell-1}}|G_\ell}(\xi)]
&=\mathbb{E}[|G_\ell|(f_{X_{h_{\ell-1}}|G_\ell}(\xi)-f_{G_\ell}(\xi))\mathds{1}_{|G_\ell|\leq C}]\\
&\quad+\mathbb{E}[|G_\ell|(f_{X_{h_{\ell-1}}|G_\ell}(\xi)-f_{G_\ell}(\xi))\mathds{1}_{|G_\ell|>C}]
+\mathbb{E}[|G_\ell|f_{G_\ell}(\xi)].
\end{aligned}
\end{equation}
Using that $\big\{g\mapsto f_{X_{h_{\ell-1}}|G_\ell=g}(\xi),\ell\geq1\big\}$ are bounded uniformly in $\ell\geq1$ and converge locally uniformly towards $g\mapsto f_g(\xi)$, by the dominated convergence theorem,
\begin{equation*}
\lim_{\ell\uparrow\infty}\big|\mathbb{E}[|G_\ell|(f_{X_{h_{\ell-1}}|G_\ell}(\xi)-f_{G_\ell}(\xi))\mathds{1}_{|G_\ell|\leq C}]\big|
\leq\lim_{\ell\uparrow\infty}C\,\mathbb{E}\big[\sup_{|g|\leq C}\big|f_{X_{h_{\ell-1}}|G_\ell=g}(\xi)-f_g(\xi)\big|\big]
=0.
\end{equation*}
Besides, via \eqref{eq:total:variance},
\begin{equation*}
\mathbb{E}[G_\ell^2]
\leq2h_\ell^{-1}\big(\mathbb{E}[(X_{h_\ell}-X_0)^2]+\mathbb{E}[(X_{h_{\ell-1}}-X_0)^2]\big)
\leq2(1+M^{-1})\mathbb{E}[\Var(\varphi(Y,Z)|Y)]
<\infty,
\end{equation*}
so that $\sup_{\ell\geq1}{\mathbb{E}[G_\ell^2]}<\infty$.
Recalling that $g\mapsto f_g(\xi)$ is bounded and that $\{g\mapsto f_{X_{h_{\ell-1}}|G_\ell=g}(\xi),\\\ell\geq1\}$ are bounded uniformly in $\ell\geq1$ by assumption, by the Cauchy-Schwarz and Markov inequalities,
\begin{equation*}
\mathbb{E}[|G_\ell|(f_{X_{h_{\ell-1}}|G_\ell}(\xi)-f_{G_\ell}(\xi))\mathds{1}_{|G_\ell|>C}]
\leq C'\,\mathbb{E}[G_\ell^2]^\frac12\mathbb{P}(|G_\ell|>C)^\frac12
\leq C'C^{-1}\sup_{\ell'\geq1}\mathbb{E}[G_{\ell'}^2],
\end{equation*}
for some positive constant $C'<\infty$.
Furthermore, since $\sup_{\ell\geq1}\mathbb{E}[G_\ell^2]<\infty$, $g\mapsto f_g(\xi)$ is continuous and bounded and $G_\ell\stackrel[\ell\uparrow\infty]{\mathcal{L}}{\longrightarrow} G$, it holds that
\begin{equation*}
\lim_{\ell\uparrow\infty}\mathbb{E}[|G_\ell|f_{G_\ell}(\xi)]=\mathbb{E}[|G|f_G(\xi)].
\end{equation*}
Finally, coming back to the decomposition \eqref{decomposition:first:term:local:variance:ML}, letting $\ell\uparrow\infty$ then $C\uparrow\infty$ yields
\begin{equation*}
\lim_{\ell\uparrow\infty}\mathbb{E}[|G_\ell|f_{X_{h_{\ell-1}}|G_\ell}(\xi)]=\mathbb{E}[|G|f_G(\xi)].
\end{equation*}

Similar lines of reasonings using the tightness of $\{G_\ell,\ell\geq1\}$, the uniform boundedness and local uniform convergence of $\{(x,g)\mapsto f_{X_{h_{\ell-1}}|G_\ell=g}(x),\ell\geq1\}$ and the continuity of $(x,g)\mapsto f_g(x)$ yield
\begin{equation*}
\lim_{\ell\uparrow\infty}r_\ell^+=\lim_{\ell\uparrow\infty}r_\ell^-=0.
\end{equation*}

\noindent
{(\ref{lmm:technical-iii})}\
Since $\mathbb{P}(X_0=\xi)=0$ by the continuity of $F_{X_0}$, a first order Taylor expansion gives
\begin{equation}\label{eq:taylor}
(X_h-\xi)^+=(X_0-\xi)^++\mathds{1}_{X_0>\xi}(X_h-X_0)+a(X_0,X_h-X_0)(X_h-X_0),
\end{equation}
where
\begin{equation*}
a(X_0,X_h-X_0)
=\int_0^1\big(\mathds{1}_{X_0+t(X_h-X_0)>\xi}-\mathds{1}_{X_0>\xi}\big)\mathrm{d}t.
\end{equation*}
Since $X_h\stackrel{\mathbb{P}}{\to}X_0$ as $\mathcal{H}\ni h\downarrow0$ by \eqref{eq:Xh->X0:L2} and $\mathbb{P}(X_0=\xi)=0$,
\begin{equation}\label{eq:a->0}
a(X_0,X_h-X_0)\stackrel[]{\mathbb{P}}{\longrightarrow}0
\quad\mbox{ as }\quad
\mathcal{H}\ni h\downarrow0.
\end{equation}

Using \eqref{eq:taylor},
\begin{equation}\label{eq:*}
\begin{aligned}
(X_{h_\ell}-\xi)^+-&(X_{h_{\ell-1}}-\xi)^+
=\mathds{1}_{X_0>\xi}(X_{h_\ell}-X_{h_{\ell-1}})\\
&+a(X_0,X_{h_\ell}-X_0)(X_{h_\ell}-X_0)
-a(X_0,X_{h_{\ell-1}}-X_0)(X_{h_{\ell-1}}-X_0).
\end{aligned}
\end{equation}
The standard CLT guarantees that both $\big\{h^{-\frac12}_\ell(X_{h_\ell}-X_0),\ell\geq1\big\}$ and $\big\{h^{-\frac12}_\ell(X_{h_{\ell-1}}-X_0),\ell\geq1\big\}$ are tight. Hence, by \eqref{eq:a->0},
\begin{equation}\label{eq:**}
h_\ell^{-\frac12}\big(a(X_0,X_{h_\ell}-X_0)(X_{h_\ell}-X_0)-a(X_0,X_{h_{\ell-1}}-X_0)(X_{h_{\ell-1}}-X_0)\big)
\stackrel[]{{\mathbb{P}}}{\longrightarrow}0
\quad\mbox{ as }\quad
\ell\uparrow\infty.
\end{equation}
Recalling that $X_0$ is $Y$-measurable by definition, using \eqref{eq:helper:4}, we get that conditionally on $Y$,
\begin{equation}\label{eq:***}
h_\ell^{-\frac12}\mathds{1}_{X_0>\xi}(X_{h_\ell}-X_{h_{\ell-1}})
\stackrel[]{\mathcal{L}}{\longrightarrow}
\mathds{1}_{X_0>\xi}\,G
\quad\mbox{ as }\quad
\ell\uparrow\infty.
\end{equation}
Eventually, combining \eqref{eq:*}, \eqref{eq:**} and \eqref{eq:***} and invoking Slutsky's theorem,
\begin{equation*}
h_\ell^{-\frac12}\big((X_{h_\ell}-\xi)^+-(X_{h_{\ell-1}}-\xi)^+\big)
\stackrel[]{\mathcal{L}}{\longrightarrow}
\mathds{1}_{X_0>\xi}\,G
\quad\mbox{ as }\quad\ell\uparrow\infty.
\end{equation*}

\section{Proof of Theorem~\ref{thm:ml:clt}}
\label{prf:ml:clt}

In the following developments, we denote by $C$ a positive constant whose value may change from line to line but does not depend upon $L$.
For simplicity, we drop the level indicating subscripts and superscripts $\ell$ from our notation and denote $\xi^{h_\ell}_n$, $\xi^{h_{\ell-1}}_n$, $X_{h_\ell}$, $X_{h_{\ell-1}}$, etc.~instead of $\xi^{h_\ell,\ell}_n$, $\xi^{h_{\ell-1},\ell}_n$, $X_{h_\ell,\ell}$, $X_{h_{\ell-1},\ell}$, etc.~the different variables intervening at level $\ell$ of the MLSA scheme, bearing in mind that they are levelwise independent.
\\

According to \eqref{eq:xi:ML}, \eqref{eq:xi*^hL} and \eqref{eq:xih-xi*},
\begin{equation}
\label{eq:xiML-xi*(hL)}
\begin{aligned}
\xi^\text{\tiny\rm ML}_\mathbf{N}-\xi^{h_L}_\star
&=\xi^{h_0}_{N_0}-\xi^{h_0}_\star
+\sum_{\ell=1}^L\big(\xi^{h_\ell}_{N_\ell}-\xi^{h_\ell}_\star-\big(\xi^{h_{\ell-1}}_{N_\ell}-\xi^{h_{\ell-1}}_\star\big)\big)\\
&=\xi^{h_0}_{N_0}-\xi^{h_0}_\star
+\sum_{\ell=1}^L\big(\xi^{h_\ell}_0-\xi^{h_\ell}_\star-(\xi^{h_{\ell-1}}_0-\xi^{h_{\ell-1}}_\star)\big)\Pi_{1:N_\ell}\\
&\quad+\sum_{\ell=1}^L\sum_{k=1}^{N_\ell}\gamma_k\Pi_{k+1:N_\ell}(g^{h_\ell}_k-g^{h_{\ell-1}}_k)
+\sum_{\ell=1}^L\sum_{k=1}^{N_\ell}\gamma_k\Pi_{k+1:N_\ell}(r^{h_\ell}_k-r^{h_{\ell-1}}_k)\\
&\quad+\sum_{\ell=1}^L\sum_{k=1}^{N_\ell}\gamma_k\Pi_{k+1:N_\ell}(\rho^{h_\ell}_k-\rho^{h_{\ell-1}}_k)
+\sum_{\ell=1}^L\sum_{k=1}^{N_\ell}\gamma_k\Pi_{k+1:N_\ell}(e^{h_\ell}_k-e^{h_{\ell-1}}_k),
\end{aligned}
\end{equation}
recalling the definitions \eqref{eq:ghn}--\eqref{eq:ehn} and \eqref{eq:Pi}.

Similarly, from \eqref{eq:C:ML}, \eqref{eq:C*^hL} and \eqref{eq:chih-chi*},
\begin{equation*}
\begin{aligned}
\chi^\text{\tiny\rm ML}_\mathbf{N}-\chi^{h_L}_\star
&=\chi^{h_0}_{N_0}-\chi^{h_0}_\star
+\sum_{\ell=1}^L\big(\chi^{h_\ell}_{N_\ell}-\chi^{h_\ell}_\star-\big(\chi^{h_{\ell-1}}_{N_\ell}-\chi^{h_{\ell-1}}_\star\big)\big)\\
&=\chi^{h_0}_{N_0}-\chi^{h_0}_\star
+\sum_{\ell=1}^L\frac1{N_\ell}\sum_{k=1}^{N_\ell}\theta^{h_\ell}_k-\theta^{h_{\ell-1}}_k\\
&\quad+\sum_{\ell=1}^L\frac1{N_\ell}\sum_{k=1}^{N_\ell}\zeta^{h_\ell}_k-\zeta^{h_{\ell-1}}_k
+\sum_{\ell=1}^L\frac1{N_\ell}\sum_{k=1}^{N_\ell}\eta^{h_\ell}_k-\eta^{h_{\ell-1}}_k,
\end{aligned}
\end{equation*}
recalling the definitions \eqref{eq:thetahk}--\eqref{eq:etahk}.
\\

\noindent
\emph{\textbf{Step~1. Study of $\big\{h_L^{-1}\big(\xi^{h_0}_{N_0}-\xi^{h_0}_\star\big),L\geq1\big\}$.}}
\newline\newline
Using Lemma~\ref{lmm:error}(\ref{lmm:error:statistical}) and \eqref{eq:N_ell},
\begin{equation*}
\mathbb{E}\big[\big|h_L^{-1}\big(\xi^{h_0}_{N_0}-\xi^{h_0}_\star\big)\big|^2\big]
\leq Ch_L^{-2}\gamma_{N_0}
\leq Ch_L^\frac{2\beta-1}{2(1+\beta)},
\end{equation*}
so that
\begin{equation*}
h_L^{-1}\big(\xi^{h_0}_{N_0}-\xi^{h_0}_\star\big)
\stackrel{L^2(\mathbb{P})}{\longrightarrow}0
\quad\mbox{ as }\quad
L\uparrow\infty.
\end{equation*}

\noindent
\emph{\textbf{Step~2. Study of $\Big\{h_L^{-1}\sum_{\ell=1}^L\big(\xi^{h_\ell}_0-\xi^{h_\ell}_\star-(\xi^{h_{\ell-1}}_0-\xi^{h_{\ell-1}}_\star)\big)\Pi_{1:N_\ell},L\geq1\Big\}$.}}
\newline\newline
Note that
\begin{equation}\label{eq:condition}
\gamma_1V_0''(\xi^0_\star)\geq\frac83\gamma_1\lambda>1
\quad\mbox{ if }\quad
\beta=1.
\end{equation}
Using that $\sup_{h\in\mathcal{H}}{\mathbb{E}[|\xi^h_0-\xi^h_\star|]}<\infty$, that $\limsup_{n\uparrow\infty}\gamma_n^{-1} |\Pi_{1:n}|=0$ by \eqref{upper:estimate:pi:i:n}, \eqref{eq:condition} and Lemma~\ref{lmm:limsup}(\ref{lmm:limsup:ii}), and \eqref{eq:N_ell},
\begin{equation*}
\begin{aligned}
\mathbb{E}\bigg[\bigg|h_L^{-1}\sum_{\ell=1}^L\big(\xi^{h_\ell}_0-\xi^{h_\ell}_\star&-(\xi^{h_{\ell-1}}_0-\xi^{h_{\ell-1}}_\star)\big)\Pi_{1:N_\ell}\bigg|\bigg]\\
&\leq2\sup_{h\in\mathcal{H}}{\mathbb{E}[|\xi^h_0-\xi^h_\star|]}\,h_L^{-1}\sum_{\ell=1}^L|\Pi_{1:N_\ell}|\\
& \leq Ch_L^{-1}\sum_{\ell=1}^L\gamma_{N_\ell}\\
& \leq Ch_L^{\frac12},
\end{aligned}
\end{equation*}
Thus
\begin{equation*}
h_L^{-1}\sum_{\ell=1}^L\big(\xi^{h_\ell}_0-\xi^{h_\ell}_\star-(\xi^{h_{\ell-1}}_0-\xi^{h_{\ell-1}}_\star)\big)\Pi_{1:N_\ell}
\stackrel{L^1(\mathbb{P})}{\longrightarrow}0
\quad\mbox{ as }\quad
L\uparrow\infty.
\end{equation*}

\noindent
\emph{\textbf{Step~3. Study of $\Big\{h_L^{-1}\sum_{\ell=1}^L\sum_{k=1}^{N_\ell}\gamma_k\Pi_{k+1:N_\ell}(g^{h_\ell}_k-g^{h_{\ell-1}}_k),L\geq1\Big\}$.}}
\newline\newline
Recalling that $\lim_{\mathcal{H}\ni h\downarrow0}\xi^h_\star=\xi^0_\star$, there exists a compact set $\mathcal{K}\subset\mathbb{R}$ such that $\xi^{h_\ell}_\star\in\mathcal{K}$, for all $\ell\geq0$. Thus, by Assumptions~\ref{asp:misc}(\ref{asp:misc:iv}) and~\ref{asp:fh-f0} and Lemma~\ref{lmm:error}(\ref{lmm:error:weak}),
\begin{equation*}
\begin{aligned}
|V_0''(\xi^0_\star)-V_{h_\ell}''(\xi^{h_\ell}_\star)|
&\leq\frac1{1-\alpha}\Big([f_{X_0}]_\mathrm{Lip}|\xi^{h_\ell}_\star-\xi^0_\star|
+\sup_{\xi\in\mathcal{K}}{\big|f_{X_0}(\xi)-f_{X_{h_\ell}}(\xi)\big|}\Big)\\
&\leq C\big(h_\ell+h_\ell^{\frac14+\delta_0}\big)
\leq Ch_\ell^{(\frac14+\delta_0)\wedge1}.
\end{aligned}
\end{equation*}
Besides, via \eqref{eq:ghn} and Lemma~\ref{lmm:error}(\ref{lmm:error:statistical}),
\begin{equation}
\label{eq:E[|g|]<}
\mathbb{E}\big[\big|g^{h_\ell}_k\big|\big]
\leq|V_0''(\xi^0_\star)-V_{h_\ell}''(\xi^{h_\ell}_\star)|\,\mathbb{E}[(\xi^{h_\ell}_{k-1}-\xi^{h_\ell}_\star)^2]^\frac12
\leq C h_\ell^{(\frac14+\delta_0)\wedge1}\gamma_k^\frac12.
\end{equation}
Finally, by \eqref{upper:estimate:pi:i:n}, \eqref{eq:condition}, Lemma~\ref{lmm:limsup}(\ref{lmm:limsup:i}) and \eqref{eq:N_ell}, distinguishing between the two cases $\delta_0<\frac34\wedge\frac{2\beta-1}{4(1+\beta)} = \frac{2\beta-1}{4(1+\beta)} $ and $\delta_0 \geq \frac{2\beta-1}{4(1+\beta)}$,
\begin{equation*}
\begin{aligned}
\mathbb{E}\bigg[\bigg|h_L^{-1}\sum_{\ell=1}^L\sum_{k=1}^{N_\ell}&\gamma_k\Pi_{k+1:N_\ell}(g^{h_\ell}_k-g^{h_{\ell-1}}_k)\bigg|\bigg]\\
&\leq h_L^{-1}\sum_{\ell=1}^L\sum_{k=1}^{N_\ell}\gamma_k|\Pi_{k+1:N_\ell}|\big(\mathbb{E}\big[\big|g^{h_\ell}_k\big|\big]+\mathbb{E}\big[\big|g^{h_{\ell-1}}_k\big|\big]\big)\\
&\leq Ch_L^{-1}\sum_{\ell=1}^Lh_\ell^{(\frac14+\delta_0)\wedge1}\sum_{k=1}^{N_\ell}\gamma_k^\frac32|\Pi_{k+1:N_\ell}|\\
&\leq Ch_L^{-1}\sum_{\ell=1}^Lh_\ell^{(\frac14+\delta_0)\wedge1}\gamma_{N_\ell}^\frac12\\
&\leq Ch_L^{\delta_0 \wedge \frac{2\beta-1}{4(1+\beta)}}.
\end{aligned}
\end{equation*}
Therefore,
\begin{equation*}
h_L^{-1}\sum_{\ell=1}^L\sum_{k=1}^{N_\ell}\gamma_k\Pi_{k+1:N_\ell}(g^{h_\ell}_k-g^{h_{\ell-1}}_k)
\stackrel{L^1(\mathbb{P})}{\longrightarrow}0
\quad\mbox{ as }\quad
L\uparrow\infty.
\end{equation*}

\noindent
\emph{\textbf{Step~4. Study of $\Big\{h_L^{-1}\sum_{\ell=1}^L\sum_{k=1}^{N_\ell}\gamma_k\Pi_{k+1:N_\ell}(r^{h_\ell}_k-r^{h_{\ell-1}}_k),L\geq1\Big\}$.}}
\newline\newline
It follows from \eqref{eq:E[|r|]<}, \eqref{upper:estimate:pi:i:n}, \eqref{eq:condition}, Lemma~\ref{lmm:limsup}(\ref{lmm:limsup:i}) and \eqref{eq:N_ell} that
\begin{equation*}
\begin{aligned}
\mathbb{E}\bigg[\bigg|h_L^{-1}\sum_{\ell=1}^L\sum_{k=1}^{N_\ell}&\gamma_k\Pi_{k+1:N_\ell}(r^{h_\ell}_k-r^{h_{\ell-1}}_k)\bigg|\bigg]\\
&\leq h_L^{-1}\sum_{\ell=1}^L\sum_{k=1}^{N_\ell}\gamma_k|\Pi_{k+1:N_\ell}|\,\big(\mathbb{E}\big[\big|r^{h_\ell}_k\big|\big] + \mathbb{E}\big[\big|r^{h_{\ell-1}}_k\big|\big]\big)\\
&\leq Ch_L^{-1}\sum_{\ell=1}^L\sum_{k=1}^{N_\ell}\gamma_k^2|\Pi_{k+1:N_\ell}|\\
&\leq Ch_L^{-1}\sum_{\ell=1}^L\gamma_{N_\ell}\\
&\leq Ch_L^\frac12.
\end{aligned}
\end{equation*}
Hence
\begin{equation*}
h_L^{-1}\sum_{\ell=1}^L\sum_{k=1}^{N_\ell}\gamma_k\Pi_{k+1:N_\ell}(r^{h_\ell}_k-r^{h_{\ell-1}}_k)
\stackrel{L^1(\mathbb{P})}{\longrightarrow}0
\quad\mbox{ as }\quad
L\uparrow\infty.
\end{equation*}

\noindent
\emph{\textbf{Step~5. Study of $\Big\{h_L^{-1}\sum_{\ell=1}^L\sum_{k=1}^{N_\ell}\gamma_k\Pi_{k+1:N_\ell}(\rho^{h_\ell}_k-\rho^{h_{\ell-1}}_k),L\geq1\Big\}$.}}
\newline\newline
Given that the innovations of the MLSA scheme are levelwise independent, the random variables $\big\{\sum_{k=1}^{N_\ell}\gamma_k\Pi_{k+1:N_\ell}(\rho^{h_\ell}_k-\rho^{h_{\ell-1}}_k),\ell\geq1\big\}$ are independent with zero-mean.
Moreover, at each level $\ell\geq1$, $\{\rho^{h_\ell}_k-\rho^{h_{\ell-1}}_k,k\geq1\}$ are $\{\mathcal{F}^{h_\ell}_k,k\geq1\}$-martingale increments.
Thus, from \eqref{eq:E[|rho|2]<}, \eqref{upper:estimate:pi:i:n}, \eqref{eq:condition}, Lemma~\ref{lmm:limsup}(\ref{lmm:limsup:i}) and \eqref{eq:N_ell},
\begin{equation*}
\begin{aligned}
\mathbb{E}\bigg[\bigg(h_L^{-1}\sum_{\ell=1}^L\sum_{k=1}^{N_\ell}&\gamma_k\Pi_{k+1:N_\ell}(\rho^{h_\ell}_k-\rho^{h_{\ell-1}}_k)\bigg)^2\bigg]\\
&=h_L^{-2}\sum_{\ell=1}^L\sum_{k=1}^{N_\ell}\gamma_k^2\Pi_{k+1:N_\ell}^2\mathbb{E}[(\rho^{h_\ell}_k-\rho^{h_{\ell-1}}_k)^2]\\
&\leq2h_L^{-2}\sum_{\ell=1}^L\sum_{k=1}^{N_\ell}\gamma_k^2\Pi_{k+1:N_\ell}^2\big(\mathbb{E}\big[\big|\rho^{h_\ell}_k\big|^2\big] + \mathbb{E}\big[\big|\rho^{h_{\ell-1}}_k\big|^2\big]\big)\\
&\leq Ch_L^{-2}\sum_{\ell=1}^L\sum_{k=1}^{N_\ell}\gamma_k^\frac52\Pi_{k+1:N_\ell}^2\\
&\leq Ch_L^{-2}\sum_{\ell=1}^L\gamma_{N_\ell}^\frac32\\
&\leq Ch_L^\frac14.
\end{aligned}
\end{equation*}
Therefore,
\begin{equation*}
h_L^{-1}\sum_{\ell=1}^L\sum_{k=1}^{N_\ell}\gamma_k\Pi_{k+1:N_\ell}(\rho^{h_\ell}_k-\rho^{h_{\ell-1}}_k)
\stackrel{L^2(\mathbb{P})}{\longrightarrow}0
\quad\mbox{ as }\quad
L\uparrow\infty.
\end{equation*}

\noindent
\emph{\textbf{Step~6. Study of $\Big\{h_L^{-\frac1\beta-\frac{2\beta-1}{4\beta(1+\beta)}}\sum_{\ell=1}^L\frac1{N_\ell}\sum_{k=1}^{N_\ell}\theta^{h_\ell}_k-\theta^{h_{\ell-1}}_k,L\geq1\Big\}$.}}
\newline\newline
Define
\begin{equation*}
\bar{\gamma}_n=\frac1n\sum_{k=1}^n\gamma_k,\quad n\geq1.
\end{equation*}
$\big\{\frac1{N_\ell}\sum_{k=1}^{N_\ell}\theta^{h_\ell}_k-\theta^{h_{\ell-1}}_k,\ell\geq1\big\}$ are independent and centered. Moreover, for any level $\ell\geq1$, $\big\{\theta^{h_\ell}_k-\theta^{h_{\ell-1}}_k,k\geq1\big\}$ are $\{\mathcal{F}^{h_\ell}_k,k\geq1\}$-martingale increments.
Hence, using \eqref{eq:E[sum(thetahk)]<}, a comparison between series and integrals and \eqref{eq:N_ell},
\begin{equation*}
\begin{aligned}
\mathbb{E}\bigg[\bigg|h_L^{-\frac1\beta-\frac{2\beta-1}{4\beta(1+\beta)}}&\sum_{\ell=1}^L\frac1{N_\ell}\sum_{k=1}^{N_\ell}\theta^{h_\ell}_k-\theta^{h_{\ell-1}}_k\bigg|^2\bigg]\\
&\leq2h_L^{-\frac2\beta-\frac{2\beta-1}{2\beta(1+\beta)}}\sum_{\ell=1}^L\frac1{N_\ell}\bigg(\mathbb{E}\bigg[\bigg(\frac1{\sqrt{N_\ell}}\sum_{k=1}^{N_\ell}\theta^{h_\ell}_k\bigg)^2\bigg] + \mathbb{E}\bigg[\bigg(\frac1{\sqrt{N_\ell}}\sum_{k=1}^{N_\ell}\theta^{h_{\ell-1}}_k\bigg)^2\bigg]\bigg)\\
&\leq Ch_L^{-\frac2\beta-\frac{2\beta-1}{2\beta(1+\beta)}}\sum_{\ell=1}^L\frac{\bar{\gamma}_{N_\ell}}{N_\ell}\\
&\leq C
\begin{cases}
   h_L^\frac{3\beta}{2(1+\beta)},
      &\beta\in\big(\frac12,1\big),\\
   h_L^\frac34\left|\ln{h_L}\right|,
      &\beta=1.
\end{cases}
\end{aligned}
\end{equation*}
Thus
\begin{equation*}
h_L^{-\frac1\beta-\frac{2\beta-1}{4\beta(1+\beta)}}\sum_{\ell=1}^L\frac1{N_\ell}\sum_{k=1}^{N_\ell}\theta^{h_\ell}_k-\theta^{h_{\ell-1}}_k
\stackrel{L^2(\mathbb{P})}{\longrightarrow}0
\quad\mbox{ as }\quad
L\uparrow\infty.
\end{equation*}

\noindent
\emph{\textbf{Step~7. Study of $\Big\{h_L^{-\frac1\beta-\frac{2\beta-1}{4\beta(1+\beta)}}\sum_{\ell=1}^L\frac1{N_\ell}\sum_{k=1}^{N_\ell}\zeta^{h_\ell}_k-\zeta^{h_{\ell-1}}_k,L\geq1\Big\}$.}}
\newline\newline
Using \eqref{eq:E[sum(zetahk)]<}, a comparison between series and integrals and \eqref{eq:N_ell},
\begin{equation*}
\begin{aligned}
\mathbb{E}\bigg[\bigg|h_L^{-\frac1\beta-\frac{2\beta-1}{4\beta(1+\beta)}}\sum_{\ell=1}^L
&\frac1{N_\ell}\sum_{k=1}^{N_\ell}\zeta^{h_\ell}_k-\zeta^{h_{\ell-1}}_k\bigg|\bigg]\\
&\leq h_L^{-\frac1\beta-\frac{2\beta-1}{4\beta(1+\beta)}}\sum_{\ell=1}^L\frac1{\sqrt{N_\ell}}\bigg(\mathbb{E}\bigg[\bigg|\frac1{\sqrt{N_\ell}}\sum_{k=1}^{N_\ell}\zeta^{h_\ell}_k\bigg|\bigg] + \mathbb{E}\bigg[\bigg|\frac1{\sqrt{N_\ell}}\sum_{k=1}^{N_\ell}\zeta^{h_{\ell-1}}_k\bigg|\bigg]\bigg)\\
&\leq Ch_L^{-\frac1\beta-\frac{2\beta-1}{4\beta(1+\beta)}}\sum_{\ell=1}^L\bar{\gamma}_{N_\ell}\\
&\leq C
\begin{cases}
   h_L^\frac9{4\beta(1+\beta)},
      &\beta\in\big(\frac12,1\big),\\
   h_L^\frac38\left|\ln{h_L}\right|,
      &\beta=1.
\end{cases}
\end{aligned}
\end{equation*}
Hence
\begin{equation*}
h_L^{-\frac1\beta-\frac{2\beta-1}{4\beta(1+\beta)}}\sum_{\ell=1}^L\frac1{N_\ell}\sum_{k=1}^{N_\ell}\zeta^{h_\ell}_k-\zeta^{h_{\ell-1}}_k
\stackrel{L^1(\mathbb{P})}{\longrightarrow}0
\quad\mbox{ as }\quad
L\uparrow\infty.
\end{equation*}

\noindent
\emph{\textbf{Step~8. Study of $\Big\{h_L^{-\frac1\beta-\frac{2\beta-1}{4\beta(1+\beta)}}\big(\chi^{h_0}_{N_0}-\chi^{h_0}_\star\big),L\geq1\Big\}$.}}
\newline\newline
It follows from \eqref{eq:N_ell} that
\begin{equation}\label{eq:use:1}
h_L^{-\frac1\beta-\frac{2\beta-1}{4\beta(1+\beta)}}\big(\chi^{h_0}_{N_0}-\chi^{h_0}_\star\big)
=h_0^{-\frac3{4(1+\beta)}}\big(1-M^{-\frac{2\beta-1}{2(1+\beta)}}\big)^\frac1{2\beta}
\frac{\sqrt{N_0}\big(\chi^{h_0}_{N_0}-\chi^{h_0}_\star\big)}{\big(1-M^{-\frac{2\beta-1}{2(1+\beta)}(L+1)}\big)^\frac1{2\beta}}.
\end{equation}
According to the decomposition \eqref{eq:chih-chi*},
\begin{equation*}
\sqrt{N_0}\big(\chi^{h_0}_{N_0}-\chi^{h_0}_\star\big)
=\frac1{\sqrt{N_0}}\sum_{k=1}^{N_0}\theta^{h_0}_k+\frac1{\sqrt{N_0}}\sum_{k=1}^{N_0}\zeta^{h_0}_k+\frac1{\sqrt{N_0}}\sum_{k=1}^{N_0}\eta^{h_0}_k.
\end{equation*}
By \eqref{eq:E[sum(thetahk)]<} and \eqref{eq:E[sum(zetahk)]<},
\begin{equation*}
\frac1{\sqrt{N_0}}\sum_{k=1}^{N_0}\theta^{h_0}_k
\stackrel{L^2(\mathbb{P})}{\longrightarrow}0,
\quad
\frac1{\sqrt{N_0}}\sum_{k=1}^{N_0}\zeta^{h_0}_k
\stackrel{L^1(\mathbb{P})}{\longrightarrow}0
\quad\mbox{ as }\quad
N_0\uparrow\infty.
\end{equation*}
Since $\{\eta^{h_0}_k,k\geq1\}$ are i.i.d.~such that $\mathbb{E}[|\eta^{h_0}_1|^2]<\infty$, the classical CLT yields
\begin{equation}\label{eq:use:end}
\frac1{\sqrt{N_0}}\sum_{k=1}^{N_0}\eta^{h_0}_k
\stackrel[]{\mathcal{L}}{\longrightarrow}
\mathcal{N}\bigg(0,\frac{\Var\big((X_{h_0}-\xi^{h_0}_\star)^+\big)}{(1-\alpha)^2}\bigg)
\quad\mbox{ as }\quad
N_0\uparrow\infty.
\end{equation}
Combining \eqref{eq:use:1}--\eqref{eq:use:end},
\begin{equation*}
\begin{aligned}
h_L^{-\frac1\beta-\frac{2\beta-1}{4\beta(1+\beta)}}\big(\chi^{h_0}_{N_0}&-\chi^{h_0}_\star\big)\\
&\stackrel[]{\mathcal{L}}{\longrightarrow}
\mathcal{N}\bigg(0,h_0^{-\frac3{2(1+\beta)}}\big(1-M^{-\frac{2\beta-1}{2(1+\beta)}}\big)^\frac1\beta\frac{\Var\big((X_{h_0}-\xi^{h_0}_\star)^+\big)}{(1-\alpha)^2}\bigg)
\quad\mbox{ as }\quad
L\uparrow\infty.
\end{aligned}
\end{equation*}

\noindent
\emph{\textbf{Step~9. Study of \[\bigg\{\bigg(h_L^{-1}\sum_{\ell=1}^L\sum_{k=1}^{N_\ell}\gamma_k\Pi_{k+1:N_\ell}(e^{h_\ell}_k-e^{h_{\ell-1}}_k),h_L^{-\frac1\beta-\frac{2\beta-1}{4\beta(1+\beta)}}\sum_{\ell=1}^L\frac1{N_\ell}\sum_{k=1}^{N_\ell}\eta^{h_\ell}_k-\eta^{h_{\ell-1}}_k\bigg),L\geq1\bigg\}.\]
}}
According to \eqref{eq:N_ell}, $N_1\geq\dots\geq N_L$.
Hence
\begin{equation*}
\sum_{\ell=1}^L\sum_{k=1}^{N_\ell}\gamma_k\Pi_{k+1:N_\ell}(e^{h_\ell}_k-e^{h_{\ell-1}}_k)
=\sum_{k=1}^{N_1}\sum_{\ell=1}^L\gamma_k\Pi_{k+1:N_\ell}(e^{h_\ell}_k-e^{h_{\ell-1}}_k)\mathds{1}_{1\leq k\leq N_\ell},
\end{equation*}
and
\begin{equation*}
\sum_{\ell=1}^L\frac1{N_\ell}\sum_{k=1}^{N_\ell}\eta^{h_\ell}_k-\eta^{h_{\ell-1}}_k
=\sum_{k=1}^{N_1}\sum_{\ell=1}^L\frac{\eta^{h_\ell}_k-\eta^{h_{\ell-1}}_k}{N_\ell}\mathds{1}_{1\leq k\leq N_\ell}.
\end{equation*}
We apply the CLT~\cite[Corollary~3.1]{nla.cat-vn2887492} to the martingale array
\begin{equation*}
\begin{aligned}
\bigg\{\bigg(h_L^{-1}\sum_{k=1}^{N_1}\sum_{\ell=1}^L\gamma_k\Pi_{k+1:N_\ell}(e^{h_\ell}_k&-e^{h_{\ell-1}}_k)\mathds{1}_{1\leq k\leq N_\ell},\\&h_L^{-\frac1\beta-\frac{2\beta-1}{4\beta(1+\beta)}}\sum_{k=1}^{N_1}\sum_{\ell=1}^L\gamma_k\Pi_{k+1:N_\ell}(e^{h_\ell}_k-e^{h_{\ell-1}}_k)\mathds{1}_{1\leq k\leq N_\ell}\bigg),L\geq1\bigg\}.
\end{aligned}
\end{equation*}

\noindent
\emph{\textbf{Step~9.1. Verification of the conditional Lindeberg condition.}}
\newline

\noindent
\emph{\textbf{Step~9.1.1. Study of $\Big\{h_L^{-1}\sum_{k=1}^{N_1}\sum_{\ell=1}^L\gamma_k\Pi_{k+1:N_\ell}(e^{h_\ell}_k-e^{h_{\ell-1}}_k)\mathds{1}_{1\leq k\leq N_\ell},L\geq1\Big\}$.}}
\newline\newline
By the levelwise independence of the innovations of the MLSA scheme and given that, for all $k\geq1$, the random variables $\big\{(e^{h_\ell}_k-e^{h_{\ell-1}}_k)\mathds{1}_{1\leq k\leq N_\ell},1\leq\ell\leq L\big\}$ are independent and centered, applying successively the Marcinkiewicz-Zygmund and the Jensen inequalities yields
\begin{align}
\sum_{k=1}^{N_1}\mathbb{E}\bigg[\bigg|h_L^{-1}\sum_{\ell=1}^L&\gamma_k\Pi_{k+1:N_\ell}(e^{h_\ell}_k-e^{h_{\ell-1}}_k)\mathds{1}_{1\leq k\leq N_\ell}\bigg|^{2+\delta}\bigg] \nonumber\\
&\leq Ch_L^{-(2+\delta)}\sum_{k=1}^{N_1}\gamma_k^{2+\delta}\,\mathbb{E}\bigg[\bigg|\sum_{\ell=1}^L\Pi_{k+1:N_\ell}^2(e^{h_\ell}_k-e^{h_{\ell-1}}_k)^2\mathds{1}_{1\leq k\leq N_\ell}\bigg|^{1+\frac\delta2}\bigg] \nonumber \\
&\leq Ch_L^{-(2+\delta)}\sum_{k=1}^{N_1}\gamma_k^{2+\delta}L^\frac\delta2\sum_{\ell=1}^L|\Pi_{k+1:N_\ell}|^{2+\delta}\,\mathbb{E}\big[\big|e^{h_\ell}_k-e^{h_{\ell-1}}_k\big|^{2+\delta}\big]\mathds{1}_{1\leq k\leq N_\ell} \nonumber \\
&=Ch_L^{-(2+\delta)}L^\frac\delta2\sum_{\ell=1}^L\sum_{k=1}^{N_\ell}\gamma_k^{2+\delta}|\Pi_{k+1:N_\ell}|^{2+\delta}\,\mathbb{E}\big[\big|e^{h_\ell}_k-e^{h_{\ell-1}}_k\big|^{2+\delta}\big]. \label{first:bound:increm:martingale:var:ML}
\end{align}
It follows from Lemma~\ref{lmm:technical}(\ref{lmm:technical-ii}), the uniform boundedness of $\big\{f_{X_{h_{\ell-1}}},\ell\ge1\big\}$, by Assumption~\ref{asp:misc}(\ref{asp:misc:iv}), and Lemma \ref{lmm:error}(\ref{lmm:error:weak}) that
\begin{equation}
\label{eq:E[|e|]<}
\begin{aligned}
\mathbb{E}\big[\big|e^{h_\ell}_k & -e^{h_{\ell-1}}_k\big|^{2+\delta}\big]
\\
& = \mathbb{E}\big[\big|H_1(\xi^{h_\ell}_\star,X_{h_\ell}^{(k)})-H_1(\xi^{h_{\ell-1}}_\star,X_{h_{\ell-1}}^{(k)})\big|^{2+\delta}\big]\\
& \leq\frac1{(1-\alpha)^{2+\delta}}\Big(\mathbb{E}\Big[\Big|\mathds{1}_{X_{h_\ell}>\xi^{h_\ell}_\star}-\mathds{1}_{X_{h_{\ell-1}}>\xi^{h_\ell}_\star}\Big|\Big]
+\mathbb{E}\Big[\Big|\mathds{1}_{X_{h_{\ell-1}}>\xi^{h_\ell}_\star}-\mathds{1}_{X_{h_{\ell-1}}>\xi^{h_{\ell-1}}_\star}\Big|\Big]\Big)\\
& \leq C \Big(\mathbb{E}\Big[\Big|\mathds{1}_{X_{h_\ell}>\xi^0_\star}-\mathds{1}_{X_{h_{\ell-1}}>\xi^0_\star}\Big|\Big]+ \mathbb{E}\Big[\Big|\mathds{1}_{X_{h_\ell}>\xi^{h_\ell}_\star}-\mathds{1}_{X_{h_{\ell}}>\xi^0_\star}\Big|\Big] \\
& \quad +  \mathbb{E}\Big[\Big|\mathds{1}_{X_{h_{\ell-1}}>\xi^{h_{\ell-1}}_\star}-\mathds{1}_{X_{h_{\ell-1}}>\xi^0_\star}\Big|\Big]
+\mathbb{E}\Big[\Big|\mathds{1}_{X_{h_{\ell-1}}>\xi^{h_\ell}_\star}-\mathds{1}_{X_{h_{\ell-1}}>\xi^{h_{\ell-1}}_\star}\Big|\Big]\Big)\\
& \leq C\big(h_\ell^{\frac12} + \mathbb{E}\big[\big|F_{X_{h_{\ell}}}(\xi^{h_{\ell}}_{\star})-F_{X_{h_{\ell}}}(\xi^0_\star)\big|\big] + \mathbb{E}\big[\big|F_{X_{h_{\ell-1}}}(\xi^{h_{\ell-1}}_{\star})-F_{X_{h_{\ell-1}}}(\xi^0_\star)\big|\big]\\
& \quad + \mathbb{E}\big[\big|F_{X_{h_{\ell-1}}}(\xi^{h_{\ell}}_{\star})-F_{X_{h_{\ell-1}}}(\xi^{h_{\ell-1}}_\star)\big|\big]\big)\\
& \leq C (h_\ell^\frac12 + |\xi^{h_{\ell}}_{\star}- \xi_{\star}| + |\xi^{h_{\ell-1}}_{\star}- \xi_{\star}| + |\xi^{h_{\ell}}_{\star}- \xi^{h_{\ell-1}}_{\star}|)\\
& \leq Ch_\ell^\frac12.
\end{aligned}
\end{equation}
Plugging the previous upper estimate into \eqref{first:bound:increm:martingale:var:ML} then using \eqref{eq:Pi}, that $(2+\delta)\gamma_1V_0''(\xi^0_\star)>1+\delta$ by \eqref{eq:condition}, Lemma~\ref{lmm:limsup}(\ref{lmm:limsup:i}) and \eqref{eq:N_ell},
\begin{equation*}
\begin{aligned}
\sum_{k=1}^{N_1}\mathbb{E}\bigg[\bigg|\sum_{\ell=1}^L&h_L^{-1}\gamma_k\Pi_{k+1:N_\ell}(e^{h_\ell}_k-e^{h_{\ell-1}}_k)\mathds{1}_{1\leq k\leq N_\ell}\bigg|^{2+\delta}\bigg]\\
&\leq Ch_L^{-(2+\delta)}L^\frac\delta2\sum_{\ell=1}^L\gamma_{N_\ell}^{1+\delta}h_\ell^\frac12\\
&\leq Ch_L^{\frac\delta2}L^{\frac\delta2}.
\end{aligned}
\end{equation*}
Hence
\begin{equation*}
\limsup_{L\uparrow\infty}
\sum_{k=1}^{N_1}\mathbb{E}\bigg[\bigg|\sum_{\ell=1}^Lh_L^{-1}\gamma_k\Pi_{k+1:N_\ell}(e^{h_\ell}_k-e^{h_{\ell-1}}_k)\mathds{1}_{1\leq k\leq N_\ell}\bigg|^{2+\delta}\bigg]
=0.
\end{equation*}

\noindent
\emph{\textbf{Step~9.1.2. Study of $\Big\{h_L^{-\frac1\beta-\frac{2\beta-1}{4\beta(1+\beta)}}\sum_{k=1}^{N_1}\sum_{\ell=1}^L\gamma_k\Pi_{k+1:N_\ell}(e^{h_\ell}_k-e^{h_{\ell-1}}_k)\mathds{1}_{1\leq k\leq N_\ell},L\geq1\Big\}$.}}
\newline\newline
Similarly,
\begin{align}
\sum_{k=1}^{N_1}\mathbb{E}\bigg[\bigg|\sum_{\ell=1}^L&h_L^{-\frac1\beta-\frac{2\beta-1}{4\beta(1+\beta)}}\frac{\eta^{h_\ell}_k-\eta^{h_{\ell-1}}_k}{N_\ell}\mathds{1}_{1\leq k\leq N_\ell}\bigg|^{2+\delta}\bigg] \nonumber \\
&\leq Ch_L^{-(2+\delta)(\frac1\beta+\frac{2\beta-1}{4\beta(1+\beta)})}\sum_{k=1}^{N_1}\mathbb{E}\bigg[\bigg|\sum_{\ell=1}^L\frac{(\eta^{h_\ell}_k-\eta^{h_{\ell-1}}_k)^2}{N_\ell^2}\mathds{1}_{1\leq k\leq N_\ell}\bigg|^{1+\frac\delta2}\bigg] \nonumber \\
&\leq Ch_L^{-(2+\delta)(\frac1\beta+\frac{2\beta-1}{4\beta(1+\beta)})}\sum_{k=1}^{N_1}L^\frac\delta2\sum_{\ell=1}^L\frac{\mathbb{E}\big[\big|\eta^{h_\ell}_k-\eta^{h_{\ell-1}}_k\big|^{2+\delta}\big]}{N_\ell^{2+\delta}}\mathds{1}_{1\leq k\leq N_\ell} \nonumber \\
&= Ch_L^{-(2+\delta)(\frac1\beta+\frac{2\beta-1}{4\beta(1+\beta)})}L^\frac\delta2\sum_{\ell=1}^L\frac1{N_\ell^{2+\delta}}\sum_{k=1}^{N_\ell}\mathbb{E}\big[\big|\eta^{h_\ell}_k-\eta^{h_{\ell-1}}_k\big|^{2+\delta}\big]. \label{bound:lindeberg:cond:ES:ML}
\end{align}
Note that
\begin{equation}\label{eq:eta^2+d<}
\begin{aligned}
\mathbb{E}\big[\big|\eta^{h_\ell}_k&-\eta^{h_{\ell-1}}_k\big|^{2+\delta}\big] \\
& \leq\frac{2^{1+\delta}}{(1-\alpha)^{2+\delta}}\,\mathbb{E}\big[\big|(X_{h_\ell}-\xi^{h_\ell}_\star)^+-(X_{h_{\ell-1}}-\xi^{h_{\ell-1}}_\star)^+\big|^{2+\delta}\big]\\
&\leq\frac{2^{2+2\delta}}{(1-\alpha)^{2+\delta}}\big(\mathbb{E}[|X_{h_\ell}-X_{h_{\ell-1}}|^{2+\delta}] + |\xi^{h_\ell}_\star-\xi^{h_{\ell-1}}_\star|^{2+\delta}\big)\\
& \leq Ch_\ell^{1+\frac\delta2},
\end{aligned}
\end{equation}
where we used that $x\mapsto x^+$, $x\in\mathbb{R}$, is $1$-Lipschitz and Lemmas~\ref{lmm:aux}(\ref{lmm:aux-i}) and~\ref{lmm:error}(\ref{lmm:error:weak}). Therefore, plugging the previous upper estimate into \eqref{bound:lindeberg:cond:ES:ML} and using \eqref{eq:N_ell}, we obtain
\begin{equation*}
\begin{aligned}
\sum_{k=1}^{N_1}\mathbb{E}\bigg[\bigg|\sum_{\ell=1}^L&h_L^{-\frac1\beta-\frac{2\beta-1}{4\beta(1+\beta)}}\frac{\eta^{h_\ell}_k-\eta^{h_{\ell-1}}_k}{N_\ell}\mathds{1}_{1\leq k\leq N_\ell}\bigg|^{2+\delta}\bigg]\\
&\leq Ch_L^{-(2+\delta)(\frac1\beta+\frac{2\beta-1}{4\beta(1+\beta)})}L^\frac\delta2\sum_{\ell=1}^L\frac{h_\ell^{1+\frac\delta2}}{N_\ell^{1+\delta}}\\
&\leq Ch_L^{\delta(\frac1\beta+\frac{2\beta-1}{4\beta(1+\beta)})}\left|\ln{h_L}\right|^\frac\delta2\sum_{\ell=1}^Lh_\ell^\frac{-(2-\beta)\delta+2\beta-1}{2(1+\beta)}.
\end{aligned}
\end{equation*}
Without loss of generality, by taking $\delta\in \big(0,\frac{2\beta-1}{2-\beta}\big)$, we get
\begin{equation*}
\sum_{k=1}^{N_1}\mathbb{E}\bigg[\bigg|\sum_{\ell=1}^Lh_L^{-\frac1\beta-\frac{2\beta-1}{4\beta(1+\beta)}}\frac{\eta^{h_\ell}_k-\eta^{h_{\ell-1}}_k}{N_\ell}\mathds{1}_{1\leq k\leq N_\ell}\bigg|^{2+\delta}\bigg]
\leq Ch_L^{\delta(\frac1\beta+\frac{2\beta-1}{4\beta(1+\beta)})}\left|\ln{h_L}\right|^\frac{\delta}2,
\end{equation*}
so that
\begin{equation*}
\limsup_{L\uparrow\infty}\sum_{k=1}^{N_1}\mathbb{E}\bigg[\bigg|\sum_{\ell=1}^Lh_L^{-\frac1\beta-\frac{2\beta-1}{4\beta(1+\beta)}}\frac{\eta^{h_\ell}_k-\eta^{h_{\ell-1}}_k}{N_\ell}\mathds{1}_{1\leq k\leq N_\ell}\bigg|^{2+\delta}\bigg]=0.
\end{equation*}

The conditional Lindeberg condition is subsequently satisfied.
\\

\noindent
\emph{\textbf{Step~9.2. Convergence of the conditional covariance matrices.}}
\newline
We now prove the convergence of the conditional covariance matrices $\big\{S_L=(S_L^{i,j})_{1\leq i,j\leq 2},\\L\geq 1\big\}$ defined by
\begin{align*}
S^{1,1}_L
& \coloneqq \sum_{k=1}^{N_1}\mathbb{E}\bigg[\bigg(\sum_{\ell=1}^Lh_L^{-1}\gamma_k\Pi_{k+1:N_\ell}(e^{h_\ell}_k-e^{h_{\ell-1}}_k)\mathds{1}_{1\leq k\leq N_\ell}\bigg)^2\bigg|\mathcal{F}^{h_L}_{k-1}\bigg],\\
S^{2,2}_L & \coloneqq \sum_{k=1}^{N_1}\mathbb{E}\bigg[\bigg(\sum_{\ell=1}^Lh_L^{-\frac1\beta-\frac{2\beta-1}{4\beta(1+\beta)}}\frac{\eta^{h_\ell}_k-\eta^{h_{\ell-1}}_k}{N_\ell}\mathds{1}_{1\leq k\leq N_\ell}\bigg)^2\bigg|\mathcal{F}^{h_L}_{k-1}\bigg],\\
S^{1,2}_L = S^{2,1}_L & \coloneqq \sum_{k=1}^{N_1} \mathbb{E}\bigg[\sum_{\ell=1}^L \frac{h_L^{-1-\frac1\beta-\frac{2\beta-1}{4\beta(1+\beta)}}}{N_\ell}\gamma_k\Pi_{k+1:N_\ell}(e^{h_\ell}_k-e^{h_{\ell-1}}_k)(\eta^{h_\ell}_k-\eta^{h_{\ell-1}}_k)\mathds{1}_{1\leq k\leq N_\ell}\bigg|\mathcal{F}^{h_L}_{k-1}\bigg].
\end{align*}

\noindent
\emph{\textbf{Step~9.2.1. Convergence of $\{S^{1,1}_L,L\geq1\}$.}}
\newline
$\big\{\sum_{\ell=1}^L\gamma_k\Pi_{k+1:N_\ell}(e^{h_\ell}_k-e^{h_{\ell-1}}_k)\mathds{1}_{1\leq k\leq N_\ell}),k\geq1\big\}$ are $\{\mathcal{F}^{h_L}_k,k\geq1\}$-martingale increments.
Moreover, $\big\{(e^{h_\ell}_k-e^{h_{\ell-1}}_k)\mathds{1}_{1\leq k\leq N_\ell},1\leq\ell\leq L\big\}$ are independent and centered.
Hence
\begin{equation*}
S^{1,1}_L
=\sum_{k=1}^{N_1}\mathbb{E}\bigg[\bigg(\sum_{\ell=1}^Lh_L^{-1}\gamma_k\Pi_{k+1:N_\ell}(e^{h_\ell}_k-e^{h_{\ell-1}}_k)\mathds{1}_{1\leq k\leq N_\ell}\bigg)^2\bigg]
=\sum_{\ell=1}^L U_\ell,
\end{equation*}
with
\begin{equation*}
\begin{aligned}
U_\ell
&\coloneqq\sum_{k=1}^{N_\ell}h_L^{-2}\gamma_k^2\Pi_{k+1:N_\ell}^2\mathbb{E}[(e^{h_\ell}_1-e^{h_{\ell-1}}_1)^2]\\
&\stackrel[L\uparrow\infty]{}{\sim}\gamma_1\bigg(\sum_{\ell=0}^Lh_\ell^{-\frac{2\beta-1}{2(1+\beta)}}\bigg)^{-1}h_\ell^{-\frac{2\beta-1}{2(1+\beta)}}\bigg(\gamma_{N_\ell}^{-1}\sum_{k=1}^{N_\ell}\gamma_k^2\Pi_{k+1:N_\ell}^2\bigg)h_\ell^{-\frac12}\mathbb{E}[(e^{h_\ell}_1-e^{h_{\ell-1}}_1)^2]\\
&\stackrel[L\uparrow\infty]{}{\sim}\gamma_1\bigg(\sum_{\ell=0}^Lh_\ell^{-\frac{2\beta-1}{2(1+\beta)}}\bigg)^{-1}h_\ell^{-\frac{2\beta-1}{2(1+\beta)}}\Sigma_{N_\ell}h_\ell^{-\frac12}\mathbb{E}[(e^{h_\ell}_1-e^{h_{\ell-1}}_1)^2],
\end{aligned}
\end{equation*}
where we used \eqref{eq:N_ell} and \eqref{eq:Sigma_n}.
Recalling the definitions \eqref{eq:Sigma*} and \eqref{eq:Delta:Sigma},
\begin{equation*}
\begin{aligned}
S^{1,1}_L
&\stackrel[L\uparrow\infty]{}{\sim}\gamma_1\bigg(\sum_{\ell=0}^Lh_\ell^{-\frac{2\beta-1}{2(1+\beta)}}\bigg)^{-1}\sum_{\ell=1}^Lh_\ell^{-\frac{2\beta-1}{2(1+\beta)}}\Sigma_\star h_\ell^{-\frac12}\mathbb{E}[(e^{h_\ell}_1-e^{h_{\ell-1}}_1)^2]\\
&\quad+\gamma_1\bigg(\sum_{\ell=0}^Lh_\ell^{-\frac{2\beta-1}{2(1+\beta)}}\bigg)^{-1}\sum_{\ell=1}^Lh_\ell^{-\frac{2\beta-1}{2(1+\beta)}}\Delta\Sigma_{N_\ell}h_\ell^{-\frac12}\mathbb{E}[(e^{h_\ell}_1-e^{h_{\ell-1}}_1)^2].
\end{aligned}
\end{equation*}

On the one hand, the uniform boundedness of $\{f_{X_{h_\ell}},\ell\geq1\}$, by Assumption~\ref{asp:misc}(\ref{asp:misc:iv}), and Lemma~\ref{lmm:error}(\ref{lmm:error:weak}) yield
\begin{equation*}
\begin{aligned}
h_\ell^{-\frac12}\Big|\mathbb{E}[(e^{h_\ell}_1&-e^{h_{\ell-1}}_1)^2]
-\frac1{(1-\alpha)^2}\,\mathbb{E}\big[\big|\mathds{1}_{X_{h_\ell}>\xi^0_\star}-\mathds{1}_{X_{h_{\ell-1}}-\xi^0_\star}\big|\big]\Big|\\
&\leq\frac{h_\ell^{-\frac12}}{(1-\alpha)^2}\big(\mathbb{E}\big[\big|\mathds{1}_{X_{h_\ell}>\xi^{h_\ell}_\star}-\mathds{1}_{X_{h_\ell}>\xi^0_\star}\big|\big]
+\mathbb{E}\big[\big|\mathds{1}_{X_{h_{\ell-1}}>\xi^{h_{\ell-1}}_\star}-\mathds{1}_{X_{h_{\ell-1}}>\xi^0_\star}\big|\big]\big)\\
&\leq\frac{h_\ell^{-\frac12}}{(1-\alpha)^2}\big(|F_{X_{h_\ell}}(\xi^{h_\ell}_\star)-F_{X_{h_\ell}}(\xi^0_\star)|
+|F_{X_{h_{\ell-1}}}(\xi^{h_{\ell-1}}_\star)-F_{X_{h_{\ell-1}}}(\xi^0_\star)|\big)\\
&\leq\frac{\sup_{\ell'\geq1}{\|f_{X_{h_{\ell'}}}\|_\infty}}{(1-\alpha)^2}h_\ell^{-\frac12}\big(|\xi^{h_\ell}_\star-\xi^0_\star|+ |\xi^{h_{\ell-1}}_\star-\xi^0_\star|\big)\\
&\leq Ch_\ell^\frac12.
\end{aligned}
\end{equation*}
Thus, by Lemma~\ref{lmm:technical}(\ref{lmm:technical-ii}),
\begin{equation}
\label{eq:limE[e]}
\lim_{\ell\uparrow\infty}{h_\ell^{-\frac12}\mathbb{E}[(e^{h_\ell}_1-e^{h_{\ell-1}}_1)^2]}
=\frac1{(1-\alpha)^2}\lim_{\ell\uparrow\infty}{h_\ell^{-\frac12}\mathbb{E}\big[\big|\mathds{1}_{X_{h_\ell}>\xi^0_\star}-\mathds{1}_{X_{h_{\ell-1}}>\xi^0_\star}\big|\big]}
=\frac{\mathbb{E}[|G| f_G(\xi^0_\star)]}{(1-\alpha)^2}.
\end{equation}

On the other hand, reusing the notation in Step~7.1 of Theorem~\ref{thm:nested:clt}'s proof, Appendix~\ref{prf:nested:clt}, for $\varepsilon>0$, there exists $n_0\geq0$ such that, for $n\geq n_0$, $1-(\mu+\varepsilon)\gamma_n>0$ and $|\Delta\Sigma_n|\leq\varepsilon$.
There also exists $L_0\geq1$ such that, for $L\geq L_0$, one has $N_1\geq\dots\geq N_L\geq n_0$, so that $|\Delta\Sigma_{N_L}|\leq\varepsilon$.
Let $L\geq L_0$.
By \eqref{eq:|DeltaSigma|} and Lemma~\ref{lmm:limsup}(\ref{lmm:limsup:i}), for $n\geq N_L$,
\begin{equation*}
\begin{aligned}
|\Delta\Sigma_n|
&\leq|\Delta\Sigma_{N_L}|\exp\bigg(-(\mu+\varepsilon)\sum_{k=N_L}^n\gamma_k\bigg)+\varepsilon\sum_{k=N_L}^n\gamma_k\exp\bigg(-(\mu+\varepsilon)\sum_{j=k}^n\gamma_j\bigg)\\
&\leq|\Delta\Sigma_{N_L}|+C\varepsilon
\leq C\varepsilon.
\end{aligned}
\end{equation*}
In particular, $\sup_{1\leq\ell\leq L}|\Delta\Sigma_{N_\ell}|\leq C\varepsilon$,
hence
\begin{equation}\label{eq:limsup:sup:Delta:Sigma}
\limsup_{L\uparrow\infty}\sup_{1\leq\ell\leq L}|\Delta\Sigma_{N_\ell}|=0.
\end{equation}

To sum up, using \eqref{eq:limE[e]} and \eqref{eq:limsup:sup:Delta:Sigma}, Ces\`aro's lemma (`$./\infty$' version) yields
\begin{equation*}
\begin{aligned}
\limsup_{L\uparrow\infty}\bigg|\gamma_1\bigg(\sum_{\ell=0}^Lh_\ell^{-\frac{2\beta-1}{2(1+\beta)}}\bigg)^{-1}\sum_{\ell=1}^L&h_\ell^{-\frac{2\beta-1}{2(1+\beta)}}\Delta\Sigma_{N_\ell}h_\ell^{-\frac12}\mathbb{E}[(e^{h_\ell}_1-e^{h_{\ell-1}}_1)^2]\bigg|\\
&\leq C\limsup_{L\uparrow\infty}\sup_{1\leq\ell\leq L}|\Delta\Sigma_{N_\ell}|
=0.
\end{aligned}
\end{equation*}
Thus, reusing \eqref{eq:limE[e]}, by Ces\`aro's lemma (`$./\infty$' version),
\begin{equation*}
\begin{aligned}
\lim_{L\uparrow\infty}S^{1,1}_L
&=\lim_{L\uparrow\infty}\gamma_1\bigg(\sum_{\ell=0}^Lh_\ell^{-\frac{2\beta-1}{2(1+\beta)}}\bigg)^{-1}
\bigg(\sum_{\ell=1}^Lh_\ell^{-\frac{2\beta-1}{2(1+\beta)}}\Sigma_\star h_\ell^{-\frac12}\,\mathbb{E}[(e^{h_\ell}_k-e^{h_{\ell-1}}_k)^2]\bigg)\\
&=\frac{\gamma_1\mathbb{E}[|G|f_G(\xi^0_\star)]}{(1-\alpha)^2\big(2V_0''(\xi^0_\star)-\frac{\mathds{1}_{\beta=1}}{\gamma_1}\big)}.
\end{aligned}
\end{equation*}

\noindent
\emph{\textbf{Step~9.2.2. Convergence of $\{S^{2,2}_L,L\geq1\}$.}}
\newline
Similarly, one has
\begin{equation*}
S^{2,2}_L
=\sum_{k=1}^{N_1}\mathbb{E}\bigg[\bigg(\sum_{\ell=1}^Lh_L^{-\frac1\beta-\frac{2\beta-1}{4\beta(1+\beta)}}\frac{\eta^{h_\ell}_k-\eta^{h_{\ell-1}}_k}{N_\ell}\mathds{1}_{1\leq k\leq N_\ell}\bigg)^2\bigg]
=\sum_{\ell=1}^L W_\ell,
\end{equation*}
where, using \eqref{eq:etahk},
\begin{equation*}
\begin{aligned}
W_\ell
&\coloneqq
\frac{h_L^{-\frac2\beta-\frac{2\beta-1}{2\beta(1+\beta)}}}{N_\ell^2}\sum_{k=1}^{N_\ell}\mathbb{E}[(\eta^{h_\ell}_k-\eta^{h_{\ell-1}}_k)^2]\\
&=\frac{h_L^{-\frac2\beta-\frac{2\beta-1}{2\beta(1+\beta)}}}{N_\ell} \mathbb{E}[(\eta^{h_\ell}_1-\eta^{h_{\ell-1}}_1)^2]\\
&=\frac{h_L^{-\frac2\beta-\frac{2\beta-1}{2\beta(1+\beta)}}}{(1-\alpha)^2}\frac{h_\ell}{N_\ell}h_\ell^{-1}\Var\big((X_{h_\ell}-\xi^{h_\ell}_\star)^+-(X_{h_{\ell-1}}-\xi^{h_{\ell-1}}_\star)^+\big).
\end{aligned}
\end{equation*}

By the $1$-Lipschitz property of $x\mapsto x^+$, $x\in\mathbb{R}$, and Lemma~\ref{lmm:error}(\ref{lmm:error:weak}),
\begin{equation*}
\begin{aligned}
h_\ell^{-\frac12}\mathbb{E}\big[\big|(X_{h_\ell}-\xi^{h_\ell}_\star)^+&-(X_{h_{\ell-1}}-\xi^{h_{\ell-1}}_\star)^+
-\big((X_{h_\ell}-\xi^0_\star)^+-(X_{h_{\ell-1}}-\xi^0_\star)^+\big)\big|\big]\\
&\leq h_\ell^{-\frac12}\big(|\xi^{h_\ell}_\star-\xi^0_\star|+|\xi^{h_{\ell-1}}_\star-\xi^0_\star|\big)
\leq Ch_\ell^\frac12.
\end{aligned}
\end{equation*}
Thus, by Lemma~\ref{lmm:technical}(\ref{lmm:technical-iii}) and Slutsky's theorem,
\begin{equation}\label{eq:weak:convergence}
h_\ell^{-\frac12}\big((X_{h_\ell}-\xi^{h_\ell}_\star)^+-(X_{h_{\ell-1}}-\xi^{h_{\ell-1}}_\star)^+\big)
\stackrel[\ell\uparrow\infty]{\mathcal{L}}{\longrightarrow}
\mathds{1}_{X_0>\xi^0_\star}\,G.
\end{equation}
Besides, the $1$-Lipschitz property of $x\mapsto x^+$, $x\in\mathbb{R}$, and Lemmas~\ref{lmm:aux}(\ref{lmm:aux-i}) and~\ref{lmm:error}(\ref{lmm:error:weak}) guarantee that, for $2<p\leq 2+\delta$,
\begin{equation}\label{eq:uniform:integrability}
\sup_{\ell\geq1}h_\ell^{-\frac{p}{2}}\,\mathbb{E}\big[\big|(X_{h_\ell}-\xi^{h_\ell}_\star)^+-(X_{h_{\ell-1}}-\xi^{h_{\ell-1}}_\star)^+\big|^p\big]<\infty.
\end{equation}
Eventually, the weak convergence \eqref{eq:weak:convergence} and the uniform integrability \eqref{eq:uniform:integrability} yield
\begin{equation}
\label{eq:limVarX+}
\lim_{\ell\uparrow\infty}{h_\ell^{-1}\,\Var\big((X_{h_\ell}-\xi^{h_\ell}_\star)^+-(X_{h_{\ell-1}}-\xi^{h_{\ell-1}}_\star)^+\big)}
=\Var(\mathds{1}_{X_0>\xi^0_\star}\,G).
\end{equation}

Ultimately, via \eqref{eq:N_ell} and Ces{\`a}ro's lemma (`$0/0$' version),
\begin{equation*}
\begin{aligned}
\lim_{L\uparrow\infty}S^{2,2}_L
&=\lim_{L\uparrow\infty}\frac{h_L^{-\frac2\beta-\frac{2\beta-1}{2\beta(1+\beta)}}}{(1-\alpha)^2}\bigg(\sum_{\ell=1}^L\frac{h_\ell}{N_\ell}\bigg)\\
&\quad\times\lim_{L\uparrow\infty}\bigg(\sum_{\ell=1}^L\frac{h_\ell}{N_\ell}\bigg)^{-1}
\bigg(\sum_{\ell=1}^L\frac{h_\ell}{N_\ell}h_\ell^{-1}\Var\big((X_{h_\ell}-\xi^{h_\ell}_\star)^+-(X_{h_{\ell-1}}-\xi^{h_{\ell-1}}_\star)^+\big)\bigg)\\
&=\frac{h_0^\frac{2\beta-1}{2(1+\beta)}\big(M^\frac{2\beta-1}{2(1+\beta)}-1\big)^{\frac1\beta-1}}{(1-\alpha)^2}\\
&\quad\times\lim_{L\uparrow\infty}\bigg(\sum_{\ell=1}^Lh_\ell^\frac{2\beta-1}{2(1+\beta)}\bigg)^{-1}
\bigg(\sum_{\ell=1}^Lh_\ell^\frac{2\beta-1}{2(1+\beta)}h_\ell^{-1}\Var\big((X_{h_\ell}-\xi^{h_\ell}_\star)^+-(X_{h_{\ell-1}}-\xi^{h_{\ell-1}}_\star)^+\big)\bigg)\\
&=\frac{h_0^\frac{2\beta-1}{2(1+\beta)}\big(M^\frac{2\beta-1}{2(1+\beta)}-1\big)^{\frac1\beta-1}}{(1-\alpha)^2}\Var(\mathds{1}_{X_0>\xi^0_\star}\,G).
\end{aligned}
\end{equation*}

\noindent
\emph{\textbf{Step~9.2.3. Convergence of $\{S^{1,2}_L,L\geq1\}$.}}
\newline
We prove hereby that $\big\{S^{1,2}_L,L\geq1\big\}$ converges to $0$ in $L^{1}(\mathbb{P})$.
It follows from \eqref{eq:limE[e]} and \eqref{eq:limVarX+} that
\begin{align}
\mathbb{E}[(e^{h_\ell}_1-e^{h_{\ell-1}}_1)^2]&=\frac1{(1-\alpha)^2}\,\mathbb{E}\big[\big|\mathds{1}_{X_{h_\ell}>\xi^0_\star}-\mathds{1}_{X_{h_{\ell-1}}>\xi^0_\star}\big|\big]\leq Ch_\ell^\frac12,
\label{eq:E[e2]}\\
\mathbb{E}[(\eta^{h_\ell}_1-\eta^{h_{\ell-1}}_1)^2]&=\frac1{(1-\alpha)^2}\Var\big((X_{h_\ell}-\xi^{h_\ell}_\star)^+-(X_{h_{\ell-1}}-\xi^{h_{\ell-1}}_\star)^+\big)\leq Ch_\ell,\nonumber
\end{align}
so that, by the Cauchy-Schwarz inequality,
\begin{equation}
\label{eq:E[|e.eta|]<}
\mathbb{E}\big[\big|(e^{h_\ell}_k-e^{h_{\ell-1}}_k)(\eta^{h_\ell}_k-\eta^{h_{\ell-1}}_k)\big|\big]
\leq Ch_\ell^\frac34.
\end{equation}
Using \eqref{eq:E[|e.eta|]<}, \eqref{upper:estimate:pi:i:n},
the fact $\gamma_1V_0''(\xi^0_\star)>0$ by Remark~\ref{rmk:nested:misc}(\ref{rmk:V''(xi*)>0}), Lemma~\ref{lmm:limsup}(\ref{lmm:limsup:i}) and \eqref{eq:N_ell}, we obtain
\begin{equation*}
\begin{aligned}
\big|S^{1,2}_L\big| & \leq h_L^{-1-\frac1\beta-\frac{2\beta-1}{4\beta(1+\beta)}}\sum_{\ell=1}^L\frac1{N_\ell}\sum_{k=1}^{N_\ell}\gamma_k |\Pi_{k+1:N_\ell}| \, \mathbb{E}\big[\big|(e^{h_\ell}_k-e^{h_{\ell-1}}_k)(\eta^{h_\ell}_k-\eta^{h_{\ell-1}}_k)\big|\big]\\
&\leq Ch_L^{-1-\frac1\beta-\frac{2\beta-1}{4\beta(1+\beta)}}\sum_{\ell=1}^L\frac{h_\ell^\frac34}{N_\ell}\sum_{k=1}^{N_\ell}\gamma_k|\Pi_{k+1:N_\ell}|\\
&\leq Ch_L^{-1+\frac1\beta+\frac{2\beta-1}{4\beta(1+\beta)}}\sum_{\ell=1}^Lh_\ell^{-\frac{3(1-\beta)}{4(1+\beta)}}\\
&\leq C
\begin{cases}
h_L^\frac{-\beta^2-\beta+3}{4\beta(1+\beta)},
      &\beta\in\big(\frac12,1\big),\\
   h_L^\frac18\left|\ln{h_L}\right|,
      &\beta=1.
\end{cases}
\end{aligned}
\end{equation*}
Thus
\begin{equation*}
S^{1,2}_L=S^{2,1}_L
\stackrel{\Pas}{\longrightarrow}0
\quad\mbox{ as }\quad
L\uparrow\infty.
\end{equation*}

All in all, by~\cite[Corollary~3.1]{nla.cat-vn2887492},
\begin{equation*}
\bigg(h_L^{-1}\sum_{\ell=1}^L\sum_{k=1}^{N_\ell}\gamma_k\Pi_{k+1:N_\ell}(e^{h_\ell}_k-e^{h_{\ell-1}}_k), h_L^{-\frac1\beta-\frac{2\beta-1}{4\beta(1+\beta)}}\sum_{\ell=1}^L\frac1{N_\ell}\sum_{k=1}^{N_\ell}\eta^{h_\ell}_k-\eta^{h_{\ell-1}}_k\bigg) \stackrel[L\uparrow\infty]{\mathcal{L}}{\rightarrow} \mathcal{N}(0, \widetilde{\Sigma}_\beta),
\end{equation*}
where
\begin{equation*}
\widetilde{\Sigma}_\beta\coloneqq\begin{pmatrix}
\frac{\gamma_1\mathbb{E}[|G|f_G(\xi^0_\star)]}{(1-\alpha)(2f_{X_0}(\xi^0_\star)-(1-\alpha)\gamma_1^{-1}\mathds{1}_{\beta=1})}
&0\\
0
&\frac{h_0^\frac{2\beta-1}{2(1+\beta)}\big(M^\frac{2\beta-1}{2(1+\beta)}-1\big)^\frac1\beta}{(1-\alpha)^2}
\frac{\Var(\mathds{1}_{X_0>\xi\star}\,G)}{M^\frac{2\beta-1}{2(1+\beta)}-1}
\end{pmatrix}.
\end{equation*}
Noting that the two sequences studied in Steps~8 and~9 are independent, by independency of the MLSA levels, concludes the proof.

\section{Proof of Theorem~\ref{thm:avg:ml:clt}}
\label{prf:avg:ml:clt}

Throughout the following developments, $C<\infty$ denotes a positive constant whose value may change from line to line but does not depend upon $L$.
As in Appendix~\ref{prf:ml:clt}, to alleviate notation, the level indicating subscripts and superscripts $\ell$ are dropped from the variables intervening at level $\ell$, keeping in mind that they are by design levelwise independent.
\\

Recalling the definitions \eqref{eq:ahn} and \eqref{eq:ghn}--\eqref{eq:ehn}, the decomposition \eqref{decomposition:var:sa} yields
\begin{equation}
\label{decomposition:var:sa:avg:bis}
\xi^h_n-\xi^h_\star
=\frac1{V_0''(\xi^0_\star)}a^h_n+\frac1{V_0''(\xi^0_\star)}g^h_n
+\frac1{V_0''(\xi^0_\star)}r^h_n+\frac1{V_0''(\xi^0_\star)}\rho^h_n+\frac1{V_0''(\xi^0_\star)}e^h_n.
\end{equation} 
Averaging out the previous identity,
\begin{equation*}
\begin{aligned}
\overline{\xi}^h_n-\xi^h_\star
&=\frac1{V_0''(\xi^0_\star)n}\sum_{k=1}^n a^h_k+\frac1{V_0''(\xi^0_\star)n}\sum_{k=1}^n g^h_k\\
&\quad+\frac1{V_0''(\xi^0_\star)n}\sum_{k=1}^nr^h_k+\frac1{V_0''(\xi^0_\star)n}\sum_{k=1}^n\rho^h_k+\frac1{V_0''(\xi^0_\star)n}\sum_{k=1}^n e^h_k.
\end{aligned}
\end{equation*}
Via \eqref{eq:xi:ML:avg}, \eqref{eq:xi*^hL} and \eqref{decomposition:var:sa:avg:bis},
\begin{align}
\overline{\xi}^\text{\tiny\rm ML}_\mathbf{N}-\xi^{h_L}_\star
&=\overline{\xi}^{h_0}_{N_0}-\xi^{h_0}_\star
+\sum_{\ell=1}^L\big(\overline{\xi}^{h_\ell}_{N_\ell}-\xi^{h_\ell}_\star-\big(\overline{\xi}^{h_{\ell-1}}_{N_\ell}-\xi^{h_{\ell-1}}_\star\big)\big)\nonumber\\
&=\overline{\xi}^{h_0}_{N_0}-\xi^{h_0}_\star+\frac1{V_0''(\xi^0_\star)}\sum_{\ell=1}^L\frac1{N_\ell}\sum_{k=1}^{N_\ell}a^{h_\ell}_k-a^{h_{\ell-1}}_k\nonumber\\
&\quad+\frac1{V_0''(\xi^0_\star)}\sum_{\ell=1}^L\frac1{N_\ell}\sum_{k=1}^{N_\ell}g^{h_\ell}_k-g^{h_{\ell-1}}_k+\frac1{V_0''(\xi^0_\star)}\sum_{\ell=1}^L\frac1{N_\ell}\sum_{k=1}^{N_\ell}r^{h_\ell}_k-r^{h_{\ell-1}}_k\nonumber\\
&\quad+\frac1{V_0''(\xi^0_\star)}\sum_{\ell=1}^L\frac1{N_\ell}\sum_{k=1}^{N_\ell}\rho^{h_\ell}_k-\rho^{h_{\ell-1}}_k+\frac1{V_0''(\xi^0_\star)}\sum_{\ell=1}^L\frac1{N_\ell}\sum_{k=1}^{N_\ell}e^{h_\ell}_k-e^{h_{\ell-1}}_k.
\label{eq:amlsa:split}
\end{align}

\noindent
\emph{\textbf{Step~1. Study of $\big\{h_L^{-1}\big(\overline{\xi}^{h_0}_{N_0}-\xi^{h_0}_\star\big),L\geq1\big\}$.}}
\newline\newline
Owing to \eqref{eq:sa:avg:nested:alg:xi}, Lemma~\ref{lmm:error}(\ref{lmm:error:statistical}), a comparison between series and integrals and the fact that $\beta<1$,
\begin{equation*}
\mathbb{E}\big[\big(\overline{\xi}^{h_0}_{N_0}-\xi^{h_0}_\star\big)^2\big]
\leq\frac1{N_0}\sum_{k=1}^{N_0}\mathbb{E}[(\xi^{h_0}_k-\xi^{h_0}_\star)^2]
\leq \frac{C}{N_0}\sum_{k=1}^{N_0}\gamma_k
\leq C\gamma_{N_0}.
\end{equation*}
Hence, using \eqref{eq:N_ell:avg},
\begin{equation*}
\mathbb{E}\big[\big|h_L^{-1}\big(\overline{\xi}^{h_0}_{N_0}-\xi^{h_0}_\star\big)\big|^2\big]
\leq Ch_L^{-2}\gamma_{N_0}
\leq Ch_L^{-2+\frac{9\beta}4}.
\end{equation*}
Recalling that $\beta\in\big(\frac89,1\big)$,
\begin{equation*}
h_L^{-1}\big(\overline{\xi}^{h_0}_{N_0}-\xi^{h_0}_\star\big)
\stackrel{L^2(\mathbb{P})}{\longrightarrow}0
\quad\mbox{ as }\quad
L\uparrow\infty.
\end{equation*}

\noindent
\emph{\textbf{Step~2. Study of $\Big\{h_L^{-1}\sum_{\ell=1}^L\frac1{N_\ell}\sum_{k=1}^{N_\ell}a^{h_\ell}_k-a^{h_{\ell-1}}_k,L\geq1\Big\}$.}}
\newline\newline
From \eqref{eq:E[a:avg]<} and \eqref{eq:N_ell:avg},
\begin{equation*}
\mathbb{E}\bigg[\bigg|h_L^{-1}\sum_{\ell=1}^L\frac1{N_\ell}\sum_{k=1}^{N_\ell}a^{h_\ell}_k-a^{h_{\ell-1}}_k\bigg|\bigg]
\leq Ch_L^{-1}\sum_{\ell=1}^L\frac1{N_\ell\sqrt{\gamma_{N_\ell}}}
\leq Ch_L^\frac{3\beta+2}4,
\end{equation*}
hence
\begin{equation*}
h_L^{-1}\sum_{\ell=1}^L\frac1{N_\ell}\sum_{k=1}^{N_\ell}a^{h_\ell}_k-a^{h_{\ell-1}}_k
\stackrel{L^1(\mathbb{P})}{\longrightarrow}0
\quad\mbox{ as }\quad
L\uparrow\infty.
\end{equation*}

\noindent
\emph{\textbf{Step~3. Study of $\Big\{h_L^{-1}\sum_{\ell=1}^L\frac1{N_\ell}\sum_{k=1}^{N_\ell}g^{h_\ell}_k-g^{h_{\ell-1}}_k,L\geq1\Big\}$.}}
\newline\newline
Using \eqref{eq:E[|g|]<} and a comparison between series and integrals, we obtain
\begin{equation*}
\begin{aligned}
\mathbb{E}\bigg[\bigg|h_L^{-1}\sum_{\ell=1}^L\frac1{N_\ell}\sum_{k=1}^{N_\ell}&g^{h_\ell}_k-g^{h_{\ell-1}}_k\bigg|\bigg]\\
&\leq2h_L^{-1}\sum_{\ell=1}^L\frac1{N_\ell}\sum_{k=1}^{N_\ell}(\mathbb{E}\big[\big|g^{h_\ell}_k\big|\big]+\mathbb{E}\big[\big|g^{h_{\ell-1}}_k\big|\big])\\
&\leq Ch_L^{-1}\sum_{\ell=1}^Lh_\ell^{(\frac14+\delta_0)\wedge1}\gamma_{N_\ell}^\frac12.
\end{aligned}
\end{equation*}
If $\delta_0\geq\frac34$, by \eqref{eq:N_ell:avg} and the fact that $\beta<1$,
\begin{equation*}
\begin{aligned}
\mathbb{E}\bigg[\bigg|h_L^{-1}\sum_{\ell=1}^L\frac1{N_\ell}\sum_{k=1}^{N_\ell} g^{h_\ell}_k-g^{h_{\ell-1}}_k\bigg|\bigg] &\leq Ch_L^{-1}\sum_{\ell=1}^Lh_\ell\gamma_{N_\ell}^\frac12\\
&\leq Ch_L^{-1+\frac98\beta}\sum_{\ell=1}^Lh_\ell^{1-\frac38\beta}\\
&\leq Ch_L^{-1+\frac98\beta},
\end{aligned}
\end{equation*}
Otherwise, if $\frac18\leq\delta_0<\frac34$, reusing \eqref{eq:N_ell:avg},
\begin{equation*}
\begin{aligned}
\mathbb{E}\bigg[\bigg|h_L^{-1}\sum_{\ell=1}^L\frac1{N_\ell}\sum_{k=1}^{N_\ell}&g^{h_\ell}_k-g^{h_{\ell-1}}_k\bigg|\bigg]\\
&\leq Ch_L^{-1}\sum_{\ell=1}^Lh_\ell^{\frac14+\delta_0}\gamma_{N_\ell}^\frac12\\
&\leq Ch_L^{-1+\frac98\beta}\sum_{\ell=1}^Lh_\ell^{\frac14+\delta_0-\frac38\beta}\\
&\leq C
\begin{cases}
h_L^{-1+\frac98\beta},
&\delta_0>\frac18,\\
h_L^{-1+\frac98\beta}\left|\ln{h_L}\right|,
&\delta_0=\frac18.
\end{cases}
\end{aligned}
\end{equation*}
Taking into account that $\beta\in\big(\frac89,1\big)$, we conclude
\begin{equation*}
h_L^{-1}\sum_{\ell=1}^L\frac1{N_\ell}\sum_{k=1}^{N_\ell}g^{h_\ell}_k-g^{h_{\ell-1}}_k
\stackrel{L^1(\mathbb{P})}{\longrightarrow}0
\quad\mbox{ as }\quad
L\uparrow\infty.
\end{equation*}

\noindent
\emph{\textbf{Step~4. Study of $\Big\{h_L^{-1}\sum_{\ell=1}^L\frac1{N_\ell}\sum_{k=1}^{N_\ell}r^{h_\ell}_k-r^{h_{\ell-1}}_k,L\geq1\Big\}$.}}
\newline\newline
Using \eqref{eq:E[|r|]<}, a comparison between series and integrals and \eqref{eq:N_ell:avg}, we obtain
\begin{equation*}
\begin{aligned}
\mathbb{E}\bigg[\bigg|h_L^{-1}\sum_{\ell=1}^L\frac1{N_\ell}\sum_{k=1}^{N_\ell}&r^{h_\ell}_k-r^{h_{\ell-1}}_k\bigg|\bigg]\\
&\leq2h_L^{-1}\sum_{\ell=1}^L\frac1{N_\ell}\sum_{k=1}^{N_\ell}(\mathbb{E}\big[\big|r^{h_\ell}_k\big|\big]+\mathbb{E}\big[\big|r^{h_{\ell-1}}_k\big|\big])\\
&\leq Ch_L^{-1}\sum_{\ell=1}^L\gamma_{N_\ell}\\
&\leq Ch_L^{-1+\frac32\beta}.
\end{aligned}
\end{equation*}
In view of the fact that $\beta\in\big(\frac89,1\big)$, we deduce
\begin{equation*}
h_L^{-1}\sum_{\ell=1}^L\frac1{N_\ell}\sum_{k=1}^{N_\ell}r^{h_\ell}_k-r^{h_{\ell-1}}_k
\stackrel{L^1(\mathbb{P})}{\longrightarrow}0
\quad\mbox{ as }\quad
L\uparrow\infty.
\end{equation*}

\noindent
\emph{\textbf{Step~5. Study of $\Big\{h_L^{-1}\sum_{\ell=1}^L\frac1{N_\ell}\sum_{k=1}^{N_\ell}\rho^{h_\ell}_k-\rho^{h_{\ell-1}}_k,L\geq1\Big\}$.}}
\newline\newline
Reusing similar arguments to Step~5 of the proof of Theorem~\ref{thm:ml:clt}, Appendix~\ref{prf:ml:clt}, we obtain via \eqref{eq:E[|rho|2]<}, \eqref{eq:N_ell:avg} and a comparison between series and integrals,
\begin{equation*}
\begin{aligned}
\mathbb{E}\bigg[\bigg(h_L^{-1}\sum_{\ell=1}^L\frac1{N_\ell}\sum_{k=1}^{N_\ell}&\rho^{h_\ell}_k-\rho^{h_{\ell-1}}_k\bigg)^2\bigg]\\
&=h_L^{-2}\sum_{\ell=1}^L\frac1{N_\ell^2}\sum_{k=1}^{N_\ell}\mathbb{E}[(\rho^{h_\ell}_k-\rho^{h_{\ell-1}}_k)^2]\\
&\leq2h_L^{-2}\sum_{\ell=1}^L\frac1{N_\ell^2}\sum_{k=1}^{N_\ell}(\mathbb{E}\big[\big|\rho^{h_\ell}_k\big|^2\big] + \mathbb{E}\big[\big|\rho^{h_{\ell-1}}_k\big|^2\big])\\
&\leq Ch_L^{-2}\sum_{\ell=1}^L\frac{\gamma_{N_\ell}^\frac12}{N_\ell}\\
&\leq Ch_L^{\frac{3\beta-2}4}.
\end{aligned}
\end{equation*}
Thus, given that $\beta\in\big(\frac89,1\big)$,
\begin{equation*}
h_L^{-1}\sum_{\ell=1}^L\frac1{N_\ell}\sum_{k=1}^{N_\ell}\rho^{h_\ell}_k-\rho^{h_{\ell-1}}_k
\stackrel{L^2(\mathbb{P})}{\longrightarrow}0
\quad\mbox{ as }\quad
L\uparrow\infty.
\end{equation*}

\noindent
\emph{\textbf{Step~6. Study of $\Big\{h_L^{-\frac98}\sum_{\ell=1}^L\frac1{N_\ell}\sum_{k=1}^{N_\ell}\theta^{h_\ell}_k-\theta^{h_{\ell-1}}_k,L\geq1\Big\}$.}}
\newline\newline
$\big\{\frac1{N_\ell}\sum_{k=1}^{N_\ell}\theta^{h_\ell}_k-\theta^{h_{\ell-1}}_k,\ell\geq1\big\}$ are independent and centered.
Hence, using \eqref{eq:E[sum(thetahk)]<}, a comparison between series and integrals and \eqref{eq:N_ell:avg},
\begin{equation*}
\begin{aligned}
\mathbb{E}\bigg[\bigg|h_L^{-\frac98}\sum_{\ell=1}^L\frac1{N_\ell}\sum_{k=1}^{N_\ell}&\theta^{h_\ell}_k-\theta^{h_{\ell-1}}_k\bigg|^2\bigg]\\
&\leq2h_L^{-\frac94}\sum_{\ell=1}^L\frac1{N_\ell}\bigg(\mathbb{E}\bigg[\bigg(\frac1{\sqrt{N_\ell}}\sum_{k=1}^{N_\ell}\theta^{h_\ell}_k\bigg)^2\bigg] + \mathbb{E}\bigg[\bigg(\frac1{\sqrt{N_\ell}}\sum_{k=1}^{N_\ell}\theta^{h_{\ell-1}}_k\bigg)^2\bigg]\bigg)\\
&\leq Ch_L^{-\frac94}\sum_{\ell=1}^L\frac{\gamma_{N_\ell}}{N_\ell}\\
&\leq Ch_L^{-\frac34+\frac32\beta}.
\end{aligned}
\end{equation*}
Therefore,
\begin{equation*}
h_L^{-\frac98}\sum_{\ell=1}^L\frac1{N_\ell}\sum_{k=1}^{N_\ell}\theta^{h_\ell}_k-\theta^{h_{\ell-1}}_k
\stackrel{L^2(\mathbb{P})}{\longrightarrow}0
\quad\mbox{ as }\quad
L\uparrow\infty.
\end{equation*}

\noindent
\emph{\textbf{Step~7. Study of $\Big\{h_L^{-\frac98}\sum_{\ell=1}^L\frac1{N_\ell}\sum_{k=1}^{N_\ell}\zeta^{h_\ell}_k-\zeta^{h_{\ell-1}}_k,L\geq1\Big\}$.}}
\newline\newline
Using \eqref{eq:E[sum(zetahk)]<}, we obtain
\begin{equation*}
\begin{aligned}
\mathbb{E}\bigg[\bigg|h_L^{-\frac98}\sum_{\ell=1}^L
\frac1{N_\ell}\sum_{k=1}^{N_\ell}&\zeta^{h_\ell}_k-\zeta^{h_{\ell-1}}_k\bigg|\bigg]\\
&\leq h_L^{-\frac98}\sum_{\ell=1}^L\frac1{\sqrt{N_\ell}}\bigg(\mathbb{E}\bigg[\bigg|\frac1{\sqrt{N_\ell}}\sum_{k=1}^{N_\ell}\zeta^{h_\ell}_k\bigg|\bigg] + \mathbb{E}\bigg[\bigg|\frac1{\sqrt{N_\ell}}\sum_{k=1}^{N_\ell}\zeta^{h_{\ell-1}}_k\bigg|\bigg]\bigg)\\
&\leq Ch_L^{-\frac98}\sum_{\ell=1}^L\gamma_{N_\ell}\\
&\leq Ch_L^{-\frac98+\frac32\beta}.
\end{aligned}
\end{equation*}
Hence, since $\beta\in\big(\frac89,1\big)$,
\begin{equation*}
h_L^{-\frac98}\sum_{\ell=1}^L\frac1{N_\ell}\sum_{k=1}^{N_\ell}\zeta^{h_\ell}_k-\zeta^{h_{\ell-1}}_k
\stackrel{L^1(\mathbb{P})}{\longrightarrow}0
\quad\mbox{ as }\quad
L\uparrow\infty.
\end{equation*}

\noindent
\emph{\textbf{Step~8. Study of $\Big\{h_L^{-\frac98}\big(\chi^{h_0}_{N_0}-\chi^{h_0}_\star\big),L\geq1\Big\}$.}}
\newline\newline
It follows from \eqref{eq:N_ell:avg} that
\begin{equation*}
h_L^{-\frac98}\big(\chi^{h_0}_{N_0}-\chi^{h_0}_\star\big)
=h_0^{-\frac38}\big(1-M^{-\frac14}\big)^\frac12
\frac{\sqrt{N_0}\big(\chi^{h_0}_{N_0}-\chi^{h_0}_\star\big)}{\big(1-M^{-\frac14(L+1)}\big)^\frac12}.
\end{equation*}
According to the decomposition \eqref{eq:chih-chi*},
\begin{equation*}
\sqrt{N_0}\big(\chi^{h_0}_{N_0}-\chi^{h_0}_\star\big)
=\frac1{\sqrt{N_0}}\sum_{k=1}^{N_0}\theta^{h_0}_k+\frac1{\sqrt{N_0}}\sum_{k=1}^{N_0}\zeta^{h_0}_k+\frac1{\sqrt{N_0}}\sum_{k=1}^{N_0}\eta^{h_0}_k.
\end{equation*}
From \eqref{eq:E[sum(thetahk)]<} and \eqref{eq:E[sum(zetahk)]<}, we get
\begin{gather*}
\frac1{\sqrt{N_0}}\sum_{k=1}^{N_0}\theta^{h_0}_k
\stackrel{L^2(\mathbb{P})}{\longrightarrow}0
\quad\mbox{ as }\quad
N_0\uparrow\infty,\\
\frac1{\sqrt{N_0}}\sum_{k=1}^{N_0}\zeta^{h_0}_k
\stackrel{L^1(\mathbb{P})}{\longrightarrow}0
\quad\mbox{ as }\quad
N_0\uparrow\infty.
\end{gather*}
Since $\{\eta^{h_0}_k,k\geq1\}$ is i.i.d.~such that $\mathbb{E}[|\eta^{h_0}_1|^2]<\infty$ by \eqref{eq:etahk} and the fact that, by assumption, $\mathbb{E}[|\varphi(Y,Z)|^{2+\delta}]<\infty$, by the classical CLT,
\begin{equation*}
\frac1{\sqrt{N_0}}\sum_{k=1}^{N_0}\eta^{h_0}_k
\stackrel[]{\mathcal{L}}{\longrightarrow}
\mathcal{N}\Big(0,\frac{\Var\big((X_{h_0}-\xi^{h_0}_\star)^+\big)}{(1-\alpha)^2}\Big)
\quad\mbox{ as }\quad
N_0\uparrow\infty.
\end{equation*}
Combining the previous results,
\begin{equation*}
h_L^{-\frac98}\big(\chi^{h_0}_{N_0}-\chi^{h_0}_\star\big)
\stackrel[]{\mathcal{L}}{\longrightarrow}
\mathcal{N}\Big(0,h_0^{-\frac34}\big(1-M^{-\frac14}\big)\frac{\Var\big((X_{h_0}-\xi^{h_0}_\star)^+\big)}{(1-\alpha)^2}\Big)
\quad\mbox{ as }\quad
L\uparrow\infty.
\end{equation*}

\noindent
\emph{\textbf{Step~9. Study of \[\bigg\{\bigg(h_L^{-1}\sum_{\ell=1}^L\frac1{N_\ell}\sum_{k=1}^{N_\ell}e^{h_\ell}_k-e^{h_{\ell-1}}_k, h_L^{-\frac98}\sum_{\ell=1}^L\frac1{N_\ell}\sum_{k=1}^{N_\ell}\eta^{h_\ell}_k-\eta^{h_{\ell-1}}_k\bigg),L\geq1\bigg\}.\]}}
Since $N_1\geq\dots\geq N_L$ by \eqref{eq:N_ell:avg},
\begin{equation*}
\sum_{\ell=1}^L\frac1{N_\ell}\sum_{k=1}^{N_\ell}e^{h_\ell}_k-e^{h_{\ell-1}}_k
=\sum_{k=1}^{N_1}\sum_{\ell=1}^L\frac{e^{h_\ell}_k-e^{h_{\ell-1}}_k}{N_\ell}\mathds{1}_{1\leq k\leq N_\ell}.
\end{equation*}
We apply the CLT~\cite[Corollary~3.1]{nla.cat-vn2887492} to the martingale arrays \begin{equation*}
\bigg\{\bigg(h_L^{-1}\sum_{k=1}^{N_1}\sum_{\ell=1}^L\frac{e^{h_\ell}_k-e^{h_{\ell-1}}_k}{N_\ell}\mathds{1}_{1\leq k\leq N_\ell}, h_L^{-\frac98}\sum_{k=1}^{N_1}\sum_{\ell=1}^L\frac{\eta^{h_\ell}_k-\eta^{h_{\ell-1}}_k}{N_\ell}\mathds{1}_{1\leq k\leq N_\ell}\bigg),L\geq1\bigg\}.
\end{equation*}

\noindent
\emph{\textbf{Step~9.1. Verification of the conditional Lindeberg condition.}}
\newline

\noindent
\emph{\textbf{Step~9.1.1. Study of $\Big\{h_L^{-1}\sum_{k=1}^{N_1}\sum_{\ell=1}^L\frac{e^{h_\ell}_k-e^{h_{\ell-1}}_k}{N_\ell}\mathds{1}_{1\leq k\leq N_\ell},L\geq1\Big\}$.}}
\newline\newline
Since, for all $k\geq1$, $\big\{\frac1{N_\ell}(e^{h_\ell}_k-e^{h_{\ell-1}}_k)\mathds{1}_{1\leq k\leq N_\ell},1\leq\ell\leq L\big\}$ are independent and centered, using the Marcinkiewicz-Zygmund inequality, Jensen's inequality and \eqref{eq:E[|e|]<},
\begin{equation*}
\begin{aligned}
\sum_{k=1}^{N_1}\mathbb{E}\bigg[\bigg|\sum_{\ell=1}^L&h_L^{-1}\frac{e^{h_\ell}_k-e^{h_{\ell-1}}_k}{N_\ell}\mathds{1}_{1\leq k\leq N_\ell}\bigg|^{2+\delta}\bigg]\\
&\leq Ch_L^{-(2+\delta)}\sum_{k=1}^{N_1}\mathbb{E}\bigg[\bigg|\sum_{\ell=1}^L\frac{(e^{h_\ell}_k-e^{h_{\ell-1}}_k)^2}{N_\ell^2}\mathds{1}_{1\leq k\leq N_\ell}\bigg|^{1+\frac\delta2}\bigg]\\
&\leq Ch_L^{-(2+\delta)}\sum_{k=1}^{N_1}L^\frac\delta2\sum_{\ell=1}^L\frac{\mathbb{E}\big[\big|e^{h_\ell}_k-e^{h_{\ell-1}}_k\big|^{2+\delta}\big]}{N_\ell^{2+\delta}}\mathds{1}_{1\leq k\leq N_\ell}\\
&\leq Ch_L^{-(2+\delta)}L^\frac\delta2\sum_{\ell=1}^L\frac{h_\ell^{\frac12}}{N_\ell^{1+\delta}}.
\end{aligned}
\end{equation*}
Thus, by \eqref{eq:N_ell:avg},
\begin{equation*}
\begin{aligned}
\sum_{k=0}^{N_1}\mathbb{E}\bigg[\bigg|\sum_{\ell=1}^L&h_L^{-1}\frac{e^{h_\ell}_k-e^{h_{\ell-1}}_k}{N_\ell}\mathds{1}_{1\leq k\leq N_\ell}\bigg|^{2+\delta}\bigg]\\
&\leq Ch_L^{-(2+\delta)}L^\frac\delta2\sum_{\ell=1}^L\frac{h_\ell^\frac12}{N_\ell^{1+\delta}}\\
&\leq Ch_L^{\frac{1+5\delta}4}L^\frac\delta2\sum_{\ell=1}^Lh_\ell^{\frac{1-3\delta}4}\\
&\leq C
\begin{cases}
h_L^{\frac{1+5\delta}4}L^\frac\delta2,
&\delta<\frac13,\\
h_L^{\frac{1+5\delta}4}L^{1+\frac\delta2},
&\delta=\frac13,\\
h_L^{\frac12(1+\delta)}L^\frac\delta2,
&\delta>\frac13.
\end{cases}
\end{aligned}
\end{equation*}
Hence
\begin{equation*}
\limsup_{L\uparrow\infty}\sum_{k=0}^{N_1}\mathbb{E}\bigg[\bigg|\sum_{\ell=1}^Lh_L^{-1}\frac{e^{h_\ell}_k-e^{h_{\ell-1}}_k}{N_\ell}\mathds{1}_{1\leq k\leq N_\ell}\bigg|^{2+\delta}\bigg]=0,
\end{equation*}

\noindent
\emph{\textbf{Step~9.1.2. Study of $\Big\{h_L^{-\frac98}\sum_{k=1}^{N_1}\sum_{\ell=1}^L\frac{\eta^{h_\ell}_k-\eta^{h_{\ell-1}}_k}{N_\ell}\mathds{1}_{1\leq k\leq N_\ell},L\geq1\Big\}$.}}
\newline\newline
Proceeding as in Step~9.1.2. in Theorem~\ref{thm:ml:clt}'s proof, Appendix~\ref{prf:ml:clt}, using in particular \eqref{eq:eta^2+d<}, that $\sup_{h \in \mathcal{H}}\mathbb{E}[|X_h|^{2+\delta}]< \infty$ by assumption and \eqref{eq:N_ell:avg},
\begin{equation*}
\begin{aligned}
\sum_{k=1}^{N_1}\mathbb{E}\bigg[\bigg|\sum_{\ell=1}^L&h_L^{-\frac98}\frac{\eta^{h_\ell}_k-\eta^{h_{\ell-1}}_k}{N_\ell}\mathds{1}_{1\leq k\leq N_\ell}\bigg|^{2+\delta}\bigg]\\
&\leq Ch_L^{-\frac98(2+\delta)}L^\frac\delta2\sum_{\ell=1}^L\frac1{N_\ell^{2+\delta}}\sum_{k=1}^{N_\ell}\mathbb{E}\big[\big|\eta^{h_\ell}_k-\eta^{h_{\ell-1}}_k\big|^{2+\delta}\big]\\
&\leq Ch_L^{-\frac98(2+\delta)}L^\frac\delta2\sum_{\ell=1}^L\frac{h_\ell^{1+\frac\delta2}}{N_\ell^{1+\delta}}\\
&\leq C
\begin{cases}
h_L^{\frac98\delta}L^\frac\delta2,
&\delta<1,\\
h_L^{\frac98\delta}L^{1+\frac\delta2},
&\delta=1,\\
h_L^{\frac14+\frac78\delta}L^\frac\delta2,
&\delta>1.
\end{cases}
\end{aligned}
\end{equation*}
Therefore,
\begin{equation*}
\limsup_{L\uparrow\infty}\sum_{k=1}^{N_1}\mathbb{E}\bigg[\bigg|\sum_{\ell=1}^Lh_L^{-\frac98}\frac{\eta^{h_\ell}_k-\eta^{h_{\ell-1}}_k}{N_\ell}\mathds{1}_{1\leq k\leq N_\ell}\bigg|^{2+\delta}\bigg]=0.
\end{equation*}

The conditional Lindeberg condition is thus satisfied.
\\

\noindent
\emph{\textbf{Step~9.2. Convergence of the conditional covariance matrices.}}
\newline
It remains to investigate the asymptotics of the conditional covariance matrices $\big\{S_L=(S_L^{i,j})_{1\leq i,j\leq 2}, L\geq 1\big\}$ defined by
\begin{align*}
S^{1,1}_L
& \coloneqq \sum_{k=1}^{N_1}\mathbb{E}\bigg[\bigg(\sum_{\ell=1}^Lh_L^{-1}\frac{e^{h_\ell}_k-e^{h_{\ell-1}}_k}{N_\ell}\mathds{1}_{1\leq k\leq N_\ell}\bigg)^2\bigg|\mathcal{F}^{h_L}_{k-1}\bigg],\\
S^{2,2}_L & \coloneqq \sum_{k=1}^{N_1}\mathbb{E}\bigg[\bigg(\sum_{\ell=1}^Lh_L^{-\frac98}\frac{\eta^{h_\ell}_k-\eta^{h_{\ell-1}}_k}{N_\ell}\mathds{1}_{1\leq k\leq N_\ell}\bigg)^2\bigg|\mathcal{F}^{h_L}_{k-1}\bigg],\\
S^{1,2}_L = S^{2,1}_L & \coloneqq \sum_{k=1}^{N_1} \mathbb{E}\bigg[\sum_{\ell=1}^L \frac{h_L^{-\frac{17}8}}{N_\ell^2}(e^{h_\ell}_k-e^{h_{\ell-1}}_k)(\eta^{h_\ell}_k-\eta^{h_{\ell-1}}_k)\mathds{1}_{1\leq k\leq N_\ell}\bigg|\mathcal{F}^{h_L}_{k-1}\bigg].
\end{align*}

\noindent
\emph{\textbf{Step~9.2.1. Convergence of $\{S^{1,1}_L,L\geq1\}$.}}
\newline
We write
\begin{equation*}
S^{1,1}_L
=\sum_{k=1}^{N_1}\mathbb{E}\bigg[\bigg(\sum_{\ell=1}^Lh_L^{-1}\frac{e^{h_\ell}_k-e^{h_{\ell-1}}_k}{N_\ell}\mathds{1}_{1\leq k\leq N_\ell}\bigg)^2\bigg]
= \sum_{\ell=1}^L \overline{U}_\ell,
\end{equation*}
where
\begin{equation*}
\overline{U}_\ell
\coloneqq\frac{h_L^{-2}}{N_\ell^2}\sum_{k=1}^{N_\ell}\mathbb{E}[(e^{h_\ell}_k-e^{h_{\ell-1}}_k)^2]\\
=\frac{h_L^{-2}}{N_\ell}\mathbb{E}[(e^{h_\ell}_1-e^{h_{\ell-1}}_1)^2].
\end{equation*}
By \eqref{eq:N_ell:avg}, \eqref{eq:limE[e]} and Ces\`aro's lemma (`$./\infty$' version),
\begin{equation*}
\begin{aligned}
\lim_{L\uparrow\infty}
S^{1,1}_L
&=\lim_{L\uparrow\infty}h_L^{-2}\bigg(\sum_{\ell=1}^L\frac{h_\ell^\frac12}{N_\ell}\bigg)
\times\lim_{L\uparrow\infty}\bigg(\sum_{\ell=1}^L\frac{h_\ell^\frac12}{N_\ell}\bigg)^{-1}
\bigg(\sum_{\ell=1}^L\frac{h_\ell^\frac12}{N_\ell}h_\ell^{-\frac12}\mathbb{E}[(e^{h_\ell}_1-e^{h_{\ell-1}}_1)^2]\bigg)\\
&=\frac1{\big(1-M^{-\frac14}\big)}
\times\lim_{L\uparrow\infty}\bigg(\sum_{\ell=1}^Lh_\ell^{-\frac14}\bigg)^{-1}
\bigg(\sum_{\ell=1}^Lh_\ell^{-\frac14}h_\ell^{-\frac12}\mathbb{E}[(e^{h_\ell}_1-e^{h_{\ell-1}}_1)^2]\bigg)\\
&=\frac{\mathbb{E}[|G|f_G(\xi^0_\star)]}{(1-\alpha)^2\big(1-M^{-\frac14}\big)}.
\end{aligned}
\end{equation*}

\noindent
\emph{\textbf{Step~9.2.2. Convergence of $\{S^{2,2}_L,L\geq1\}$.}}
\newline
We have
\begin{equation*}
S^{2,2}_L
=\sum_{k=1}^{N_1}\mathbb{E}\bigg[\bigg(\sum_{\ell=1}^Lh_L^{-\frac98}\frac{\eta^{h_\ell}_k-\eta^{h_{\ell-1}}_k}{N_\ell}\mathds{1}_{1\leq k\leq N_\ell}\bigg)^2\bigg]
= \sum_{\ell=1}^L\overline{W}_\ell,
\end{equation*}
with
\begin{equation*}
\begin{aligned}
\overline W_\ell
&\coloneqq \frac{h_L^{-\frac94}}{N_\ell} \mathbb{E}[(\eta^{h_\ell}_1-\eta^{h_{\ell-1}}_1)^2]\\
&=\frac{h_L^{-\frac94}}{(1-\alpha)^2}\frac{h_\ell}{N_\ell}h_\ell^{-1}\Var\big((X_{h_\ell}-\xi^{h_\ell}_\star)^+-(X_{h_{\ell-1}}-\xi^{h_{\ell-1}}_\star)^+\big).
\end{aligned}
\end{equation*}
Via \eqref{eq:N_ell:avg} , \eqref{eq:limVarX+} and Ces{\`a}ro's lemma (`$0/0$' version),
\begin{equation*}
\begin{aligned}
\lim_{L\uparrow\infty}S^{2,2}_L
&=\lim_{L\uparrow\infty}\frac{h_L^{-\frac94}}{(1-\alpha)^2}\bigg(\sum_{\ell=1}^L\frac{h_\ell}{N_\ell}\bigg)\\
&\quad\times\lim_{L\uparrow\infty}\bigg(\sum_{\ell=1}^L\frac{h_\ell}{N_\ell}\bigg)^{-1}
\bigg(\sum_{\ell=1}^L\frac{h_\ell}{N_\ell}h_\ell^{-1}\Var\big((X_{h_\ell}-\xi^{h_\ell}_\star)^+-(X_{h_{\ell-1}}-\xi^{h_{\ell-1}}_\star)^+\big)\bigg)\\
&=\frac{h_0^\frac14}{(1-\alpha)^2M^\frac14}\\
&\quad\times\lim_{L\uparrow\infty}\bigg(\sum_{\ell=1}^Lh_\ell^\frac14\bigg)^{-1}
\bigg(\sum_{\ell=1}^Lh_\ell^\frac14h_\ell^{-1}\Var\big((X_{h_\ell}-\xi^{h_\ell}_\star)^+-(X_{h_{\ell-1}}-\xi^{h_{\ell-1}}_\star)^+\big)\bigg)\\
&=\frac{h_0^\frac14}{(1-\alpha)^2M^\frac14}\Var(\mathds{1}_{X_0>\xi^0_\star}\,G).
\end{aligned}
\end{equation*}

\noindent
\emph{\textbf{Step~9.2.3. Convergence of $\{S^{1,2}_L,L\geq1\}$.}}
\newline
Using \eqref{eq:E[|e.eta|]<} and \eqref{eq:N_ell:avg},
\begin{equation*}
\begin{aligned}
\big|S^{1,2}_L\big|
&\leq h_L^{-\frac{17}8}\sum_{\ell=1}^L\frac1{N_\ell^2}\sum_{k=1}^{N_\ell}\mathbb{E}\big[\big|(e^{h_\ell}_k-e^{h_{\ell-1}}_k)(\eta^{h_\ell}_k-\eta^{h_{\ell-1}}_k)\big|\big]\\
&\leq Ch_L^{-\frac{17}8}\sum_{\ell=1}^L\frac{h_\ell^\frac34}{N_\ell}\\
&\leq Ch_L^{\frac18}L,
\end{aligned}
\end{equation*}
Hence
\begin{equation*}
S^{1,2}_L=S^{2,1}_L
\stackrel{\Pas}{\longrightarrow}0
\quad\mbox{ as }\quad
L\uparrow\infty.
\end{equation*}

Invoking the independence of the sequences studied in Steps~8 and~9 completes the proof.

\section{Square Convergence of the Averaged MLSA Algorithm}

To fully apprehend its properties, we hereby study the $L^2(\mathbb{P})$-error and complexity of the VaR AMLSA algorithm \eqref{eq:xi:ML:avg}.

\subsection{Convergence Rate Analysis}

\begin{lemma}[{\cite[Theorem~2.7(ii) \& Proposition~3.2(ii)]{mlsa}}]\label{lmm:recall}\

\begin{enumerate}[(i)]
    \item\label{lmm:variance-cv}
Suppose that $\varphi(Y,Z)\in L^2(\mathbb{P})$, that Assumption~\ref{asp:misc} holds and that
\begin{equation*}
\sup_{h\in\mathcal{H}}{\mathbb{E}\Big[|\xi^h_0|^4\exp\Big(\frac{16}{1-\alpha}c_\alpha\sup_{h\in\mathcal{H}}\|f_{X_h}\|_\infty|\xi_0^h|\Big)\Big]}<\infty,
\end{equation*}
where $c_\alpha=1\vee\frac\alpha{1-\alpha}$. Set
\begin{equation*}
\lambda_2 = \inf_{h \in \mathcal{H}\cup\{0\}}\frac34V_h''(\xi^h_\star)\wedge\|V_h''\|_\infty\frac{V_h''(\xi^h_\star)^4}{[V_h'']_\mathrm{Lip}^2}>0.
\end{equation*}
If $\gamma_n=\gamma_1n^{-\beta}$, $n\geq1$, $\gamma_1>0$, $\beta\in(0,1]$,
with $\lambda_2\gamma_1>2$ if $\beta=1$, then there exists a positive constant $C<\infty$ such that, for any positive integer $n$,
\begin{equation*}
\sup_{h\in\mathcal{H}}\mathbb{E}[(\xi^h_n-\xi^h_\star)^4]\leq C\gamma_n^2.
\end{equation*}

\item\label{lmm:lipschitz}
Let $\big\{G_\ell\coloneqq h_\ell^{-\frac12}(X_{h_\ell}-X_{h_{\ell-1}}),\ell\geq1\big\}$, define $\{F_{X_{h_{\ell-1}}\,|\,G_\ell=g}\colon x\mapsto\mathbb{P}(X_{h_{\ell-1}}\leq x\,|\,G_\ell=g),g\in\supp(\mathbb{P}_{G_\ell}),\ell\geq1\}$ and consider the random variables $\{K_\ell\coloneqq K_\ell(G_\ell),\ell\geq1\}$, where
\begin{equation*}
K_\ell(g)\coloneqq\sup_{x\neq y}\,\frac{\big|F_{X_{h_{\ell-1}}\,|\,G_\ell=g}(x)-F_{X_{h_{\ell-1}}\,|\,G_\ell=g}(y)\big|}{|x-y|},
\quad\ell\geq1,g\in\supp(\mathbb{P}_{G_\ell}).
\end{equation*}
Assume that $\{K_\ell,\ell\geq1\}$ satisfy
\begin{equation*}
\sup_{\ell\geq1}\,\mathbb{E}[|G_\ell|K_\ell]<\infty.
\end{equation*}
Then
\begin{equation*}
\sup_{\ell\geq1,\xi\in\mathbb{R}}\,
h_\ell^{-\frac12}\mathbb{E}\big[\big|\mathbf1_{X_{h_\ell}>\xi}-\mathbf1_{X_{h_{\ell-1}}>\xi}\big|\big]
<\infty.
\end{equation*}

\end{enumerate}
\end{lemma}

\begin{theorem}
Assume that the frameworks of Lemmas~\ref{lmm:recall}(\ref{lmm:variance-cv}) and~\ref{lmm:recall}(\ref{lmm:lipschitz}) hold and that Assumption~\ref{asp:fh-f0} is satisfied.
If $\gamma_n = \gamma_1 n^{-\beta}$, $n\geq1$, $\gamma>0$, $\beta\in \big(\frac12,1\big)$, then there exists a positive constant $C<\infty$ such that, for any positive integer $L$ and any positive integer sequence $\mathbf{N}=\{N_\ell,0\leq\ell\leq L\}$, 
\begin{equation}
\label{eq:xi:bar:L2}
\begin{aligned}
\mathbb{E}\big[\big(\overline{\xi}^\text{\tiny\rm ML}_\mathbf{N}&-\xi^{h_L}_\star\big)^2\big]
\leq C \bigg(
\gamma_{N_0}
+\bigg(\sum_{\ell=1}^L\frac1{N_\ell\sqrt{\gamma_{N_\ell}}}\bigg)^2\\
&+\bigg(\sum_{\ell=1}^Lh_\ell^\frac{1+4\delta_0\wedge3}4\sqrt{\gamma_{N_\ell}}\bigg)^2
+\bigg(\sum_{\ell=1}^L\gamma_{N_\ell}\bigg)^2
+\sum_{\ell=1}^L\frac{\sqrt{\gamma_{N_\ell}}}{N_\ell}
+\sum_{\ell=1}^L\frac{h_\ell^\frac12}{N_\ell} \bigg).
\end{aligned}
\end{equation}
\end{theorem}

\begin{proof}
In the developments below, we denote by $C<\infty$ a positive constant whose value may vary from line to line but does not depend upon $L$.
\\

We come back to the decomposition \eqref{eq:amlsa:split} and analyze each term separately.
\\

\noindent
\emph{\textbf{Step~1. Study of $\big\{\overline{\xi}^{h_0}_{N_0}-\xi^{h_0}_\star,L\geq1\big\}$.}}
\newline\newline
By Lemma~\ref{lmm:error}(\ref{lmm:error:statistical}) and a comparison between series and integrals, recalling that $\beta<1$,
\begin{equation*}
\mathbb{E}\big[\big(\overline{\xi}^{h_0}_{N_0}-\xi^{h_0}_\star\big)^2\big]
\leq\frac1{N_0}\sum_{k=1}^{N_0}\mathbb{E}[(\xi^{h_0}_k-\xi^{h_0}_\star)^2]
\leq \frac{C}{N_0}\sum_{k=1}^{N_0}\gamma_k
\leq C\gamma_{N_0}.
\end{equation*}

\noindent
\emph{\textbf{Step~2. Study of $\Big\{\sum_{\ell=1}^L\frac1{N_\ell}\sum_{k=1}^{N_\ell}a^{h_\ell}_k-a^{h_{\ell-1}}_k,L\geq1\Big\}$.}}
\newline\newline
We have
\begin{equation*}
\mathbb{E}\bigg[\bigg(\sum_{\ell=1}^L\frac1{N_\ell}\sum_{k=1}^{N_\ell}a^{h_\ell}_k-a^{h_{\ell-1}}_k\bigg)^2\bigg]^\frac12
\leq\sum_{\ell=1}^L\mathbb{E}\bigg[\bigg(\frac1{N_\ell}\sum_{k=1}^{N_\ell}a^{h_\ell}_k\bigg)^2\bigg]^\frac12+\mathbb{E}\bigg[\bigg(\frac1{N_\ell}\sum_{k=1}^{N_\ell}a^{h_{\ell-1}}_k\bigg)^2\bigg]^\frac12.
\end{equation*}
From \eqref{eq:ahn}, \eqref{eq:abel}--\eqref{eq:abel:2} and Lemma~\ref{lmm:error}(\ref{lmm:error:statistical}), we get
\begin{equation*}
\begin{aligned}
\mathbb{E}\bigg[\bigg(\frac1n\sum_{k=1}^na^h_k\bigg)^2\bigg]^\frac12
&\leq\frac1n\Big(\frac1{\gamma_n}\mathbb{E}[(\xi^h_n-\xi^h_\star)^2]^\frac12+\frac1{\gamma_1}\mathbb{E}[(\xi^h_0-\xi^h_\star)^2]^\frac12\Big)\\
&\quad+\frac1n\sum_{k=1}^{n-1}\Big(\frac1{\gamma_{k+1}}-\frac1{\gamma_k}\Big)\mathbb{E}[(\xi^h_k-\xi^h_\star)^2]^\frac12\\
&\leq C\bigg(\frac1n\Big(\frac1{\sqrt{\gamma_n}}+1\Big)+\frac1n\sum_{k=1}^{n-1}\Big(\frac1{\gamma_{k+1}}-\frac1{\gamma_k}\Big)\sqrt{\gamma_k}\bigg)\\
&\leq \frac{C}{n\sqrt{\gamma_n}}.
\end{aligned}
\end{equation*}
Therefore,
\begin{equation*}
\mathbb{E}\bigg[\bigg(\sum_{\ell=1}^L\frac1{N_\ell}\sum_{k=1}^{N\ell}a^{h_\ell}_k-a^{h_{\ell-1}}_k\bigg)^2\bigg]
\leq C\bigg(\sum_{\ell=1}^L\frac1{N_\ell\sqrt{\gamma_{N_\ell}}}\bigg)^2.
\end{equation*}

\noindent
\emph{\textbf{Step~3. Study of $\Big\{\sum_{\ell=1}^L\frac1{N_\ell}\sum_{k=1}^{N_\ell}g^{h_\ell}_k-g^{h_{\ell-1}}_k,L\geq1\Big\}$.}}
\newline\newline
We have
\begin{equation*}
\mathbb{E}\bigg[\bigg(\sum_{\ell=1}^L\frac1{N_\ell}\sum_{k=1}^{N_\ell}g^{h_\ell}_k-g^{h_{\ell-1}}_k\bigg)^2\bigg]^\frac12
\leq\sum_{\ell=1}^L\frac1{N_\ell}\sum_{k=1}^{N_\ell}\mathbb{E}[(g^{h_\ell}_k)^2]^\frac12+\mathbb{E}[(g^{h_{\ell-1}}_k)^2]^\frac12.
\end{equation*}
By Assumption~\ref{asp:fh-f0}, the uniform Lipschitz regularity of $\{V_{h_\ell},\ell\geq1\}$ by Assumption~\ref{asp:misc}(\ref{asp:misc:iv}) and Lemma~\ref{lmm:error}(\ref{lmm:error:weak}),
\begin{equation}\label{useful}
|V_0''(\xi^0_\star)-V_{h_\ell}''(\xi^{h_\ell}_\star)| \leq C\big(h_\ell^{(\frac14+\delta_0)\wedge1} + h_\ell\big).
\end{equation}
Using \eqref{eq:ghn}, \eqref{useful} and Lemma~\ref{lmm:error}(\ref{lmm:error:statistical}),
\begin{equation*}
\mathbb{E}[(g^{h_\ell}_k)^2]^\frac12
\leq|V_0''(\xi^0_\star)-V_{h_\ell}''(\xi^{h_\ell}_\star)|\,\mathbb{E}[(\xi^{h_\ell}_{k-1}-\xi^{h_\ell}_\star)^2]^\frac12
\leq Ch_\ell^{(\frac14+\delta_0)\wedge1}\gamma_k^\frac12,
\quad k\geq1.
\end{equation*}
A comparison between series and integrals thus yields
\begin{equation*}
\mathbb{E}\bigg[\bigg(\sum_{\ell=1}^L\frac1{N_\ell}\sum_{k=1}^{N_\ell}g^{h_\ell}_k-g^{h_{\ell-1}}_k\bigg)^2\bigg]\\
\leq C\bigg(\sum_{\ell=1}^Lh_\ell^{\frac14+(\delta_0\wedge\frac34)}\gamma_{N_\ell}^\frac12\bigg)^2.
\end{equation*}

\noindent
\emph{\textbf{Step~4. Study of $\Big\{\sum_{\ell=1}^L\frac1{N_\ell}\sum_{k=1}^{N_\ell}r^{h_\ell}_k-r^{h_{\ell-1}}_k,L\geq1\Big\}$.}}
\newline\newline
Note that
\begin{equation*}
\mathbb{E}\bigg[\bigg(\sum_{\ell=1}^L\frac1{N_\ell}\sum_{k=1}^{N_\ell}r^{h_\ell}_k-r^{h_{\ell-1}}_k\bigg)^2\bigg]^\frac12
\leq\sum_{\ell=1}^L\frac1{N_\ell}\sum_{k=1}^{N_\ell}\mathbb{E}[(r^{h_\ell}_k)^2]^\frac12+\mathbb{E}[(r^{h_{\ell-1}}_k)^2]^\frac12.
\end{equation*}
From \eqref{eq:rhn}, \eqref{eq:Vh':taylor:1}, the uniform Lipschitz regularity of $\{V_{h_\ell}'',\ell\geq1\}$, stemming from Assumption~\ref{asp:misc}(\ref{asp:misc:iv}), and Lemma~\ref{lmm:recall}(\ref{lmm:variance-cv}), we get
\begin{equation*}
\mathbb{E}[|r^h_k|^2]
\leq C\,\mathbb{E}[(\xi^h_{k-1}-\xi^h_\star)^4]
\leq C\gamma_k^2,
\quad k\geq1.
\end{equation*}
Combining the two previous inequalities and using a comparison between series and integrals, we conclude
\begin{equation*}
\mathbb{E}\bigg[\bigg(\sum_{\ell=1}^L\frac1{N_\ell}\sum_{k=1}^{N_\ell}r^{h_\ell}_k-r^{h_{\ell-1}}_k\bigg)^2\bigg]
\leq C\bigg(\sum_{\ell=1}^L\gamma_{N_\ell}\bigg)^2.
\end{equation*}

\noindent
\emph{\textbf{Step~5. Study of $\Big\{\sum_{\ell=1}^L\frac1{N_\ell}\sum_{k=1}^{N_\ell}\rho^{h_\ell}_k-\rho^{h_{\ell-1}}_k,L\geq1\Big\}$.}}
\newline\newline
Recall that $\big\{\frac1{N_\ell}\sum_{k=1}^{N_\ell}(\rho^{h_\ell}_k-\rho^{h_{\ell-1}}_k),\ell\geq1\big\}$ are independent and centered and that, at each level $\ell\geq1$, $\{\rho^{h_\ell}_k-\rho^{h_{\ell-1}}_k,k\geq1\}$ are $\{\mathcal{F}^{h_\ell}_k,k\geq1\}$-martingale increments.
Using \eqref{eq:E[|rho|2]<} and a comparison between series and integrals,
\begin{equation*}
\mathbb{E}\bigg[\bigg(\sum_{\ell=1}^L\frac1{N_\ell}\sum_{k=1}^{N_\ell}\rho^{h_\ell}_k-\rho^{h_{\ell-1}}_k\bigg)^2\bigg]
\leq2\sum_{\ell=1}^L\frac1{N_\ell^2}\sum_{k=1}^{N_\ell}\mathbb{E}\big[\big|\rho^{h_\ell}_k\big|^2\big]+\mathbb{E}\big[\big|\rho^{h_{\ell-1}}_k\big|^2\big]
\leq C\sum_{\ell=1}^L\frac{\sqrt{\gamma_{N_\ell}}}{N_\ell}.
\end{equation*}

\noindent
\emph{\textbf{Step~6. Study of $\Big\{\sum_{\ell=1}^L\frac1{N_\ell}\sum_{k=1}^{N_\ell}e^{h_\ell}_k-e^{h_{\ell-1}}_k,L\geq1\Big\}$.}}
\newline\newline
Since $\big\{\frac1{N_\ell}\sum_{k=1}^{N_\ell}(e^{h_\ell}_k-e^{h_{\ell-1}}_k),\ell\geq1\big\}$ are independent and centered and, for any level $\ell\geq1$, $\big\{e^{h_\ell}_k-e^{h_{\ell-1}}_k,k\geq1\big\}$ are $\{\mathcal{F}^{h_\ell}_k,k\geq1\}$-martingale increments,
using \eqref{eq:E[e2]} and Lemma~\ref{lmm:recall}(\ref{lmm:lipschitz}),
\begin{equation*}
\mathbb{E}\bigg[\bigg(\sum_{\ell=1}^L\frac1{N_\ell}\sum_{k=1}^{N_\ell}e^{h_\ell}_k-e^{h_{\ell-1}}_k\bigg)^2\bigg]
=\sum_{\ell=1}^L\frac1{N_\ell^2}\sum_{k=1}^{N_\ell}\mathbb{E}[(e^{h_\ell}_k-e^{h_{\ell-1}}_k)^2]
\leq C\sum_{\ell=1}^L\frac{h_\ell^\frac12}{N_\ell}.
\end{equation*}
\end{proof}

\subsection{Complexity Analysis}
\label{apx:B2}

The global error of the VaR AMLSA scheme \eqref{eq:xi:ML:avg} decomposes into a statistical and a bias errors:
\begin{equation*}
\overline\xi^\text{\tiny\rm ML}_\mathbf{N}-\xi^0_\star=\big(\overline\xi^\text{\tiny\rm ML}_\mathbf{N}-\xi^{h_L}_\star\big)+(\xi^{h_L}_\star-\xi^0_\star).
\end{equation*}

Let $\varepsilon\in(0,1)$ be a prescribed accuracy.
Lemma~\ref{lmm:error}(\ref{lmm:error:weak}) guarantees the bias of estimating $\xi^0_\star$ by $\overline\xi^\text{\tiny\rm ML}_\mathbf{N}$ is of order $h_L$, hence, if $h_0>\varepsilon$, we have to choose the number of levels
\begin{equation*}
L=\bigg\lceil\frac{\ln{(h_0/\varepsilon)}}{\ln{M}}\bigg\rceil
\quad\mbox{ so that }\quad h_L=\frac{h_0}{M^L}\leq\varepsilon.
\end{equation*}

The cost of the VaR AMLSA scheme satisfies
\begin{equation*}
\Cost_\text{\tiny\rm AMLSA}\leq C\sum_{\ell=0}^L\frac{N_\ell}{h_\ell},
\end{equation*}
for some positive constant $C<\infty$ independent of $L$. To retrieve the optimal iterations amounts $\mathbf{N}=\{N_\ell,0\leq\ell\leq L\}$, we test out several leading term candidates for the upper bound \eqref{eq:xi:bar:L2}, of which we retain three compelling cases. We minimize the computational cost while constraining the chosen candidate to an order of $\varepsilon^2$.

We first consider the problem
\begin{equation*}
\begin{array}{lll}
\minimize_{N_0,\dots,N_L>0}
&\sum_{\ell=0}^LN_\ell h_\ell^{-1}\\
\subjectTo
&\sum_{\ell=0}^Lh_\ell^{(1+4\delta_0\wedge3)/4}\sqrt{\gamma_{N_\ell}}=C^{-1}\varepsilon,
\end{array}
\end{equation*}
yielding
\begin{equation*}
N_\ell=
\bigg\lceil C^\frac2\beta\gamma_1^\frac1\beta\varepsilon^{-\frac2\beta}
   \bigg(\sum_{\ell'=0}^Lh_{\ell'}^{-\frac{2\beta-1-4\delta_0\wedge3}{2(2+\beta)}}\bigg)^\frac2\beta h_\ell^\frac{5+4\delta_0\wedge3}{2(2+\beta)}\bigg\rceil,
   \quad 0\leq\ell\leq L.
\end{equation*}
Such a choice adds the constraints
\begin{equation*}
\frac12<\beta\leq\frac23
\quad\mbox{ and }\quad
\delta_0\geq\frac34
\end{equation*}
to guarantee that the upper estimate \eqref{eq:xi:bar:L2} is of order $\varepsilon^2$.
The corresponding complexity satisfies
\begin{equation*}
\Cost^\beta_\text{\rm\tiny AMLSA}
\leq C\varepsilon^{-\frac2\beta}
\stackrel[]{}{\to}\varepsilon^{-3}
\quad\mbox{ as }\quad
\beta\uparrow\frac23.
\end{equation*}

We next look at the problem
\begin{equation*}
\begin{array}{lll}
\minimize_{N_0,\dots,N_L>0}
&\sum_{\ell=0}^LN_\ell h_\ell^{-1}\\
\subjectTo
&\sum_{\ell=0}^LN_\ell^{-1}\gamma_{N_\ell}^{-\frac12}=C^{-1}\varepsilon.
\end{array}
\end{equation*}
The minimizers are given by
\begin{equation*}
N_\ell=
\bigg\lceil C^\frac2{2-\beta}\,\gamma_1^{-\frac1{2-\beta}}\varepsilon^{-\frac2{2-\beta}}
   \bigg(\sum_{\ell'=0}^Lh_{\ell'}^{-\frac{2-\beta}{4-\beta}}\bigg)^\frac2{2-\beta}h_\ell^\frac2{4-\beta}\bigg\rceil,
   \quad 0\leq\ell\leq L.
\end{equation*}
If the additional conditions
\begin{equation*}
\frac23\leq\beta<1
\quad\mbox{ and }\quad
\delta_0>\frac1{12}
\end{equation*}
are met, one attains the order $\varepsilon^2$ for the global $L^2(\mathbb{P})$-error \eqref{eq:xi:bar:L2}, with an optimal cost satisfying
\begin{equation*}
\Cost^\beta_\text{\rm\tiny AMLSA}
\leq C\varepsilon^{-\frac2{2-\beta}-1}
\stackrel[]{}{\to}\varepsilon^{-\frac52}
\quad\mbox{ as }\quad
\beta\downarrow\frac23.
\end{equation*}

The final case worth exploring is
\begin{equation*}
\begin{array}{lll}
\minimize_{N_0,\dots,N_L>0}
&\sum_{\ell=0}^LN_\ell h_\ell^{-1}\\
\subjectTo
&\sum_{\ell=0}^Lh_\ell^{\frac12}N_\ell^{-1}=C^{-1}\varepsilon^2.
\end{array}
\end{equation*}
This yields, under constraints stated in Theorem~\ref{thm:avg:ml:clt},
\begin{equation*}
N_\ell=
\bigg\lceil C^2\varepsilon^{-2}\bigg(\sum_{\ell'=0}^Lh_{\ell'}^{-\frac14}\bigg)h_\ell^\frac34\bigg\rceil,
   \quad 0\leq\ell\leq L.
\end{equation*}
The convergence analysis of this last heuristic is deferred to Theorem~\ref{thm:avg:ml:clt} and the ensuing complexity is discussed in Theorem~\ref{thm:cost:aml}.

\end{document}